\documentclass[ aps,pra, reprint,superscriptaddress, floatfix,  
nofootinbib,
nobibnotes,
amsmath,amssymb,
]{revtex4-2}

\usepackage{subfiles}

\usepackage{algorithm}
\usepackage{algpseudocode}
\usepackage[T1]{fontenc}
\usepackage{inputenc}
\usepackage{color}
\usepackage{babel}
\usepackage{physics}
\usepackage{mathtools}
\usepackage{bbding}
\usepackage{pifont}
\usepackage{textcomp}
\usepackage{wasysym}
\usepackage{amsthm}
\usepackage[normalem]{ulem}
\usepackage{nicefrac}
\usepackage{wasysym}
\usepackage{dsfont}
\usepackage{amsmath}
\usepackage{amssymb}
\usepackage{graphicx}
\usepackage{xcolor}
\usepackage{bbold}
\usepackage{notes2bib}
\bibnotesetup{
note-name = ,
use-sort-key = false
}
\usepackage[dvipsnames]{xcolor}

\usepackage{xcolor}
\usepackage{mdframed}

\newmdenv[
backgroundcolor=yellow!2,
linecolor=brown!75,
linewidth=1pt,
skipabove=10pt,
skipbelow=10pt
]{highlightbox}

\newmdenv[
  backgroundcolor=red!2,   
  linecolor=red!75,        
  linewidth=1pt,
  roundcorner=3pt,              
  skipabove=10pt,
  skipbelow=10pt,
  innerleftmargin=8pt,
  innerrightmargin=8pt
]{DefHighlightbox}

\newmdenv[
  backgroundcolor=blue!2, 
 linecolor=blue!50,         
  linewidth=1pt,
  skipabove=10pt,
  skipbelow=10pt,
  innerleftmargin=10pt,         
  innerrightmargin=10pt
]{ObsHighlightbox}

\newmdenv[
  backgroundcolor=cyan!2, 
 linecolor=cyan!75,         
  linewidth=1pt,
  skipabove=10pt,
  skipbelow=10pt,
  innerleftmargin=10pt,         
  innerrightmargin=10pt
]{PropHighlightbox}

\newmdenv[
  backgroundcolor=green!2,      
  linecolor=ForestGreen!75,        
  linewidth=1pt,
  skipabove=10pt,
  skipbelow=10pt
]{ThmHighlightbox}

\newmdenv[
  backgroundcolor=green!2,      
  linecolor=PineGreen!75,          
  topline=false,
  bottomline=false,
  rightline=false,
  linewidth=1pt,
  skipabove=10pt,
  skipbelow=10pt,
  innerleftmargin=10pt,
  innerrightmargin=10pt,
  innertopmargin=8pt,
  innerbottommargin=8pt
]{PropertyHighlightbox}

\newmdenv[
  backgroundcolor=orange!2,     
  linecolor=violet!75,     
  linewidth=1pt,
  roundcorner=3pt,              
  skipabove=10pt,
  skipbelow=10pt,
  innerleftmargin=8pt,
  innerrightmargin=8pt
]{CorHighlightbox}

\newcommand{\Ast}{\mathop{\scalebox{1.5}{\raisebox{-0.2ex}{$\ast$}}}}%

\usepackage{hyperref}

\usepackage{hyperref}

\hypersetup{
     colorlinks   = true,
     linkcolor    = blue,
     citecolor    = blue,
     urlcolor     = blue,
     }

\usepackage{appendix}
\renewcommand{\selectlanguage}[1]{}
\makeatletter
\@ifundefined{textcolor}{}
{%
	\definecolor{BLACK}{gray}{0}
	\definecolor{WHITE}{gray}{1}
	\definecolor{RED}{rgb}{1,0,0}
	\definecolor{GREEN}{rgb}{0,1,0}
	\definecolor{BLUE}{rgb}{0,0,1}
	\definecolor{CYAN}{cmyk}{1,0,0,0}
	\definecolor{MAGENTA}{cmyk}{0,1,0,0}
	\definecolor{YELLOW}{cmyk}{0,0,1,0}
}
\theoremstyle{plain}

\theoremstyle{plain}

\ifx\proof\undefined
\newenvironment{proof}[1][\protect\proofname]{\par
	\normalfont\topsep6\p@\@plus6\p@\relax
	\Trivlist
	\itemindent\parindent
	\item[\hskip\labelsep
	\scshape
	#1]\ignorespaces
}{%
	\endtrivlist\@endpefalse
}
\providecommand{\proofname}{\textbf{Proof}}
\fi
\theoremstyle{plain}

\providecommand{\lemmaname}{Lemma}
\providecommand{\definitionname}{Definition}
\providecommand{\propositionname}{Proposition}

\usepackage{babel}
\usepackage{txfonts}
\usepackage{colortbl}\definecolor{myurlcolor}{rgb}{0,0,0.7}
\newcommand{\id}{\mathcal{I}}
\newcommand{\Id}{\mathds{1}}

\newcommand{\St}{\mathsf{St}}
\newcommand{\Chan}{\mathsf{Chan}}
\newcommand{\jam}{\mathsf{Jam}}
\newcommand{\Choi}{\mathsf{Choi}}
\newcommand{\eff}{\mathsf{Eff}}

\renewcommand{\ket}[1]{\left| #1 \right\rangle}

\renewcommand{\ketbra}[2]{\left|#1\middle\rangle\!\middle\langle#2\right|}
\newcommand{\proj}[1]{\ketbra{#1}{#1}}

\DeclareMathOperator{\N}{\mathcal{N}}
\DeclareMathOperator{\C}{\mathcal{C}}
\DeclareMathOperator{\M}{\mathcal{M}}
\newcommand{\haH}

\newtheorem{theorem}{Theorem}

\newtheorem*{thm*}{Theorem}

\newtheorem{lemma}{Lemma}
\newtheorem*{lem*}{Lemma}

\newtheorem{observation}{Observation}
\newtheorem{corollary}{Corollary}
\newtheorem*{cor*}{Corollary}
\newtheorem{proposition}{Proposition}
\newtheorem{definition}{Definition}

\newtheorem{axm}{Axiom}
\def\H{\mathcal{H}}
\def\E{\mathcal{E}}

\makeatother

\begin{document}

\title{Multi-Time States in Quantum and Operational Probabilistic Theories: An Axiomatic Framework} 
\author{Kaumudibikash Goswami}
\email{goswami.kaumudibikash@gmail.com}
\affiliation{QICI Quantum Information and Computation Initiative, Department of Computer Science, The University of Hong Kong, Pokfulam Road, Hong Kong}
\affiliation{Hon Hai (Foxconn) Research Institute, Quantum Computing Research Center, Taipei 114, Taiwan }

\author{Archana Barvar}
\affiliation{Department of Physical Sciences, Indian Institute of Science Education and Research (IISER), Mohali, Punjab 140306, India}

\author{Manabendra Nath Bera}
\email{mnbera@gmail.com}
\affiliation{Department of Physical Sciences, Indian Institute of Science Education and Research (IISER), Mohali, Punjab 140306, India}

\begin{abstract}

Quantum states over time were introduced, in part, to extend Bayesian probability to quantum theory and have since found applications across diverse areas, ranging from quantum causality and open-system dynamics to tests of macroscopic realism. While axiomatic characterizations of these states exist for Markovian dynamics, a general framework accommodating non-Markovian dynamics is still lacking. Moreover, extending this framework beyond quantum mechanics to general physical theories also remains elusive. In this work, we develop an axiomatic framework for multi-time states within a class of operational probabilistic theories (OPTs) that includes quantum mechanics as a special case. Based on three assumptions---compositionality, reconstructibility, and spatio-temporal marginal consistency---the framework accommodates both Markovian and non-Markovian dynamics, and provides sufficient conditions to certify signaling and temporal memory within OPT. As a concrete application of our OPT framework, we fully characterize the unique multi-time state in boxworld and develop a method for Bayesian retrodiction of physical transformations there. Our general framework is consistent with quantum theory, reproducing several existing results. In the non-Markovian quantum dynamics, we establish a correspondence between multi-time states and the process-matrix framework, and fully characterize the multi-time states compatible with general dynamics. This correspondence also lets us infer quantum causal models from multi-time states, and as a proof of principle, we propose a semidefinite program to infer causal order from quantum multi-time states. These results provide a unified framework for studying spatio-temporal states, causal relations, and Bayesian reasoning within OPTs, with implications for both fundamental and applied research. 
\end{abstract}

\maketitle

\section{Introduction}
In classical theory, a spatio-temporal state is typically a joint
probability distribution over random variables, regardless of whether they describe two spacelike separated systems or one system at two timelike events. In this sense, classical theory may be viewed as causally neutral~\cite{leifer2013towards}. In quantum theory, however, the situation is fundamentally different. 
While conventional bipartite quantum states describe two spacelike
separated systems, a two-time quantum state that describes an initial state evolving through a quantum channel is far less straightforward, and remains an active area of research.

In recent years, several approaches have been proposed to describe quantum states associated with timelike events, commonly referred to as quantum states over time~\cite{ohya1983note, leifer2013towards, fitzsimons2015quantum, sutter2016strengthened, fullwood2022quantum}. These developments aim to generalize Bayesian methods in quantum
theory~\cite{parzygnat2023time}. Such Bayesian generalizations have motivated notions such as quantum conditional expectations~\cite{tsang2022generalized} and the quantum conditional density operator~\cite{cerf1997negative}, and found applications in thermodynamic fluctuation relations~\cite{Buscemi2021,Aw2021}, quantum state smoothing~\cite{Tsang2009,guevara2015smoothing,Liu2025_retrodictive,liu2026unifying},
and related problems~\cite{Sarkka2013-mn, Bishop2006-qp}.

Beyond Bayesian formulation, the proposed constructions of quantum states over time have found broad applications in several areas of quantum information and quantum foundations. For instance, the Petz recovery map associated with the approach in~\cite{leifer2013towards} appears naturally in contexts such as in formulation of observational entropy~\cite{Buscemi_2023, bai2024observational}, and quantum transport~\cite{hoogsteder2025approach}. Likewise, the formalism presented in~\cite{fitzsimons2015quantum} has been utilized to certify causal influence~\cite{fitzsimons2015quantum, song2024causal}, the investigation of Leggett-Garg macro-realism~\cite{comar2026using}. Additional applications and a comprehensive review of these constructions can be found in~\cite{parzygnat2023time}.

Given all these distinct constructions, the natural next step was to single out a unique quantum state over time via suitable axioms. Despite the initial no-go theorem~\cite{horsman2017can}, a refined axiomatic formulation made the program successful, yielding a unique quantum state over time compatible with multi-time Markovian quantum
dynamics~\cite{parzygnat2024virtual,lie2024quantum,lie2025multipartite}. 

In parallel to quantum states over time, several approaches based on quantum processes, or higher-order maps,  have also been developed to understand dynamical aspects of quantum theory. This framework adopts a process-centric approach, with transformation of channels as the primary object of interest. The most general processes, also called multi-time processes, are represented by process matrices~\cite{oreshkov12, Giarmatzi_2021}, process tensors~\cite{Pollock2018Operational, Pollock2018}, or quantum combs~\cite{chiribella09b}. These frameworks provide a complete operational description of the evolution of an open quantum system, including non-Markovian dynamics in which the environment retains memory of past interactions with the system. The formalism has enabled various applications, including quantum metrology~\cite{Altherr2021} and thermodynamic protocols~\cite{Zambon2025}. While the process framework also covers more exotic cases, such as indefinite causal order~\cite{oreshkov12,chiribella09} and indefinite input-output directions~\cite{chiribella2022quantum,mrini2024indefinitecausalstructurecausal}, here we focus on multi-time processes.

In contrast to multi-time processes, the quantum states-over-time framework adopts a state-centric approach: unlike multi-time processes, which split each time step into input and output systems to allow for interventions, this framework treats the system at each time step as a single object, with the goal of representing the multi-time
state as a generalization of a joint probability. This framework has proven promising: it provides a natural generalization of Bayesian probability and offers a rigorous formulation of Bayesian retrodiction of channels in quantum theory~\cite{parzygnat2023time, Parzygnat2023}. Yet, despite its promise, several open questions remain within this framework. In particular, two important gaps motivate the present work.

First, while the states-over-time framework has been axiomatically characterized to capture Markovian evolution~\cite{lie2024quantum,lie2025multipartite}, a similar axiomatic characterization of non-Markovian evolution within this framework remains elusive. This contrasts with process-based approaches, which offer a complete description of arbitrary open quantum dynamics.  Although specific classes of multi-time states, such as the pseudo-density matrix~\cite{fitzsimons2015quantum} and a time-symmetric variant based on pre- and post-selected ensembles~\cite{aharonov1964time,Aharonov1991Complete,Aharonov2009}, have been connected to the process matrix framework~\cite{zhang2020quantumcorrelationstime, Silva2017,Liu_2024}, an axiomatic correspondence between the two frameworks is still lacking. Such a characterization would be valuable, as it could provide a unified framework with a possibly unique structure, rather than relying on distinct formalisms with different correspondence rules to multi-time processes.

Second, so far, multi-time states and their associated Bayesian inference frameworks have been developed mainly for classical and quantum theories. More recently, these ideas have been investigated in finite-dimensional C*-algebras~\cite{parzygnat2024reversinginformationflowretrodiction} and, in infinite-dimensional settings, via von Neumann algebras~\cite{karmakar2026bayesianinferenceretrodictionfaithful}.  While these extensions are significant, they naturally raise an interesting question: can the notion of spatio-temporal states and Bayesian retrodiction be formulated in physical theories that go beyond the Hilbert space framework?  It is well known that every finite-dimensional C*-algebra is a direct sum of matrix algebras, effectively corresponding to a superselection structure within quantum theory~\cite{takesaki2001theory}. Similarly, von Neumann algebras are inherently tied to a Hilbert space framework via their spatial representations and operator topologies~\cite{takesaki2001theory, haag2012local}. However, there exist physical theories, such as boxworld~\cite{Popescu1994,Barrett2007}, that go beyond this standard complex Hilbert space structure and cannot be embedded into quantum theory~\cite{garner2023characterization, Muller2023}. Examples of states over time in these broader classes of physical theories are currently lacking. A systematic development of the states-over-time framework in this general setting is important, as it would reveal whether the structural features of multi-time states and Bayesian rules arise from theory-independent operational principles alone, and which specific features of a physical theory are required---or not required---to formulate them.

We address these questions within the broader class of theories described by operational probabilistic theories (OPTs)~\cite{chiribella2010probabilistic,Chiribella_2011,DAriano2016}, characterized by finite-dimensional convex state spaces, causality (i.e., a unique discarding operation for marginalization)~\cite{Chiribella_2011,chiribella2010probabilistic}, and local tomography~\cite{hardy2001quantum, hardy2011reformulating,chiribella2010probabilistic,Barrett2007,DAriano2016, Chiribella2021}, and formulate our results and proofs on multi-time states within this general OPT framework. This generalization is a central novelty of our approach: it demonstrates that the structural and characteristic features we identify for multi-time states are not specific to quantum theory, but instead follow solely from general operational principles. Our framework not only recovers and refines known results on quantum states over time, but also establishes new results for multi-time states associated with general quantum evolutions, including their ability to capture non-Markovian, causal, and Bayesian aspects. More importantly, the general OPT formulation provides a foundation for extending the framework of multi-time states beyond quantum theory, allowing us to characterize their structure and properties in genuinely non-quantum theories such as square-bit boxworld.

To start with, we first focus on the multi-time state within the OPT framework that arises from a Markovian evolution, which we call a Markov compatible multi-time state. In the same spirit of Refs.~\cite{lie2024quantum,lie2025multipartite}, we introduce two operationally motivated axioms to characterize the structure of such multi-time states. These include a notion of compositionality, which ensures that marginalizing over an intermediate time step yields the expected composition of the evolutions. We also introduce a reconstructibility axiom, inspired by the notion of a short Markov chain~\cite{petz1988sufficiency,hayden2004structure}, which captures the idea that the future state can be recovered from the present state alone via a suitable recovery map, thereby providing a natural criterion for compatibility with Markovian dynamics. It is worth pointing out that the reconstruction maps are not, in general, physical transformations within the theory; rather, they belong to the linear span of such transformations. From the Bayesian perspective, these maps provide a way to update our belief about the future state based solely on available present data, much like Bayesian retrodiction, which is understood as a form of belief propagation backward in time~\cite{Aw2021,Buscemi2021}. These axioms allow us to identify a structure analogous to one found in quantum theory~\cite{lie2025multipartite} for Markov compatible multi-time states within the OPT framework. In the language of the process-matrix formalism, these Markov compatible multi-time states are realized by inserting a broadcasting map into the intervention slots of a Markovian process, where, like the (virtual) quantum broadcasting maps~\cite{parzygnat2024virtual}, the broadcasting maps lie in the linear span of physical transformations.

We subsequently extend the notion of multi-time states to non-Markovian evolutions in OPTs. To characterize the structure of multi-time states in such a scenario, we introduce an additional axiom, spatio-temporal marginal consistency, which prescribes the behavior of multi-time states corresponding to initial product states evolving independently under product channels. Combining spatio-temporal marginal consistency with compositionality and reconstructibility, we show that the structure of general multi-time states arises from inserting broadcasting maps into a general process. This establishes an axiomatic relationship between process matrices, representing the evolution of open quantum systems (or their OPT generalizations \cite{Kissinger2019}), and multi-time states.

Our operational axioms, introduced to characterize multi-time states, lead to several interesting byproducts. First, we develop a complete framework of multi-time states in the boxworld theory of square bits. To characterize the virtual broadcasting map in this setting, we employ axioms analogous to those used for the virtual broadcasting map in quantum theory~\cite{parzygnat2024virtual} and find a unique such map in boxworld. With this unique broadcasting map, together with compositionality and reconstructibility, we provide the first unique multi-time state in boxworld. We then show how Bayesian retrodiction can be formulated in this theory. Subsequently, we find the unique broadcasting map for the composite case (2 inputs to 4 outputs). Unlike the quantum case, it factorizes as a product of individual broadcasting maps. With this composite broadcasting map, we show that spatio-temporal marginal consistency holds and that all structural results for the non-Markovian process in the broader OPT framework carry over directly to boxworld. With this example, we confirm that it is possible to formalize multi-time states and Bayes' rule in a broader class of physical theories.

Second, we generalize the result of Ref.~\cite{Liu2025} on quantum state over time, which introduced a sufficient condition for detecting signaling through a bipartite channel in terms of the negativity of a certain marginal of the multi-qubit pseudo-density matrix. We extend this criterion to any multi-time state that satisfies compositionality, reconstructibility, and spatio-temporal marginal consistency. This provides a certified test for signaling through a bipartite channel within the broader OPT framework, with the criteria in Ref.~\cite{Liu2025} being a special case. We subsequently apply this condition to derive a proof-of-principle sufficient criterion for identifying non-Markovian behavior in a multi-time process. Furthermore, we show that the correspondence between multi-time processes and states enables us to formulate a semidefinite programming (SDP) approach for inferring underlying causal models, for instance, determining whether a given multi-time state is compatible with a specific causal order. These results demonstrate the operational utility of our axioms, providing tools to detect information flow, memory effects, and causal structure, with potential applicability in quantum theory and beyond.

The article is organized as follows. In Section~\ref{sec:AxiomaticFramework}, we introduce the axiomatic framework for multi-time states within operational probabilistic theories (OPTs). In Section~\ref{sec:MarkovCaopatible}, we employ this framework to characterize Markov compatible multi-time states and establish a corresponding structure theorem. Section~\ref{sec:GenMTQSnonMarkov} extends the discussion to general evolutions, in which the system and environment evolve jointly and the environment may retain memory, leading to non-Markov compatible multi-time states and a corresponding structure theorem. Here, we further explore the connection between signaling and non-Markov compatibility and show that every multi-space state, of spacelike separated parties, can equivalently be represented by a multi-time state if we allow non-Markovian evolution. As the first application of our axiomatic framework, we develop the unique multi-time state in the square-bit boxworld, together with broadcasting maps, Bayesian retrodictions, and causal structures in evolutions, in Section~\ref{sec:MTSinBoxWorld}. Then, in Section~\ref{sec:MTSinQuantumTheory}, we apply the axiomatic framework to characterize the multi-time state in quantum theory, establish the correspondence with multi-time quantum processes, and refine some of the known results. We also demonstrate how, when restricted to quantum theory, the framework enables causal inference via semidefinite programming. Finally, we conclude in Section~\ref{sec:Conclusion}.

\section{An axiomatic framework for multi-time states in operational probabilistic theories \label{sec:AxiomaticFramework}}

In this section, we introduce the axiomatic framework used to characterize multi-time states (MTSs) in a wide range of operational probabilistic theories (OPTs), while the structure of the OPTs considered in this work is outlined in Appendix~\ref{App:QSOTinGPT}. To motivate the framework, we begin with the classical case. Consider the random variables $X$, $Y$, and $Z$, with the values $x$, $y$, and $z$ respectively, forming a classical stochastic process $X \to Y \to Z$, governed by the stochastic matrices $P(Y|X)$ and $P(Z|Y)$ for the transitions $X \to Y$ and $Y \to Z$, respectively. For an initial distribution $P(X)$, the multi-time probability distribution $P(ZYX) = P(Z|Y) P(Y|X) P(X)$ satisfies the properties (C1)-(C4), given below.

\begin{PropertyHighlightbox}
\noindent
(C1) {\bf Convex-linearity:} for a $P(X)=\sum_{i}q_i \ P_i (X)$, the multi-time distribution becomes $P(ZYX)=\sum_i q_i \ [P(Z|Y) P(Y|X) P_i(X)]$.\\

\noindent
(C2) {\bf Marginal consistency:} $\sum_{z} P(ZYX)=P(YX)=P(Y|X) P(X)$ and $\sum_{x} P(ZYX)=P(ZY)=P(Z|Y) P(Y)$, where $P(Y)=\sum_x P(YX)$. \\ 

\noindent
(C3) {\bf Compositionality:} for $P(Z|X)=\sum_y P(Z|Y) P(Y|X)$, the joint distribution of $X$ and $Z$ becomes $P(ZX)=\sum_{y}P(ZYX)=P(Z|X) P(X)$.\\

\noindent
(C4) {\bf Reconstructibility:} There exists a stochastic map $Q(Z|Y)$, independent of all possible choices of $P(Y|X)$ and $P(X)$, such that the three-time joint distribution can be reconstructed from its two-time marginal: $P(ZYX)=Q(Z|Y)P(YX)$. Clearly, for a classical Markov chain, this map is $Q(Z|Y)=P(Z|Y)$.
\end{PropertyHighlightbox}

While multi-time classical states are typically described by classical probability distributions satisfying the properties (C1)-(C4), extending this notion to the OPTs, including quantum theory, is not obvious. Here, we develop an axiomatic framework to study MTSs in OPTs. This framework not only reproduces and refines several recent results on quantum MTSs, such as those presented in ~\cite{lie2025multipartite}, but also provides a powerful tool for analyzing spatio-temporal states---namely, states extended across both space and time---while offering fundamental insights into Markovian and non-Markovian evolutions, signaling, and causal structures, both within quantum theory and beyond. We begin by introducing physically motivated axioms to reveal their underlying structure.

\subsection{Notations}
Throughout this work, we consider OPTs (see Appendix~\ref{App:QSOTinGPT} for details on the OPT framework) with finite-dimensional convex state spaces satisfying causality---i.e., a unique way of discarding a system, equivalently a unique deterministic effect---and local tomography, see Appendix~\ref{Sec:OPT_Assumptions} for detailed discussions on these assumptions.

In OPTs, for a physical system $A$, we denote its state space by $\St(A)$, with complex span $\St_{\mathbb{C}}(A) := \operatorname{Span}_{\mathbb{C}}(\St(A)) = \bigl\{ \sum_i a_i \rho_i \;\big|\; a_i \in \mathbb{C},\ \rho_i \in \St(A) \bigr\}$, which we assume throughout to be finite-dimensional. The corresponding effect space is denoted by
$\eff(A) := \bigl\{\, e : e(\rho) \in [0,1] \ \text{for all } \rho \in \St(A) \,\bigr\}$, and its complex span by $\eff_{\mathbb{C}}(A)$. Motivated by quantum theory, the unique deterministic effect is denoted by $\Tr$, satisfying $\Tr(\rho) = 1$ for all $\rho \in \St(A)$. The set of physical transformations in OPTs from $A$ to $B$ is denoted by $\mathsf{Transf}(A\to B)$, with complex span $\mathsf{Transf}_{\mathbb{C}}(A\to B)$, whose elements are linear maps ${\mathcal{N}}: \St_{\mathbb{C}}(A) \to \St_{\mathbb{C}}(B)$; we often write a transformation $\mathcal{N} \in \mathsf{Transf}_{\mathbb{C}}(A \to B)$ simply as $\N: A\to B$. The sequential composition of $\N \in \mathsf{Transf}(A\to B)$ and $\M \in \mathsf{Transf}(B\to C)$ is denoted $\M \circ \N \in \mathsf{Transf}(A\to C)$. Similarly, for $\N \in \mathsf{Transf}(A_1\to B_1)$ and $\M \in \mathsf{Transf}(A_2\to B_2)$, their parallel composition is denoted $\N \otimes \M \in \mathsf{Transf}(A_1 \otimes A_2 \to B_1 \otimes B_2)$. Both, $\circ$ and $\otimes$ are associative. Deterministic transformations---also called channels $\Chan(A \to B)$---form the subset $\Chan(A \to B) \subset \mathsf{Transf}(A \to B)$. Any channel $\N \in \Chan(A \to B)$ satisfies $\Tr_B \circ \N = \Tr_A$, which we call the trace preserving property, by analogy with quantum theory.

Quantum theory in finite dimensions is a special class of OPT satisfying causality and local tomography. In quantum theory, a system $A$ is associated with a Hilbert space $\H _A$, with dimension $d_A$. The space of bounded linear operators on $\H _A$ is denoted by $\mathcal{L}(\H_{A})$. The state space $\St(A)$ consists of density operators---positive trace-one operators in $\mathcal{L}(\H_{A})$---and spans the entire operator space: $\St_{\mathbb{C}}(A) = \mathcal{L}(\mathcal{H}_A)$. For a composite system $A\otimes B$, we write $\St(A \otimes B)$ and $\St(AB)$ interchangeably. Similarly, when the context is clear, we often identify the system $A$ with its operator space $\mathcal{L}(\H_A)$. Physical transformations in $\mathsf{Transf}(A\to B)$ are completely positive trace non-increasing (CPTNI) maps from $\mathcal{L}(\H_{A})$ to $\mathcal{L}(\H_{B})$, while channels in $\Chan(A \to B)$ are completely positive trace preserving (CPTP) maps from $\mathcal{L}(\H_{A})$ to $\mathcal{L}(\H_{B})$. In quantum theory, $\mathsf{Transf}_{\mathbb{C}}(A\to B)$ coincides with the space of linear maps $\mathsf{Lin}(\mathcal{L}(\H_A), \mathcal{L}(\H_B))$, where $\mathsf{Lin}(V,W)$ denotes the vector space of all linear maps from $V$ to $W$.

\subsection{Definitions and Axioms}
We now define Markov compatible\footnote{We say a multi-time state is \emph{compatible} with a causal model---Markovian, non-Markovian, or otherwise---if the state admits an explanation in terms of that model. For instance, a Markov compatible multi-time state is one for which we can assign an initial state and a sequence of local channels on the system, regardless of whether the state actually arose from a Markovian dynamics. Indeed, the same multi-time state may also be compatible with non-Markovian dynamics (see Sec.~\ref{sec:Markvo_non-Markov_compatiblity}) or other causal models~\cite{Song2025}. } multi-time states and the associated axioms, in analogy with the classical Markov chain. Let us consider a sequence of channels $\{ \E_{N-1}, \cdots, \E_2, \E_1 \} \equiv \{ \E_n \}_{n=1}^{N-1}$ leading to an evolution
\[ A_1 \to A_2 \to \cdots A_{n-1} \to A_n \to A_{n+1} \to \cdots A_N,\] 
where each channel $\E_n : A_n \to A_{n+1}$, describes the evolution from time step $n$ to $n+1$, whereby a state $\rho_{A_n} \in \St{(A_n)}$ evolves into a state $\rho_{A_{n+1}} \in \St{(A_{n+1})}$, with the relation $\rho_{A_{n+1}}=\E_n(\rho_{A_n})$. For such a sequence of channels, a multi-time state is defined as follows:

\begin{DefHighlightbox}
\begin{definition} \label{def:MTQS}
(Multi-time state (MTS)): An  $N$-time state $\tau_{A_N \cdots A_2 A_1} \in \St _\mathbb{C} (A_N \otimes \cdots \otimes A_2 \otimes A_1)$ is defined via the function `star product' ($\star$), which takes as inputs a state $\rho_{A_1}\in \St(A_1)$, and a sequence of channels $\{ \E_n \}_{n=1}^{N-1}$, with $\E_n \in \Chan(A_n\to A_{n+1})$, and results in
\begin{align}
\tau_{A_N \cdots A_2 A_1}:= \{\E_{N-1}, \cdots, \E_2, \E_1  \}\star \rho _{A_1},
\end{align}
satisfying the marginal consistency condition (Q2). That is, for $N=2$, the MTS satisfies 
\begin{align}
\Tr_{A_2} \left[ \tau_{A_2A_1} \right] = \rho_{A_1}, \ \mathrm{and} \ \Tr_{A_1} \left[  \tau_{A_2A_1} \right] = \E_1(\rho_{A_1}),
\label{eq:2tMarginal}
\end{align}
and for $N \geq 3$, the MTS satisfies
\begin{align}
\Tr_{A_N} \left[ \tau_{A_N \cdots A_2 A_1} \right] &= \{\E_{N-2}, \cdots, \E_2, \E_1  \}\star \rho _{A_1} , \ \mathrm{and} \nonumber \\ 
\Tr_{A_1} \left[ \tau_{A_N \cdots A_2 A_1} \right] &= \{\E_{N-1}, \cdots, \E_3, \E_2  \} \star \E_1(\rho_{A_1}).
\label{eq:NtMarginal}
\end{align}
\end{definition}
\end{DefHighlightbox}

Note, when restricted to quantum theory, we recover the definition of quantum MTS introduced in~\cite{lie2025multipartite}. Nevertheless, in this axiomatic construction, the `star product' ($\star$) appearing in this definition is yet to be specified and should, in principle, emerge naturally from the axioms that MTSs are required to satisfy, as we uncover below. Note that Eqs.~\eqref{eq:2tMarginal} and ~\eqref{eq:NtMarginal} are generalizations of (C2) to OPT frameworks, and constitute the minimal requirement for the MTS to be an extension of a joint probability distributions over time. These equations imply the state $\tau_{A_N \cdots A_2 A_1}$ satisfies $\Tr[\tau_{A_N \cdots A_2 A_1}] = 1$. Additionally, MTS must satisfy further properties analogous to (C3) and (C4), as stated in the axiom below.

\begin{highlightbox}
\begin{axm}[Properties of MTSs]\label{axm:PropMTQS}
The $N$-time states $\tau_{A_N \cdots A_2 A_1}$, introduced in Definition~\ref{def:MTQS}, satisfy the properties (Q3)--(Q4).\\

\noindent
(Q3) {\bf Compositionality:} For all $ 2 \le n \le N-1 $, the MTS satisfies
\begin{align}
\Tr_{A_n} &\left[ \big\{\mathcal{E}_{N-1}, \cdots,\mathcal{E}_n,  \mathcal{E}_{n-1}, \cdots , \mathcal{E}_1  \big\}\star \rho_{A_1} \right] \nonumber \\
&= \big\{\mathcal{E}_{N-1}, \cdots, \left(\mathcal{E}_n \circ \mathcal{E}_{n-1} \right), \cdots , \mathcal{E}_1  \big\}\star \rho_{A_1},
\end{align}
where $(\mathcal{E}_n \circ \mathcal{E}_{n-1}): A_{n-1} \to A_{n+1}$. \\

\noindent
(Q4) {\bf Reconstructibility:} For any arbitrary cut $A_N,\cdots,A_{n+1}$, and $A_n,\cdots,A_1$, there exists a linear map, $\mathcal{R}_n \in \mathsf{Transf_{\mathbb{C}}}(A_n\to A_N \otimes \cdots \otimes A_{n+1}\otimes A_n)$, independent of all $\rho \in \St(A_1)$ and past channels $\{\E_1,\E_2,\cdots,\E_{n-1}\}$,  such that 
\begin{equation}
\mathcal{R}_n\circ \Tr_{A_N\cdots\ A_{n+1}} \left[\{\E_n\}_{n=1}^{N-1} \star \rho_{A_1} \right]=\{\E_n\}_{n=1}^{N-1} \star \rho_{A_1},  \ \forall n.
\end{equation}
We call $\mathcal{R}_n$ the reconstructing map. Note that, since the reconstructing map $\mathcal{R}_{n=1}$ is linear, the reconstructibility implies the star product is linear in states.
\end{axm}
\end{highlightbox}

The property (Q3) signifies the compositionality of MTS, analogous to (C3) in the classical cases. It is important to emphasize the key distinction between compositionality defined by (Q3) and the one introduced in Ref.~\cite{lie2024quantum}, where the latter is restricted to quantum theory only. In particular, for the case $N = 3$, the compositionality axiom in Ref.~\cite{lie2024quantum} is formulated as follows:
\begin{align}
\Tr_{A_2}\left[ \mathcal{E}_2\star \left(\mathcal{E}_1 \star \rho_{A_1}\right)\right] {=} \left(\mathcal{E}_2{\circ}  \mathcal{E}_1\right) \star \rho_{A_1}.\label{Eq:Lie_Ng_Compositionality}  
\end{align}
Note, Eq.~\eqref{Eq:Lie_Ng_Compositionality} is a special case of our compositionality property (Q3), under the assumption:
\begin{align}
  \left\{\mathcal{E}_2,  \ \mathcal{E}_1\right\}\star \rho_{A_1}{=}\mathcal{E}_2\star \left(\mathcal{E}_1 \star \rho_{A_1}\right).~\label{Eq:Lie-Ng-Compositionality-assumption}
\end{align}
We do not impose such an assumption a priori; the structure given in Eq.~\eqref{Eq:Lie-Ng-Compositionality-assumption} is derived axiomatically at a later stage. 

For the evolution generated by a sequence of channels considered here, the corresponding MTS is expected to exhibit a property similar to the short Markov property~\cite{petz1988sufficiency,hayden2004structure}. This means that the future state can be constructed using the information available at present without requiring any information from the past, i.e., each successive evolution is memoryless and admits a reconstructing map that generates the full state from the present state. This is analogous to the property (C4) in the classical cases. Therefore, for MTSs, if we posit any cut splitting the timeline into a ``past'' ($A_{n-1},\cdots, A_1$), a ``present'' ($A_n$), and a ``future'' ($A_N, \cdots, A_{n+1}$), the full state over time $\tau_{A_N \cdots A_1}$ should be reconstructible via a linear map $\mathcal{R}_n$ acting only on the present system $A_n$. The reconstructibility property (Q4), in essence, captures the Markov compatibility of the evolution considered here.

While the property (Q4) bears some resemblance to
reconstruction in Markov chain states in quantum theory, it departs from this analogy in two important respects. First, the object to be reconstructed is not a physical state, but a state over time. As emphasized earlier, temporal states differ fundamentally from physical states; for example, a quantum MTS need not be positive semidefinite~\cite{leifer2013towards, fitzsimons2015quantum, fullwood2022quantum}. Thus, they already lie outside the standard setting of Markov states. Second, we do not impose a priori that the reconstructing map be a physical channel allowed by the underlying theory, e.g., a completely positive trace preserving (CPTP) map in quantum theory. Within our axiomatic framework applicable to OPTs, the reconstructing map is allowed to be a linear map in the span of physical transformations. One can interpret these maps as Bayesian update rule that can be used to reconstruct the future from the present.

\subsection{Temporal extension map}
We now introduce a linear map, called the \emph{temporal extension map}, which, when applied to the input physical state, reproduces the star product. The temporal extension map reflects the spirit of operational probabilistic theories~\cite{hardy2001quantum, Barrett2007, chiribella2010probabilistic, Chiribella_2011, DAriano2016} and categorical quantum theory~\cite{abramsky2007categorical,abramsky2008categorical,Coecke2017, Selby2021}, where transformations, or morphisms in the language of category theory, are treated as primitive objects of interest. In this language, a channel is one such morphism, transforming a physical state to a physical state. This perspective offers a natural foundation for generalizing multi-time states to broader operational frameworks. For instance, the temporal extension map in quantum theory takes an input multi-time state to an output multi-time state, abstracting away the Hilbert-space-specific details of the star product. Note that a map similar to the temporal extension map, in the context of quantum states over time, has appeared in earlier works~\cite{lie2024quantum, lie2025multipartite}. For two-time quantum states, it is called a \emph{bloom} of the quantum channel in Ref.~\cite{lie2025multipartite} (see Definition~5). Now, we introduce the general temporal extension map, with the definition below, and systematically develop its role within our axiomatic framework for OPTs.

\begin{DefHighlightbox}
\begin{definition}[Temporal extension map]\label{Def:tempExt}
A temporal extension map $\mathcal{C}_{\{\mathcal{E}_n\}_{n=1}^{N-1}} \in \mathsf{Transf}_{\mathbb{C}}(A_1 \to A_N\otimes{\cdots }\otimes A_2\otimes A_1)$, is parameterized by a sequence of channels $\{ \E_{N-1}, \cdots, \E_2, \E_1 \} \equiv \{ \E_n \}_{n=1}^{N-1}$ with $\E_n : A_n \to A_{n+1}$ and is defined by the properties (L0) and (L1).
\begin{itemize}
\item \noindent
(L0) \textbf{Consistency with the star product ($\star$):} Its action on a state $\rho_{A_1} \in \St(A_1)$ reproduces the star product, mapping to the multi-time state:
\begin{align}
 &\{ \E_n \}_{n=1}^{N-1} \star \rho_{A_1} = \mathcal{C}_{\{ \E_n \}_{n=1}^{N-1}} (\rho_{A_1}).  \label{Eq:def1_TempExt_map}
\end{align}

\item\noindent
(L1) \textbf{Linearity:} It is a linear map with the input vector space being $\St_{\mathbb{C}}(A_1)$, i.e., for all states $\rho_{A_1}^{(i)} \in \St(A_1)$ and complex numbers $a_i\in \mathbb{C}$,
\begin{align}
 \mathcal{C}_{{\{\E_n \}_{n=1}^{N-1}}} \left(\sum_i  a_i \ \rho_{A_1}^{(i)}\right) & = \sum_i a_i \ \mathcal{C}_{\{\E_n \}_{n=1}^{N-1}} \left(\rho_{A_1}^{(i)}\right) \nonumber \\ &= \sum_{i}a_i \ \left(\{\E_n \}_{n=1}^{N-1}\star \rho_{A_1}^{(i)}\right).  \label{Eq:linearTempExt_map} 
\end{align} 
\end{itemize}
\end{definition}
\end{DefHighlightbox}

The properties of the multi-time states in Definition~\ref{def:MTQS}, together with the Axiom~\ref{axm:PropMTQS} satisfied by the star product ($\star$), can be directly translated to the temporal extension map, given by (L2)-(L4) below.

\begin{PropertyHighlightbox}
\noindent
(L2) {\bf Marginal consistency:} For $N{=}2$, i.e., with $\E_1: A_1 \to A_2$, the corresponding temporal extension map $\C_{\E_1}: A_1 \to A_2\otimes A_1$, satisfies
\begin{align}
&\Tr_{A_2} \circ \C_{\E_1}=\id_1 \label{Eq:D2}, \\
 &\Tr_{A_1} \circ \C_{\E_1}=\E_1, \label{Eq:D1}  
\end{align}
where $\id_1: A_1 \to A_1$ is an identity map. For $N \ge 3$, the temporal extension map $\C_{\{ \E_n\}_{n=k}^{N-1}}:A_k \to A_N\otimes \cdots \otimes A_{k+1} \otimes A_k$, with $1\le k \le N-1$, satisfies
\begin{align}
 &\Tr_{A_N} \circ \C_{\{\E_n\}_{n=1}^{N-1}}=\C_{\{\E_n\}_{n=1}^{N-2}}, \label{Eq:D3}\\
 &\Tr_{A'_1} \circ \C_{\{\E_n\}_{n=1}^{N-1}}=\C_{\{\E_n\}_{n=2}^{N-1}} \ \circ \  \E_1. \label{Eq:D4}
\end{align}

\noindent
(L3) {\bf Compositionality:} For any $N$, the temporal extension map satisfies the compositionality condition $\forall \ 2\le n \le N-1$, as
\begin{align}
\Tr_{A_n} \circ \C_{\big\{\E_{N-1}, \cdots,\E_n, \  \E_{n-1}, \cdots, \E_1  \big\}} =\ \C_{\big\{\E_{N-1}, \cdots, \left(\E_n \circ  \E_{n-1}\right), \cdots, \E_1  \big\}}, 
\end{align}
where $(\mathcal{E}_n \circ \mathcal{E}_{n-1}): A_{n-1} \to A_{n+1}$.\\

\noindent
(L4) {\bf Reconstructibility:} If a temporal extension map $\C _{\{ \E_n \}_{n=1}^{N-1}}$ exists, the property (Q4) ensures:
\begin{align}
\mathcal{R}_n\circ \Tr_{A_N,\cdots,\ A_{n+1}} \circ \C_{\{ \E_n \}_{n=1}^{N-1}}=\C_{\{ \E_n \}_{n=1}^{N-1}}.
\end{align}  
\end{PropertyHighlightbox}

Properties (L2)–(L4) follow directly from Lemma~\ref{lem:local-tomography-linear} in Appendix~\ref{App:QSOTinGPT}.  For example, consider Eq.~\eqref{Eq:D3}:
\begin{align}
    \Tr_{A_N}\circ&\C_{\{\mathcal{E}\}_{n=1}^{N-1}}(\rho_{A_1})\nonumber\\
    &=\Tr_{A_N}\circ\{\E_{N-1},\cdots,\E_{2},\E_{1}\}\star\rho_{A_1}\nonumber\\
    &=\{\E_{N-2},\cdots,\E_{2},\E_{1}\}\star\rho_{A_1}\nonumber\\
  &=\C_{\{\mathcal{E}\}_{n=1}^{N-2}}(\rho_{A_1})\nonumber
\end{align}
The first equality is due to the definition of the temporal extension map. The second equality is due to the definition of the star product. And lastly, because the OPTs, considered in this work, satisfy local tomography (See Lemma~\ref{lem:local-tomography-linear}), we have
\begin{equation*}
     \Tr_{A_N}\circ\C_{\{\mathcal{E}\}_{n=1}^{N-1}}=\C_{\{\mathcal{E}\}_{n=1}^{N-2}}
\end{equation*} 
The rest of the statements in (L2)-(L4) are also true due to a similar line of arguments. Note that the properties (L2)-(L4) hold for any OPT that satisfies the assumptions of Appendix~\ref{Sec:OPT_Assumptions}, which includes quantum theory as a special case.

It is worth noticing that Axiom~\ref{axm:PropMTQS} guarantees a unique temporal extension map. In particular, the following proposition shows that the temporal extension map is well-defined and unique whenever the star product is convex-linear in states---a condition implied by reconstructibility.

\begin{PropHighlightbox}
\begin{proposition}\label{Prop:stLinEquivLinTemExt}
A unique temporal extension map $\C_{\{\E_n\}_{n=1}^{N-1}}
$, as defined in Definition~\ref{Def:tempExt}, exists if and only if the associated star product is convex-linear in the states, in accordance with Axiom~\ref{axm:PropMTQS}.
\end{proposition} 
\end{PropHighlightbox}

The proof of this proposition in quantum theory is given in Ref.~\cite[Remark~4]{lie2024quantum}, and naturally extends to our OPT framework. The latter is proven in Appendix~\ref{App:proof_stLinEquivLinTemExt}.

A particular case of interest arises when the map is the identity, $\E_1 = \id$, for $N=2$. According to Eqs.~\eqref{Eq:D2} and \eqref{Eq:D1}, the corresponding temporal extension map $\C_{{\id}}: A_1 \rightarrow A''_1\otimes A_1'$ is a broadcasting map, defined as follows.

\begin{DefHighlightbox}
\begin{definition}[Broadcasting map] \label{def:broadcasting}
A linear map $\mathcal{C}_{\mathcal{I}}\in \mathsf{Transf}_{\mathbb{C}}(A\to A''\otimes A')$, with the isomorphic input and output systems: $A{\simeq} A' {\simeq} A''$, is a broadcasting map if 
\begin{align}
\Tr_{A'}\circ \mathcal{C}_{\mathcal{I}} &=\id_{A\to A''} ,\ \ \mathrm{and} \\
\Tr_{A''}\circ \mathcal{C}_{\mathcal{I}} &=\id_{A\to A'}.
\end{align}
\end{definition}
\end{DefHighlightbox}

In the rest of the paper, we use $A'', A'$ and $A$ interchangeably whenever there is no confusion. Several no-go results~\cite{barnum1996noncommuting, barnum2007generalized, Luo2010, Rolino2025} forbid broadcasting maps from being physical channels in $\Chan(A \to A'' \otimes A'$) in quantum and broader physical theories. Nevertheless, such maps can be defined within the linear span of physical transformations.  For example, when the star product is restricted to the Jordan product quantum state over time, the map reduces to the virtual broadcasting map introduced in~\cite{parzygnat2024virtual}, thereby recovering it as a special case of the axiomatic construction defined above. We shall revisit this map later, showing that it naturally arises in the construction of multi-time states.

Having introduced the axiomatic framework for multi-time states (MTSs) in OPTs and reformulated it in terms of temporal extension maps, we now apply the framework to investigate various dynamical aspects of physical evolutions. These include the characterization and analysis of Markov and non-Markov compatible MTSs, the role of dynamical memory, signaling, and causal structures.    

\section{Markov compatible Multi-Time States \label{sec:MarkovCaopatible}}

Having defined multi-time states in Definition~\ref{def:MTQS}, we now explore their structure when generated by a sequence of channels $\{ \E_n \}_{n=1}^{N-1}$ acting on an initial state  $\rho_{A_1} \in \St(A_1)$.  Our first step is to derive the implications of Axiom~\ref{axm:PropMTQS} for such states.

\begin{PropHighlightbox}
\begin{proposition}\label{Pro:ReconAndComposition}
Axiom~\ref{axm:PropMTQS} has the following implications for multi-time states:
\begin{enumerate}
\item The Reconstructibility axiom (Q4) implies that any multi-time state $\{ \E_n \}_{n=1}^{N-1}\star \rho_{A_1}$ decomposes as a composition of temporal extension maps:
\begin{align}\label{Eq:MarkovCompositional}
\{ \E_n \}_{n=1}^{N-1} \star \rho_{A_1}  &=\mathcal{C}_{\mathcal{E}_{N-1}}\circ \cdots \mathcal{C}_{\mathcal{E}_2} \circ \mathcal{C}_{\mathcal{E}_1}(\rho_{A_1}).
\end{align}

\item Assume every multi-time state has the structure shown in Eq.~\eqref{Eq:MarkovCompositional}. Then the Compositionality axiom (Q3) further restricts the structure of every two-time state $\E \star \rho_{A_1}$:
\begin{align}\label{Eq:Compositionality}
\E \star \rho_{A_1} = \C_{\E}(\rho_{A_1}) = (\E \otimes \id) \circ \mathcal{C}_{\id}(\rho_{A_1}).
\end{align}
\end{enumerate}
\end{proposition}
\end{PropHighlightbox}

The proof of the proposition for OPTs is given in Appendix~\ref{App:Pro:ReconAndComposition}. Note, the Eq.~\eqref{Eq:MarkovCompositional} on reconstructibility shows that the multi-time state $\{\E_n\}_{n=1}^{N-1} \star \rho_{A_1}$ can be expressed as an overall temporal extension map acting on the initial state:
\begin{align}
    \{\E_n\}_{n=1}^{N-1} \star \rho_{A_1} &= \mathcal{C}_{\{\E_n\}_{n=1}^{N-1}} (\rho_{A_1}), \nonumber \\
    &= \mathcal{C}_{\mathcal{E}_{N-1}} \circ \cdots \circ \mathcal{C}_{\mathcal{E}_2} \circ \mathcal{C}_{\mathcal{E}_1} (\rho_{A_1}).
    \label{Eq:two_time_structure}
\end{align}
Since each temporal extension map $\mathcal{C}_{\mathcal{E}_i}$ is linear, by Proposition~\ref{Prop:stLinEquivLinTemExt}, the multi-time state must be linear in the state. The converse, however, is not true. One can construct multi-time states that are linear in states yet do not satisfy the reconstructibility axiom. We provide an explicit example of one such `temporal state' in quantum theory in Appendix~\ref{app:uniqeMTQS}. 

\subsection{Structure theorem for Markov compatible MTSs}

Eq.~\eqref{Eq:two_time_structure} can be understood as follows. Consider a three-time state $\{ \E_2, \E_1\}\star \rho_{A_1}$, where $\E_1 = \id$ and $\E_2=\E$ is an arbitrary channel. Given this, applying compositionality to Eq.~\eqref{Eq:MarkovCompositional} yields the two-time state structure. This reveals that the broadcasting map $\mathcal{C}_{\id}$ is a fundamental building block: each temporal extension map $\mathcal{C}_{\E}$ can be written as $(\E \otimes \id) \circ \mathcal{C}_{\id}$. The proposition above, therefore, directly implies the following theorem:

\begin{ThmHighlightbox}
\begin{theorem}\label{Thm:StructureThm_MarkovQSOT}
(Decomposable multi-time state.) A multi-time state $ \{ \E_n \}_{n=1}^{N-1} \star \rho_{A_1}$ satisfying Axiom~\ref{axm:PropMTQS} has the structure
\begin{align}
    &\{ \E_n \}_{n=1}^{N-1} \star \rho_{A_1}=\mathcal{E}_{N-1}\star(\cdots\star(\mathcal{E}_1\star\rho_{A_1})) \\
    &=(\E_{N-1}\otimes \id^{N-1})\circ \cdots \circ  (\E_2 \otimes \id^{\otimes 2}) \nonumber \\ 
    &\qquad \circ (\C_{\mathcal{I}} \otimes 
 \id)\circ (\E_1 \otimes \id) \circ \C_{\mathcal{I}}(\rho_{A_1}).
\end{align}
\end{theorem}
\end{ThmHighlightbox}

Theorem~\ref{Thm:StructureThm_MarkovQSOT} tells us that at each time step, the construction proceeds by applying a broadcasting map $\C_{\mathcal{I}}$ followed by a channel $\E_n$ acting on one of the subsystems. The theorem leaves $\mathcal{C}_{\id}$ unspecified; its precise form can be fixed by additional axioms (see Refs.~\cite{parzygnat2024virtual, xiao2025no}, or our discussion on boxworld multi-time states in Section~\ref{sec:MTSinBoxWorld}). Beyond this freedom, however, the overall structure of multi-time states is uniquely determined. 

The structure of MTS, given in Theorem~\ref{Thm:StructureThm_MarkovQSOT}, is significant because, in quantum theory, it establishes a connection between quantum MTS and the well-established framework of multi-time processes ~\cite{chiribella09b, Pollock2018Operational, Pollock2018, Giarmatzi_2021, Goswami2021,Goswami2025}, while the latter provides a complete operational characterization of Markovian and non-Markovian dynamics. In particular, the structure in Theorem~\ref{Thm:StructureThm_MarkovQSOT} is directly related to Markovian multi-time processes and allows for a natural generalization to non-Markovian settings, which we explore next.

\section{General multi-time states and non-Markov compatibility \label{sec:GenMTQSnonMarkov}}

We have so far focused on multi-time states compatible with Markovian dynamics. We now extend this framework to the non-Markovian case. More precisely, we first formally define what it means for a multi-time state to be compatible with a general multi-time process, including those exhibiting non-Markovianity. We then introduce a new axiom that, along with Axiom~\ref{axm:PropMTQS}, that allows us to determine the structure of such states. 

To build intuition, consider a classical analog. Suppose we have `system' random variables $X,Y,Z$ and `environment' random variables $\Lambda_1, \Lambda_2, \Lambda_3$. The dynamics begin with an initial joint distribution $P(X,\Lambda_1)$, followed by the evolutions corresponding to stochastic matrices $P(Y,\Lambda_2|X,\Lambda_1)$ and $P(Z,\Lambda_3|Y,\Lambda_2)$. The distribution over the observed system variables obtained by marginalizing over the environment variables:
\begin{align}
&P(X,Y,Z)=\sum_{\lambda_1,\lambda_2,\lambda_3} P(Z, \lambda_3| Y,\lambda_2 )P(Y, \lambda_2| X,\lambda_1 )  P(X, \lambda_1 ).
\end{align}
Here, $\lambda_i$ denotes a specific value of the random variable $\Lambda_i$, and the sum is over all possible values. Note, the resulting distribution may be non-Markovian when $P(XYZ)\neq P(Z|Y)P(Y|X)P(X)$. However, this distribution can be `Markovianized' by including the hidden environment variables, forming the Markov chain $(X,\Lambda_1) \to (Y,\Lambda_2) \to (Z,\Lambda_3)$. 

This classical picture generalizes to OPTs as follows. We start with a Markov compatible multi-time state defined over both the system and the environment. In other words, the system considered in Definition~\ref{def:MTQS} is now taken to include both the original system and its environment, which jointly undergo a sequence of global channels. Subsequently, we discard all the environmental degrees of freedom. Thus, we have a formal definition of the general multi-time state below.

\begin{DefHighlightbox}
\begin{definition}[General multi-time states]\label{Def:QSOTGen}
Given an initial state $\sigma\in \St(A_1\otimes E_1)$, and a set of channels $\{\mathcal{F}_{n} \in \Chan (A_n\otimes E_n\to A_{n+1}\otimes E_{n+1})\}_{n=1}^{N-1}$, a general multi-time state $\tau\in \St_{\mathbb{C}}(A_N\otimes A_{N-1}\otimes \cdots \otimes A_1)$ is defined as
\begin{align}
 \tau_{A_N\cdots A_2A_1}=\Tr_{E_N\cdots E_2E_1} \left(\{\mathcal{F}_{n}\}_{n=1}^{N-1} \star \sigma \right).   
\end{align}
\end{definition}
\end{DefHighlightbox}

 This general setting naturally accommodates processes with memory, commonly referred to as \emph{non-Markovian processes}. Such memory arises when information about the system is transferred to the environment during an earlier interaction and subsequently returns to affect the system dynamics at a later time. This backflow of information gives rise to temporal correlations between the system's local intervention outputs across different time steps, even when an intermediate intervention fully resets the state. 

Since Definition~\ref{Def:QSOTGen} involves joint system-environment evolutions, followed by discarding the environment, we need an axiom that specifies how the star product behaves under such discarding of a subsystem. In particular, consider the simplest scenario with two time steps, in which both the initial physical state and the joint evolution take a product form. Let the initial state be $\sigma_{A_1E_1}=\rho\otimes \tau \in \St(A_1 \otimes E_1)$ and let the channel be $\mathcal{F}=\E \otimes \N$ with $\E:A_1\to A_2$ acting on the system and $\N:E_1\to E_2$ acting on the environment. As a result, the final state become $\E (\rho)\otimes \N(\rho) \in \St(A_2 \otimes E_2)$.  The question is what do we expect when we compute $\Tr_{E_2E_1} \big[ (\mathcal{E} \otimes \mathcal{N}) \star (\rho \otimes \tau) \big]$.

To gain insight, we note that, while $A_1$ and $E_1$ (or $A_2$ and $E_2$) are spacelike separated, the $A_1$ and $A_2$ (or $E_1$ and $E_2$) are timelike separated. Also, the composites $A_1E_1$ and $A_2E_2$ are timelike separated. For the initial state $\rho \otimes \tau$, any physical transformation, such as a measurement or instrument, performed on $A_1$ (or on $E_1$) does not affect the state of $E_1$ (or of $A_1$). The same holds for the final state $\E(\rho) \otimes \N(\tau)$. Moreover, for uncorrelated states and independent channels, any physical transformation performed on $A_1$ (or on $E_1$) does not influence the state of $E_2$ (or of $A_2$). An analogous property holds in the reverse temporal direction: any physical transformation on $A_2$ (or on $E_2$) does not affect the state of $E_1$ (or of $A_1$). This we call spatio-temporal marginal consistency. Such marginal consistency conditions should likewise hold for multi-time states spanning two or more spacetime events, as formalized in the following axiom imposing a minimal set of constraints\footnote{The term "minimal" is used because a stronger axiom can be considered: $ (\E \otimes \N)\star (\rho \otimes \tau)=(\E \star \rho) \otimes (\N \star \tau)$. This holds in classical theory for independent initial joint probabilities and stochastic maps, and many quantum generalizations---such as the Leifer-Spekkens~\cite{leifer2013towards}, uncorrelated~\cite{parzygnat2023time}, Ohya compound~\cite{ohya1983note}, left, and right bloom star products~\cite{parzygnat2023time}, and boxworld multi-time states in Section~\ref{sec:MTSinBoxWorld}---but fails for the star product considered in Eq.~\eqref{eq:StarProduct}. Axiom~\ref{axm:spatial-marginal-consistency}, on the other hand, appears to hold for all these star products.}.

\begin{highlightbox}
\begin{axm}[Spatio-Temporal Marginal Consistency]
\label{axm:spatial-marginal-consistency}
For all local channels $\mathcal{E}: A_1 \to A_2$ and $\mathcal{N}: E_1 \to E_2$, and uncorrelated states $\rho_{A_1}$ and $\tau_{E_1}$, corresponding to spacelike separated systems $A_1$ and $E_1$, the spatio-temporal state $\gamma_{A_2A_1E_2E_1}=(\mathcal{E} \otimes \mathcal{N}) \star (\rho \otimes \tau)$ must satisfy the following:
\begin{align}
\gamma_{A_2A_1E_2} = \Tr_{E_1}\big[\gamma_{A_2A_1E_2E_1}] &= (\mathcal{E} \star \rho)_{A_2A_1} \otimes \mathcal{N}(\tau)_{E_2}, \label{eq:A1A2E2} \\
\gamma_{A_2A_1E_1} = \Tr_{E_2}\big[\gamma_{A_2A_1E_2E_1}\big] &= (\mathcal{E} \star \rho)_{A_2A_1} \otimes \tau_{E_1}, \label{eq:A1A2E1}\\
\gamma_{A_2E_2E_1}  = \Tr_{A_1}\big[\gamma_{A_2A_1E_2E_1}\big]&= \mathcal{E}(\rho)_{A_2} \otimes (\mathcal{N} \star \tau)_{E_2E_1}, \label{eq:A2E1E2} \\
\gamma_{A_1E_2E_1}  = \Tr_{A_2}\big[\gamma_{A_2A_1E_2E_1}\big]&= \rho_{A_1} \otimes (\mathcal{N} \star \tau)_{E_2E_1}. \label{eq:A1E1E2}
\end{align}
\end{axm}
\end{highlightbox}

For a star product that is linear in states, Proposition~\ref{Prop:stLinEquivLinTemExt} guarantees a unique temporal extension map, letting us re-express Axiom~\ref{axm:spatial-marginal-consistency} via Proposition~\ref{prop:TemporalExtProdChan} below.

\begin{PropHighlightbox}
\begin{proposition} \label{prop:TemporalExtProdChan}
When the star product ($\star$) is convex-linear in states and satisfies spatial marginal consistency axiom (Axiom~\ref{axm:spatial-marginal-consistency}) then for all channels $\E:A_1\to A_2$ and $\N : E_1\to E_2$, the corresponding temporal extension map $\C_{\E\otimes \N}: {A_1\otimes E_1\to A_2\otimes A_1\otimes E_2\otimes E_1}$ satisfies
\begin{align}
    &\Tr_{E_1}\circ \C _{\E\otimes \mathcal{N}} = \C _{\E} \otimes \mathcal{N}_{E_1\to E_2},\\
    & \Tr_{E_2}\circ \C _{\E\otimes \mathcal{N}}= \C _{\E} \otimes \id_{E_1\to E_1},\\
   & \Tr_{A_1}\circ \C _{\E\otimes \mathcal{N}} = \E_{A_1\to A_2} \otimes \C _{\mathcal{N}},\\
   & \Tr_{A_2}\circ \C _{\E\otimes \mathcal{N}} = \id_{A_1\to A_1} \otimes \C _{\mathcal{N}}.
\end{align}
\end{proposition}
\end{PropHighlightbox}

The proof immediately follows from Lemma~\ref{lem:local-tomography-linear} in Appendix~\ref{App:QSOTinGPT}. In the next subsection we show that the property of the joint broadcasting map $\C_{\id \otimes \id}$, obtained as a special case of Proposition~\ref{prop:TemporalExtProdChan} with $\mathcal{E} = \id$ and $\mathcal{N} = \id$, together with Axiom~\ref{axm:PropMTQS}, suffices to fully characterize the structure of general multi-time quantum states.

\begin{figure}
    \centering
    \includegraphics[width=\columnwidth]{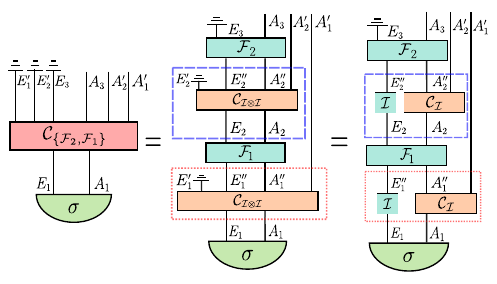}
    \caption{\textbf{Proof sketch for Theorem 2:} We show a pictorial proof sketch of Theorem~\ref{Thm:NMQSOT} for $N=3$. In the first figure, for a state-linear star product, we have $\Tr_{E_3E_2'E_1'}\big[\{\mathcal{F}_2,\mathcal{F}_1\}\star \sigma \big] =\Tr_{E_3E_2'E_1'} \big[\mathcal{C}_{\{\mathcal{F}_2,\mathcal{F}_1\}}(\sigma)\big]$, i.e., we express the left-hand side in terms of the temporal extension map $\mathcal{C}_{\{\mathcal{F}_2,\mathcal{F}_1\}}$. If the properties of Axiom~\ref{axm:PropMTQS} are respected, then the first equality follows from Theorem~\ref{Thm:StructureThm_MarkovQSOT} applied to the joint system-environment multi-time state. The second equality follows from Proposition~\ref{prop:TemporalExtProdChan} (a consequence of Axiom~\ref{axm:spatial-marginal-consistency}), with all product channels taken as identity channels, $\id$. We illustrate this by comparing the dashed blue and dotted red lines in the second and third figures, where Proposition~\ref{prop:TemporalExtProdChan} is applied. The identity maps in the environment compose the state $\sigma$ and the channels $\mathcal{F}_1$ and $\mathcal{F}_2$ through their environment degrees of freedom. When restricted to quantum theory, this corresponds precisely to the multi-time process in Eq.~\eqref{Eq:semicausality_twoT_process} for $N=3$.}
    \label{fig:Intuition_Thm2}
\end{figure}

\subsection{Structure theorem for general multi-time states}
The structure of a general multi-time state is given in the following theorem. 

\begin{ThmHighlightbox}
\begin{theorem}\label{Thm:NMQSOT}
Given a sequence of channels $\{\mathcal{F}_n\}_{n=1}^{N-1}$, with $\mathcal{F}_{n} \in \Chan (A_n\otimes E_n\to A_{n+1}\otimes E_{n+1})$, and an initial state $\sigma \in \St{(A_1E_1)}$, the corresponding general multi-time state $\tau_{A_N\cdots A_2A_1}$, respecting Axioms~\ref{axm:PropMTQS} and \ref{axm:spatial-marginal-consistency} satisfies
\begin{align}
&\tau_{A_N\cdots A_2A_1}=\Tr_{E_N\cdots E_1}\big[\{\mathcal{F}_n\}_{n=1}^{N-1}\star \sigma \big] \nonumber \\
&=\Tr_{E_{N}}\circ \mathcal{F}_{N-1}\circ (\id_{E_{N-1}} \otimes \C_{\id} ){\circ \cdots \circ} \mathcal{F}_1\circ (\id_{E_1} \otimes \C_{\id} )\circ \sigma.   
    \end{align}
\end{theorem}
\end{ThmHighlightbox}

We have provided a pictorial proof sketch of Theorem~\ref{Thm:NMQSOT} in Fig.~\ref{fig:Intuition_Thm2} for $ N=3$. The proof can be extended to general $N$ by a simple inductive argument. As mentioned earlier, we see that Theorem~\ref{Thm:StructureThm_MarkovQSOT} (alternatively, Corollary~\ref{Cor:MarkovProc_StTime}) is a special case of Theorem~\ref{Thm:NMQSOT} when restricted to Markovian processes.

\subsection{Semicausality and non-Markov compatibility \label{sec:CausalSignalNonMarkov}}
Our axiomatic framework not only characterizes the structure of multi-time states but also provides a sufficient condition for certifying signaling between systems. To build intuition, first consider the simplest scenario: a single input-output channel $\E:A\to B$. In the extreme case where the channel $\E$ does not allow any signaling from $A$ to $B$, it must be a \emph{constant channel}: $\E(\rho)=\Tr(\rho)\phi$ for all $\rho \in \St(A)$, i.e., it outputs a fixed state $\phi$ independent of the input. For such a channel, if Axiom~\ref{axm:PropMTQS} is satisfied, then due to Theorem~\ref{Thm:StructureThm_MarkovQSOT}, the corresponding two-time state takes the form  
\begin{align}
 \E \star \rho = \phi \otimes \Tr_{A''}[\mathcal{C}_{\id}(\rho)] = \phi \otimes \rho.   
\end{align}
This implies that, when there is no signaling, the two-time state is always a well-defined state (in fact, a product state). Conversely, if the output of the star product fails to be an element of the state space $\St(A\otimes B)$, for instance, a negative operator in the quantum theory, then the channel must allow signaling from $A$ to $B$. This line of reasoning has been fruitfully explored in the literature on pseudo-density matrices~\cite{fitzsimons2015quantum,Liu2025, song2024causal} in quantum theory, where negativity serves as a sufficient condition of signaling. The origin of such negativity lies in the nature of the broadcasting map $\mathcal{C}_{\id}$. While the no-broadcasting theorem~\cite{barnum1996noncommuting, barnum2007generalized} precludes the existence of a completely positive trace preserving (CPTP) broadcasting map, a virtual broadcasting map, which may not be completely positive, can still be defined~\cite{parzygnat2024virtual}. In quantum theory, such maps generally output operators that are trace-one but not necessarily positive, thereby introducing negativity in the resulting multi-time state.

This intuition extends naturally to bipartite channels, where the notion of signaling becomes more nuanced. Consider a bipartite channel $\mathcal{F}:A_I B_I\rightarrow A_O B_O$. A particularly important class of these bipartite channels is \emph{semicausal channels}~\cite{Preskill_casual, Eggeling2002Semicuasal, Kissinger2019} which forbids signaling from one input to the other output. Here, without loss of generality, we focus on channels that prevent signaling from $A_I$ to $B_O$ ($A_I\not \to B_O$).

\begin{figure}
\centering
\includegraphics[width=\columnwidth]{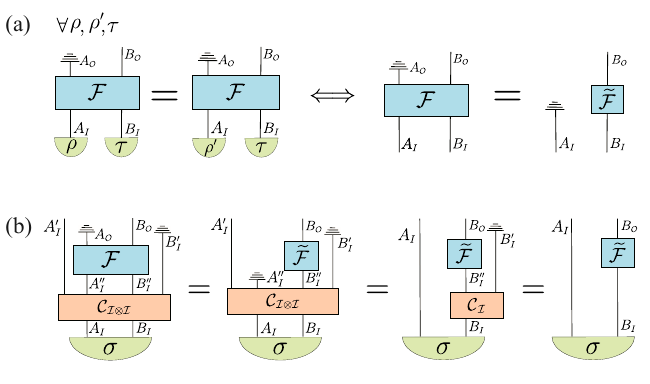}
\caption{\textbf{Consequence of semicausality on MTS.} (a) A semicausal channel $\mathcal{F}:A_IB_I\rightarrow A_O B_O$ that forbids signaling from $A_I$ to $B_O$, denoted as $A_I\not \to B_O$. Operationally, $A_I\not \to B_O$ implies for all $\rho, \rho' \in \St(A_I)$ and $\tau \in \St(B_I)$, we have $\Tr_{A_O}\circ \mathcal{F}(\rho \otimes \tau) = \Tr_{A_O}\circ \mathcal{F}(\rho' \otimes \tau)$~\cite{Preskill_casual}. Equivalently, in local tomographic OPTs, such a channel admits a marginal channel $\widetilde{\mathcal{F}}_{B_I\to B_O}$ satisfying $\Tr_{A_O}\circ \mathcal{F} = \Tr_{A_I}\otimes \widetilde{\mathcal{F}}$~\cite{schumacher2004locality,Ormrod2023}. (b) Pictorial proof of Theorem~\ref{thm:semicausal}. For a semicausal channel $\mathcal{F}$, if the star product satisfies Axioms \ref{axm:PropMTQS}, and \ref{axm:spatial-marginal-consistency}, then for an arbitrary input state $\sigma\in \St(A_I\otimes B_I)$, the reduced operator $\Tr_{A_OB_I'} ( \mathcal{F} \star \sigma)$ simplifies to $(\id \otimes \widetilde{\mathcal{F}})(\sigma)$, which is a valid state. (b) In the figure on the lower panel, the first equality is due to the semicausality of $\mathcal{F}$, the second equality is due to the spatio-temporal marginal consistency Axiom, and the third equality is due to the property of the broadcasting map. }
\label{fig:QSOT-Semicausality}
\end{figure}

As shown in Fig.~\ref{fig:QSOT-Semicausality}(a), by definition~\cite[Definition 2.9]{Kissinger2019} a semicausal channel admits a marginal channel $\widetilde{\mathcal{F}}:B_I\to B_O$ satisfying:
\begin{align}
\Tr_{A_O}\circ \mathcal{F}_{A_IB_I\to A_OB_O} = \Tr_{A_I}\otimes \widetilde{\mathcal{F}}_{B_I\to B_O}.
\end{align}

In local tomographic OPTs, the above definition is equivalent~\cite{Eggeling2002Semicuasal, Hsieh2022, Ormrod2023} to the operational scenario shown in Fig.~\ref{fig:QSOT-Semicausality}(a). For any fixed input state $\tau\in \St(B_I)$, varying the input state at $A_I$ leaves the reduced output state at $B_O$ unchanged. If this holds for all possible input states, the channel is semicausal, i.e., $A_I\not \to B_O$:
\begin{align}
    \Tr_{A_O}[\mathcal{F}(\rho \otimes \tau)] {=}     \Tr_{A_O}[\mathcal{F}(\rho' \otimes \tau)], \ \forall \rho,\rho' {\in }\St(A_I), \tau {\in} \St(B_I).
\end{align}

Now consider the general two-time state $\mathcal{F}\star\sigma \in A_I'\otimes A_O\otimes B_O\otimes B_I'$ (where $A_I\simeq A_I'$ and $B_I\simeq B_I'$), arising from the semicausal ($A_I\not \to B_O$) channel $\mathcal{F}:{A_I\otimes B_I\to A_O\otimes B_O}$ with an initial state $\sigma\in \St(A_I\otimes B_I)$ (see Fig.~\ref{fig:QSOT-Semicausality}(b)). A natural question then arises: if we retain only the input system $A_I'$ and the output system $B_O$, discarding the remaining systems, what property does the reduced state $\Tr_{A_OB_I'}(\mathcal{F}\star\sigma)$ satisfy? This question is of particular interest because it investigates whether the causal structure imposed by semicausality manifests as a constraint on the resulting two-time state, and conversely, whether such constraints can be used to detect signaling in OPTs. We answer this in the following theorem.
\begin{ThmHighlightbox}
\begin{theorem}\label{thm:semicausal}
Let $\mathcal{F}:A_I\otimes B_I\to A_O\otimes B_O$ be a semicausal channel with the associated marginal channel $\widetilde{\mathcal{F}}:B_I\to B_O$ satisfying $\Tr_{A_O}\circ \mathcal{F} = \Tr_{A_I}\otimes \widetilde{\mathcal{F}}$, and let $\sigma\in \St(A_I\otimes B_I)$ be an arbitrary initial state. If the two-time state $\mathcal{F}\star\sigma$ satisfies Axioms~\ref{axm:PropMTQS} and \ref{axm:spatial-marginal-consistency}, then
    \begin{align}
        \Tr_{A_OB_I'}(\mathcal{F} \star \sigma)= \big[\id_{A_I} \otimes \widetilde{\mathcal{F}}\big](\sigma)
    \end{align}
    is a physical state in $\St(A_I\otimes B_O)$.
\end{theorem}
\end{ThmHighlightbox}
The proof is illustrated in Fig.~\ref{fig:QSOT-Semicausality}(b). This theorem shows that the reduced state $\Tr_{A_O B_I'} (\mathcal{F} \star \sigma)$ is always a physical state, as it is obtained by applying the marginal channel $\widetilde{\mathcal{F}}$ to one subsystem of the input state $\sigma$. Consequently, the contrapositive of Theorem~\ref{thm:semicausal} provides a sufficient condition to certify signaling: 
\begin{CorHighlightbox}
    
\begin{corollary}\label{cor:semicausality}
Under the same assumptions as Theorem~\ref{thm:semicausal}, if  $\Tr_{A_O B_I'} (\mathcal{F} \star \sigma)$ is not a physical state in the underlying OPT, then the channel $\mathcal{F}$ must allow signaling from $A_I$ to $B_O$.
\end{corollary}
\end{CorHighlightbox}

It is worth noting that, in quantum theory, a similar criterion was identified in Ref.~\cite{Liu2025} for multiqubit pseudo-density matrices. Since a pseudo-density matrix, which is a special case of the Jordan product state over time~\cite[Appendix~A]{horsman2017can}, is self-adjoint, the departure from a physical state is manifested as it being a negative operator. In our framework, we provide an axiomatic foundation for this phenomenon: any multi-time state satisfying Axioms~\ref{axm:PropMTQS} and \ref{axm:spatial-marginal-consistency} can be used to certify signaling. Such axiomatic characterization is important because it clarifies the essential physical principles underlying the link between signaling and two-time states, independent of any specific mathematical construction such as the pseudo-density matrix formalism. Further, this idea of detecting signaling through a bipartite channel can be extended to multi-time evolutions, leading to multi-time states, as well. Furthermore, as a concrete application of Theorem~\ref{thm:semicausal} to multi-time evolutions, we shall use it to certify the presence of memory effects in multi-time states in quantum theory in Section~\ref{sec:SemiCausalProcessMatrix}. We conclude this section on multi-time states in the OPT framework with the following observation.

\subsection{Non-Markov compatibility and multipartite states} \label{sec:Markvo_non-Markov_compatiblity}

The formulation of states over time naturally gives rise to the following question: in OPTs, can states shared between spacelike separated parties be equivalently represented as states over timelike separated events? If not all, which multi-space (or multipartite) states admit a causal explanation, i.e., can be represented as multi-time states? To answer this question, we make an important observation for multipartite states.

\begin{figure}
\centering
\includegraphics[width=\columnwidth]{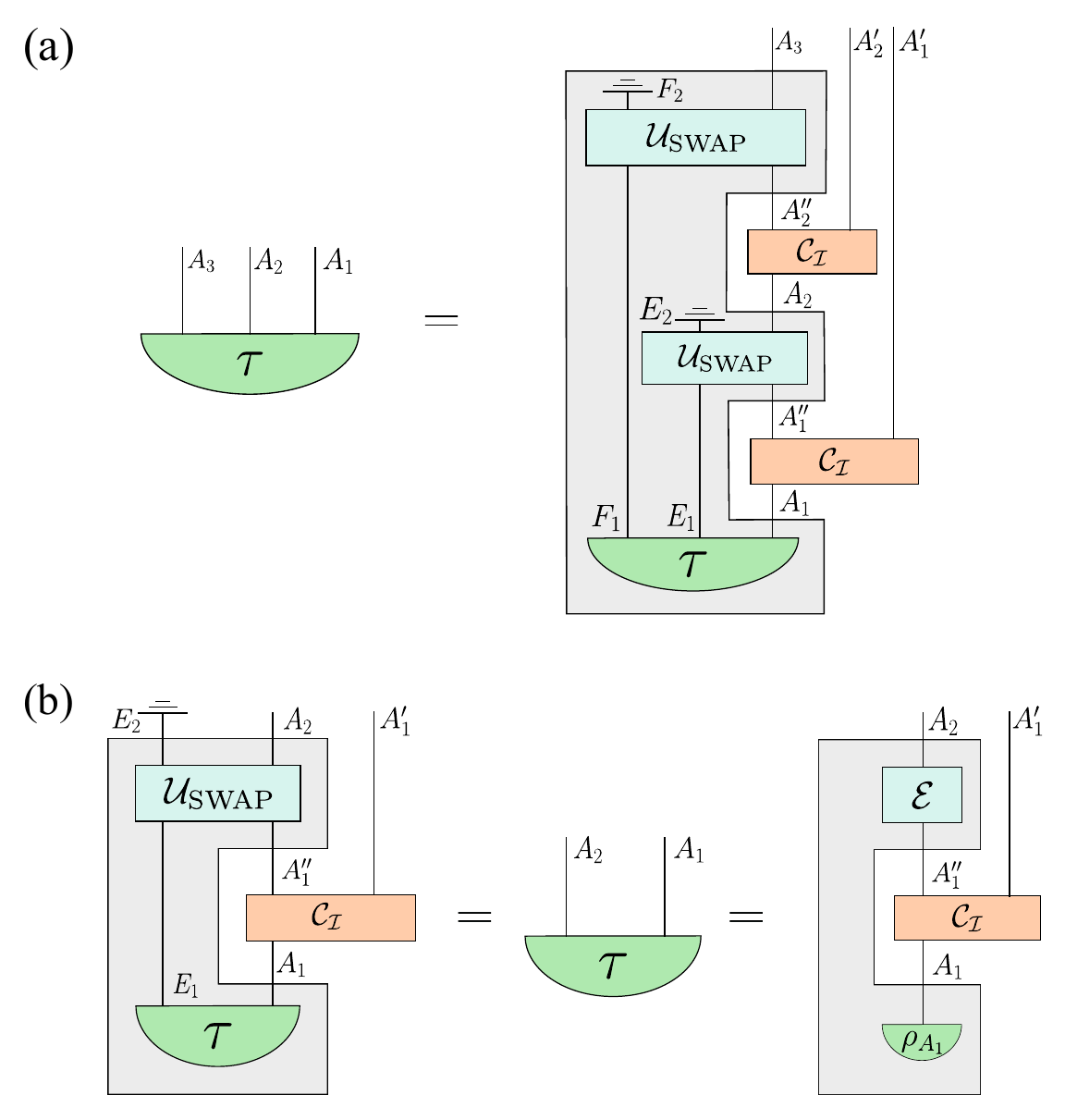}
\caption{\textbf{Multipartite states and multi-time states.} The gray-shaded regions in the figure represent higher-order maps within OPTs~\cite{Kissinger2019, Wilson2021}, which correspond to the quantum combs in the quantum theory. Here, $\mathcal{U}_{\mathrm{SWAP}}\in \Chan(E\otimes A \to A \otimes E)$ is the SWAP transformation in OPTs: $\mathcal{U}_{\mathrm{SWAP}}(\psi \otimes \phi)=\phi \otimes \psi$, for every state $\psi \in \St(E)$ and $\phi\in\St(A)$. (a) Any multipartite state can be seen as arising out of a general process matrix, as illustrated for a tripartite case in the figure above. Since this state is obtained by inserting a broadcasting map in the slots of the process matrix, it also corresponds to a multi-time state. If the initial state, $\tau_{A_1E_1F_1}$, is a correlated (i.e., a non-product) state, then the resulting dynamics is non-Markovian. (b) Now, consider bipartite states over time that are compatible with both Markovian and non-Markovian processes. As a consequence of Fig.~\ref{fig:observation1}(a), every bipartite state $\tau_{A_2A_1}$ is compatible with a non-Markovian process, as illustrated in the leftmost panel of Fig.~\ref{fig:observation1}(b). However, there exists a subset of bipartite states---including, potentially, the state $\tau_{A_2A_1}$ considered here---that also admits a Markovian realization, as illustrated in the rightmost panel of Fig.~\ref{fig:observation1}(b) and discussed in Ref.~\cite{Song2025} for quantum state over time. Therefore, there exists a nonempty class of bipartite states that are simultaneously compatible with both Markovian and non-Markovian evolutions.}
\label{fig:observation1}
\end{figure}

\begin{ObsHighlightbox}
\begin{observation}\label{obs:obs1}
Once non-Markovian processes are allowed, every multi-space state admits an equivalent representation as a general multi-time state. In this sense, every multi-space state admits a causal explanation. 
\end{observation}
\end{ObsHighlightbox}
This observation follows directly from the illustration in Fig.~\ref{fig:observation1}(a). Therefore, every multi-space state (of spacelike separated parties) is non-Markov compatible.

Interestingly, for Markovian dynamics in the quantum domain, the authors of Ref.~\cite{Song2025} derived conditions for a bipartite quantum state \(\tau_{A_2A_1}\) to be Markov compatible, i.e., compatible with a two-time state \(\mathcal{E}_{A_1\to A_2}\star \rho_{A_1}\), where \(\mathcal{E}_{A_1\to A_2}\) is a quantum channel and \(\rho_{A_1}=\Tr_{A_2}\tau_{A_2A_1}\). The set of such bipartite states also includes some entangled states. Therefore, following our Observation~\ref{obs:obs1}, all such bipartite quantum states are compatible with both Markovian and non-Markovian processes simultaneously, as illustrated in Fig.~\ref{fig:observation1}(b). 

So far, we have established an axiomatic framework for multi-time states, characterized their structure in OPTs under Markovian and non-Markovian evolutions, and demonstrated their use in certifying signaling. Building on this, we show how the framework applies to post-quantum theories. As a proof of principle, we construct unique multi-time states in boxworld theory, characterize them, and develop a notion of Bayesian retrodiction there. We then turn briefly to quantum theory, where we connect multi-time states with multi-time processes represented by process matrices, refine and augment some of the earlier results, and propose some applications.

\section{Multi-time states in the boxworld \label{sec:MTSinBoxWorld}}

In this section, we explore applications of the axiomatic framework for multi-time states (MTS) developed in the preceding sections. We first demonstrate that our framework provides a unified characterization of multi-time states across a broader class of physical theories, encompassing quantum theory and, more generally, operational probabilistic theories (OPTs). As a proof of principle, we first apply the framework to square-bit boxworld, an OPT satisfying the no-signaling principle. This demonstrates the versatility and generality of our axiomatic approach. 

\subsection{Unique state over time in the square-bit boxworld}

So far, we have claimed that our results on the structure of multi-time states hold for a broader class of physical theories, provided a suitable broadcasting map exists within the span of physical transformations. To substantiate this claim, we provide an example of a theory beyond quantum mechanics in which the required broadcasting map exists. Specifically, we construct a unique state over time in the square-bit boxworld~\cite{Popescu1994, Scarani2006, Barrett2007, Janotta2012, Plvala2023} with the maximal tensor product. We first briefly introduce the relevant features of this theory with notations mostly inspired by Ref.~\cite{Plvala2023}. 

\subsubsection{Preliminaries}

In this theory, the physical system $A$ has four pure states arranged at the corners of a square (hence the name ``square-bit''):
\begin{align}
\rho_{00} = \begin{pmatrix} 0 \\ 0 \\ 1 \end{pmatrix}, \quad
\rho_{10} = \begin{pmatrix} 1 \\ 0 \\ 1 \end{pmatrix}, \quad
\rho_{11} = \begin{pmatrix} 1 \\ 1 \\ 1 \end{pmatrix}, \quad
\rho_{01} = \begin{pmatrix} 0 \\ 1 \\ 1 \end{pmatrix}. \label{eq:squarebit-states}
\end{align}
They satisfy the linear relation $\rho_{11}= \rho_{10}+\rho_{01}-\rho_{00}$. The state space is the convex hull of these four pure states
$\St(A) = \operatorname{conv}\{\rho_{00}, \rho_{10}, \rho_{11}, \rho_{10}\}$. Any state in this hull can be parameterized by coordinates $(x,y)\in [0,1]^2$ as $\rho = (x,y,1)^T$, with one possible convex decomposition being $ \rho= (1{-}x{-}y) \rho_{00} + x \rho_{10} + y \rho_{01}$. 

The effect space in this theory is $\eff(A):=\operatorname{conv}\{\hat{0}, \Tr, e_x, \bar{e}_x, e_y, \bar{e}_y\}$ where $\hat{0}=(0,0,0)^T$ is the null-effect and
\begin{align}
&\Tr = \begin{pmatrix} 0 \\ 0 \\ 1 \end{pmatrix}, \qquad
e_x = \begin{pmatrix} 1 \\ 0 \\ 0 \end{pmatrix}, \qquad
\bar{e}_x := \Tr - e_x = \begin{pmatrix} -1 \\ 0 \\ 1 \end{pmatrix}, \nonumber \\
&e_y = \begin{pmatrix} 0 \\ 1 \\ 0 \end{pmatrix}, \qquad
\bar{e}_y:=\Tr - e_y = \begin{pmatrix} 0 \\ -1 \\ 1 \end{pmatrix}.
\end{align}
The action of an effect $e$ on a state $\rho$ is given by the standard Euclidean inner product: $e(\rho) := \langle e, \rho \rangle$. Thus, for any state $\rho = (x,y,1)^T$, we have $\Tr(\rho)=1$, $e_x(\rho)=x$, $\bar{e}_x(\rho) = 1-x$, $e_y(\rho)=y$ and $\bar{e}_y(\rho) = 1-y$. 

The reversible transformations of the square bit form the dihedral group $D_4$, generated by a rotation $R$ and a reflection $M$ with $R^4 = M^2 = \id$ and $MR = R^{-1}M$. Their action on the pure states is:
$R: \rho_{00} \to \rho_{10} \to \rho_{11} \to \rho_{01} \to \rho_{00}$,
and
$M: \rho_{00} \leftrightarrow \rho_{11}, \quad \rho_{10} \leftrightarrow \rho_{10}, \quad \rho_{01} \leftrightarrow \rho_{01}$.

With these reversible transformations, each system $A$ has a unique invariant state $\omega_A$ satisfying $\mathcal{U}(\omega_A) = \omega_A$ for all $\mathcal{U} \in D_4$. For the square bit, this is
\begin{align}
\omega &= \frac{1}{4}(\rho_{00} + \rho_{10} + \rho_{11} + \rho_{01}) \nonumber \\ 
&=\frac{1}{2}(\rho_{00} + \rho_{11}) =  \frac{1}{2}(\rho_{10} + \rho_{01}) 
= \begin{pmatrix} 1/2 \\ 1/2 \\ 1 \end{pmatrix}.\label{eq:invariant-state}
\end{align}

To define the parallel composition $A\otimes E$ of systems $A$ and $E$, we adopt the maximal tensor product~\cite{Plvala2023}. The composite effect and state spaces are then

\begin{align}
 &\eff(A \otimes E) = \operatorname{conv}\{ e \otimes f \mid e \in \eff(A),\ f \in \eff(E) \}, \label{Eq:eff_space_max_tens}\\
 &\St(A \otimes E) = \bigg\{ \tau = \sum_k a_{kl} \ \rho_k \otimes \sigma_l \mid a_{kl} \in \mathbb{R}, \ \sum_{kl} a_{kl} = 1 \nonumber \\
& \qquad \qquad  \rho_k \in \St(A), \sigma_l \in \St(E) , \ (e \otimes f)(\tau) \in [0,1], \nonumber \\
&\qquad \qquad \forall e \in \eff(A), f \in \eff(B) \bigg \}.\label{Eq:state_space_max_tens}
\end{align}

\subsubsection{Single system broadcasting map and state over time}

With this introduction in place, it is clear that the square-bit boxworld satisfies all the OPT assumptions, as laid out in Appendix.~\ref{Sec:OPT_Assumptions}. In particular, the theory is locally tomographic, and each system admits a unique deterministic effect. Consequently, if we apply Axiom~\ref{axm:PropMTQS}, Theorem~\ref{Thm:StructureThm_MarkovQSOT} holds and determines the structure of the multi-time state. The only missing ingredient is the broadcasting map, which we now characterize. 

To pick out a unique virtual broadcasting map in this theory, we adopt axioms similar to those in Ref.~\cite{parzygnat2024virtual}. 

\begin{highlightbox}
\begin{axm}[Single system broadcasting map] \label{axm:single_system_broadcasting} The broadcasting map $\mathcal{B}: \mathsf{Transf}_{\mathbb{C}}(A \to A'' \otimes A')$ should satisfy:

\begin{enumerate}
    \item \textbf{Covariance:} For every reversible transformation $\mathcal{U}$,
    \begin{align}
    \mathcal{B} \circ \mathcal{U} = (\mathcal{U} \otimes \mathcal{U}) \circ \mathcal{B}.
    \end{align}

    \item \textbf{Classical consistency:} For any classical basis $\{\rho_\alpha, \rho_\beta\}$ with distinguishing measurement $\{e, \bar{e} := \Tr - e\}$ satisfying $e(\rho_\alpha) = 1$ and $e(\rho_\beta) = 0$, the broadcasting map must satisfy
    \begin{align}
    (\mathcal{D} \otimes \mathcal{D}) \circ \mathcal{B} \circ \mathcal{D} = \mathcal{B}_{\mathrm{cl}},
    \end{align}
    where $\mathcal{D}(\rho) = e(\rho)\rho_\alpha + \bar{e}(\rho)\rho_\beta$ is the dephasing channel, and $\mathcal{B}_{\mathrm{cl}}(\rho) = e(\rho)\rho_\alpha\otimes\rho_\alpha + \bar{e}(\rho)\rho_\beta\otimes\rho_\beta$ is the corresponding classical  broadcasting map.
\end{enumerate}
\end{axm} 
\end{highlightbox}

Note, in contrast to Ref.~\cite{parzygnat2024virtual}, we do not impose permutation invariance, as we will show that the above constraints already yield a unique permutation-symmetric broadcasting map in boxworld. A few remarks are needed to clarify the classical consistency principle. Classical bits correspond precisely to unordered pairs of perfectly distinguishable pure states. Since any two of the four pure states in Eq.~\eqref{eq:squarebit-states} are perfectly distinguishable, there are six such pairs in total. These six bases decompose into two orbits under the action of the reversible transformations. The first orbit consists of the four adjacent pairs,
\begin{align}
\mathcal{O}_{\mathrm{adj}} = \{ \{\rho_{00},\rho_{10}\}, \{\rho_{10},\rho_{11}\}, \{\rho_{11},\rho_{01}\}, \{\rho_{01},\rho_{00}\} \}, \label{Eq:adj_orbit}
\end{align}
while the second orbit consists of the two diagonal pairs,
\begin{align}
\mathcal{O}_{\mathrm{diag}} = \{ \{\rho_{00},\rho_{11}\}, \{\rho_{01},\rho_{10}\} \}.\label{Eq:diag_orbit}
\end{align}
Since reversible transformations preserve the adjacency graph of the square, diagonal pairs cannot be mapped to adjacent ones. Thus, the two orbits are distinct. This differs from quantum theory, where \emph{strong symmetry}~\cite{Barnum2014_HigherOrder} holds: all sets of perfectly distinguishable pure states containing the same number of elements are equivalent under reversible transformations.

Now suppose we have two bases from the same equivalence class, $\mathcal{C} = \{\rho_\alpha, \rho_\beta\}$ and $\mathcal{C}' = \{\mathcal{U}(\rho_\alpha), \mathcal{U}(\rho_\beta)\}$. In this scenario, both the dephasing channel and the classical broadcasting map undergo the corresponding reference frame transformation:
\begin{align}
&\mathcal{D}' = \mathcal{U} \circ \mathcal{D} \circ \mathcal{U}^{-1}, \\
&\mathcal{B}'_{\mathrm{cl}} = (\mathcal{U} \otimes \mathcal{U}) \circ \mathcal{B}_{\mathrm{cl}} \circ \mathcal{U}^{-1}.
\end{align}
To see this, note that the distinguishing effect transforms as $e' = e \circ \mathcal{U}^{-1}$, i.e., $e'(\rho) = e(\mathcal{U}^{-1}(\rho))$. Then,
\begin{align}
\mathcal{D}'(\rho) &= e'(\rho) \mathcal{U}(\rho_\alpha) + \bar{e}'(\rho) \mathcal{U}(\rho_\beta)
= \mathcal{U} [\mathcal{D}(\mathcal{U}^{-1}(\rho)], \\
\mathcal{B}'_{\mathrm{cl}}(\rho) &= e'(\rho) \mathcal{U}(\rho_\alpha) \otimes \mathcal{U}(\rho_\alpha) + \bar{e}'(\rho) \mathcal{U}(\rho_\beta) \otimes \mathcal{U}(\rho_\beta) \nonumber \\
&= (\mathcal{U} \otimes \mathcal{U}) [\mathcal{B}_{\mathrm{cl}}(\mathcal{U}^{-1}\rho)].
\end{align}
Now assume that the broadcast map $\mathcal{B}$ satisfies classical consistency in the basis $\mathcal{C}$. Then,
\begin{align}
&(\mathcal{D}' \otimes \mathcal{D}') \circ \mathcal{B} \circ \mathcal{D}' \nonumber \\
&= (\mathcal{U} \otimes \mathcal{U})\circ (\mathcal{D} \otimes \mathcal{D})\circ (\mathcal{U}^{-1} \otimes \mathcal{U}^{-1}) \circ \mathcal{B} \circ \mathcal{U}\circ \mathcal{D}\circ \mathcal{U}^{-1} \nonumber \\
&= (\mathcal{U} \otimes \mathcal{U})\circ(\mathcal{D} \otimes \mathcal{D})\circ \mathcal{B} \circ \mathcal{D} \circ\mathcal{U}^{-1} \nonumber \\
&= (\mathcal{U} \otimes \mathcal{U}) \circ\mathcal{B}_{\mathrm{cl}} \circ \mathcal{U}^{-1}
= \mathcal{B}'_{\mathrm{cl}}.
\end{align}
where we used covariance in the third equality and classical consistency in the fourth.  Hence, classical consistency is preserved when moving between bases in the same equivalence class. Furthermore, if $\mathcal{B}$ is unique according to Axiom~\ref{axm:single_system_broadcasting} with classical consistency in a basis $\mathcal{C}$, then it is unique for any basis $\mathcal{C}'$ in the same equivalence class. Indeed, suppose there existed another broadcasting map $\mathcal{B}' \neq \mathcal{B}$ satisfying the same axioms for $\mathcal{C}'$. By the reference-frame transformation, this would yield a distinct broadcasting map for the original basis $\mathcal{C}$, contradicting the assumed uniqueness of $\mathcal{B}$. Thus, we establish the following proposition.

\begin{PropHighlightbox}
\begin{proposition} \label{Prop:classical_consistency_equivalence_class}
Classical consistency is preserved within the same equivalence class of perfectly distinguishable states. Consequently, if a broadcasting map is uniquely determined by Axiom~\ref{axm:single_system_broadcasting} together with classical consistency in one basis $\mathcal{C}$, then it is uniquely determined for any basis $\mathcal{C}'$ in the same equivalence class.
\end{proposition}
\end{PropHighlightbox}

Thus, in principle, we can pick just one representative basis from each equivalence class and look for a unique map. A subtlety arises for the diagonal representative $\{\rho_{01}, \rho_{10}\}$. Unlike the adjacent case, each state in this pair can be distinguished by two different measurements: $\{e_x, \bar{e}_x\}$ and $\{e_y, \bar{e}_y\}$. This results in two distinct dephasing channels,
\begin{align}
\mathcal{D}_x(\rho) &= e_x(\rho)\rho_{10} + \bar{e}_x(\rho)\rho_{01}, \\
\mathcal{D}_y(\rho) &= e_y(\rho)\rho_{01} + \bar{e}_y(\rho)\rho_{10}.
\end{align}
However, they are related by the reflection $M$, which fixes both $\rho_{01}$ and $\rho_{10}$ while exchanging the distinguishing effects: $e_x \circ M^{-1} = \bar{e}_y$. Consequently, $\mathcal{D}_y = M \circ \mathcal{D}_x \circ M^{-1}$, and similarly for the corresponding classical broadcasting maps. Thus, by the same argument as Proposition~\ref{Prop:classical_consistency_equivalence_class}, it suffices to check classical consistency for only one of the two dephasing channels. With this, we have the following theorem.

\begin{ThmHighlightbox}
\begin{theorem}\label{thm:unique-broadcasting-map}
Let $v = \rho_{01} - \rho_{10}$ and $u = \rho_{00} - \rho_{11}$, and let $\omega = \frac{1}{2}(\rho_{00} + \rho_{11}) = \frac{1}{2}(\rho_{10} + \rho_{01})$ be the unique invariant state as in Eq.~\eqref{eq:invariant-state}. Then the unique map $\mathcal{B} \in \mathsf{Transf}_{\mathbb{C}}(A \to A'' \otimes A')$ satisfying Axiom~\ref{axm:single_system_broadcasting} is given by
\begin{align}
\mathcal{B}(\rho) {=} \omega \otimes \rho {+} \rho \otimes \omega {-}  \Tr (\rho)  \left( \omega \otimes \omega {-} \frac{1}{8}(v \otimes v + u \otimes u) \right).
\end{align}
Crucially, the broadcasting map is unique and does not depend on whether the classical consistency conditions are imposed on an adjacent or a diagonal basis.
\end{theorem}
\end{ThmHighlightbox}

A brief proof strategy is as follows. The vector space 
$\mathsf{St}_{\mathbb{C}}(A) = \operatorname{Span}_{\mathbb{C}}\{\rho_{00}, \rho_{10}, \rho_{01}\}$ 
decomposes in terms of the irreducible subrepresentations
$\mathsf{St}_{\mathbb{C}}(A) \cong \mathbf{1} \oplus E$, 
where $\mathbf{1} := \operatorname{Span}_{\mathbb{C}}\{\omega\}$ and $E := \operatorname{Span}_{\mathbb{C}}\{v,u\}$ 
are subrepresentations of the dihedral group $D_4$, with action given by
\begin{align*}
R(\omega) &= \omega, & R(u) &= -v & R(v)&=u, \\
M(\omega) &= \omega,  & M(u) &= -u, & M(v) &= v.
\end{align*}
Subsequently, identifying the broadcasting map 
$\mathcal{B} \in \mathsf{Transf}_{\mathbb{C}}(A \to A'' \otimes A')$ 
with a vector 
$X_{\mathcal{B}} \in \mathsf{St}_{\mathbb{C}}(A)^* \otimes \mathsf{St}_{\mathbb{C}}(A'') \otimes \mathsf{St}_{\mathbb{C}}(A')$, 
where $\mathsf{St}_{\mathbb{C}}(A)^*$ is the dual space of $\mathsf{St}_{\mathbb{C}}(A)$, 
the covariance condition of $\mathcal{B}$ is equivalent to 
$({R^{-1}}^* \otimes R \otimes R)(X_{\mathcal{B}}) = ({M^{-1}}^* \otimes M \otimes M)(X_{\mathcal{B}}) = X_{\mathcal{B}}$, 
with $(R^{-1})^*$ and $(M^{-1})^*$ being the dual actions on $\mathsf{St}_{\mathbb{C}}(A)^*$.  
Hence, the task is to characterize the invariant vectors in 
$(\mathbf{1} \oplus E)^* \otimes (\mathbf{1} \oplus E) \otimes (\mathbf{1} \oplus E)$. 
Expanding this tensor product, it turns out, only the components 
$\mathbf{1}^* \otimes \mathbf{1} \otimes \mathbf{1}$, 
$\mathbf{1}^* \otimes E \otimes E$, 
$E^* \otimes E \otimes \mathbf{1}$, and 
$E^* \otimes \mathbf{1} \otimes E$ contain non-zero invariant vectors; all other combinations are eliminated by the group action. 
Transforming the resulting invariant vector back to an element of 
$\mathsf{Transf}_{\mathbb{C}}(A \to A'' \otimes A')$ 
and applying the remaining constraint from Axiom~\ref{axm:single_system_broadcasting} 
then uniquely determines the broadcasting map in Theorem~\ref{thm:unique-broadcasting-map}. See Appendix~\ref{App:unique-broadcasting-map} for the detailed proof.

 This allows us to propose the following unique star product in the square-bit boxworld: for any channel $\mathcal{E} \in \Chan(A {\to} B)$ and initial state $\rho \in \St(A)$,
\begin{align}
&\mathcal{E} \star \rho  = (\mathcal{E} \otimes \id)\circ  \mathcal{B}(\rho) \nonumber \\ 
&= \mathcal{E}(\omega) \otimes (\rho - \omega) {+} \mathcal{E}(\rho) \otimes \omega {+} \frac{1}{8} \big( \mathcal{E}(v) \otimes v + \mathcal{E}(u) \otimes u \big) .\label{Eq:Boxworld_state_over_time}
\end{align}
The corresponding multi-time state can be straightforwardly formulated using Theorem~\ref {Thm:StructureThm_MarkovQSOT}.

\subsubsection{Bayesian retrodiction in the boxworld}

One of the primary motivations for the state-over-time formalism is to provide a rigorous framework for Bayesian retrodiction beyond standard classical theory. While extensions of Bayesian retrodiction have been extensively explored in quantum theory, our formulation of states-over-time goes beyond the quantum setting, thereby broadening the scope of such retrodiction schemes. In particular, we demonstrate the method of Bayesian retrodiction in boxworld as follows.

For a state $\rho \in \mathsf{St}(A)$ and a channel $\mathcal{E} \in \mathsf{Chan}(A \to B)$, the construction $\mathcal{E} \star \rho$ in Eq.~\eqref{Eq:Boxworld_state_over_time} allows us to evaluate the Bayesian retrodiction map $\mathcal{R}_{(\mathcal{E}, \rho)} \in \mathsf{Transf}_{\mathbb{C}}(B \to A)$ by solving
\begin{align}
&\mathcal{R}_{(\mathcal{E}, \rho)} \star \E(\rho) = \E \star \rho \\
&\implies (\id \otimes \mathcal{R}_{(\mathcal{E}, \rho)}) \circ \mathcal{B}[\mathcal{E}(\rho)]
= (\mathcal{E} \otimes \id) \circ \mathcal{B}(\rho). \label{Eq:Retrodiction_definition}
\end{align}

To solve it, note that since $\{\omega,v,u\}$ forms a basis of the $\operatorname{St}_{\mathbb{C}}(A)$, the linear map $\mathcal{R}_{(\mathcal{E},\rho)}$ is completely characterized by its action on these three vectors. 

Let $\{\omega^*,v^*,u^*\}$ be the dual basis of $\operatorname{Eff}_{\mathbb{C}}(A)$ satisfying
\begin{align}
&\omega^*(\omega)=v^*(v)=u^*(u)=1, \\
&\omega^*(v)=\omega^*(u)=v^*(\omega)=u^*(\omega)=v^*(u)=u^*(v)=0,
\end{align}
where $\omega^*=\operatorname{Tr}$ is the unit effect and the remaining dual functionals are given explicitly by
\begin{equation}
v^*=\frac12(e_y-e_x),\qquad
u^*=\frac12(\bar{e}_x-e_y).
\end{equation}

Expanding both sides of Eq.~\eqref{Eq:Retrodiction_definition} with the concrete form of $\mathcal{B}$ as in Theorem~\ref{thm:unique-broadcasting-map} and applying the dual functionals $\omega^*_B\otimes\id_A\equiv \Tr_B \otimes \id_A$, $v^*_B\otimes\id_A$ and $u^*_B\otimes \id_A$ in turn yields the three relations

\begin{align}
&\mathcal{R}_{(\mathcal{E},\rho)}\bigl[\mathcal{E}(\rho)\bigr]
=\rho,
\label{eq:R-trace}
\\[4pt]
&v^*[\mathcal{E}(\rho)]\,
\mathcal{R}_{(\mathcal{E},\rho)}(\omega)
+\frac18\mathcal{R}_{(\mathcal{E},\rho)}(v)
\nonumber \\
&{=}
v^*[\mathcal{E}(\omega)](\rho{-}\omega)
{+}v^*[\mathcal{E}(\rho)]\,\omega
{+}\frac18\Big(
v^*[\mathcal{E}(v)]\,v
{+}v^*[\mathcal{E}(u)]\,u
\Big),
\label{eq:R-v}
\\[4pt]
&u^*[\mathcal{E}(\rho)]\,
\mathcal{R}_{(\mathcal{E},\rho)}(\omega)
+\frac18\mathcal{R}_{(\mathcal{E},\rho)}(u)
\nonumber \\
&{=}
u^*[\mathcal{E}(\omega)](\rho{-}\omega)
{+}u^*[\mathcal{E}(\rho)]\,\omega
{+}\frac18\Big(
u^*[\mathcal{E}(v)]\,v
{+}u^*[\mathcal{E}(u)]\,u
\Big).
\label{eq:R-u}
\end{align}
Because $\mathcal{E}$ is a channel, it is trace preserving. So, its action on any state admits the expansion
\begin{equation}
\mathcal{E}(\rho)=\omega+v^*(\mathcal{E}(\rho))\,v+u^*(\mathcal{E}(\rho))\,u.
\end{equation}
Consequently,
\begin{align}
\mathcal{R}_{(\mathcal{E},\rho)}[\mathcal{E}(\rho)]
{=}
\mathcal{R}_{(\mathcal{E},\rho)}(\omega)
{+}v^*[\mathcal{E}(\rho)]\,\mathcal{R}_{(\mathcal{E},\rho)}(v)
{+}u^*[\mathcal{E}(\rho)]\,\mathcal{R}_{(\mathcal{E},\rho)}(u).
\end{align}
Substituting this expression into \eqref{eq:R-trace} yields
\begin{align}
\mathcal{R}_{(\mathcal{E},\rho)}(\omega)
=
\rho
-v^*[\mathcal{E}(\rho)]\,\mathcal{R}_{(\mathcal{E},\rho)}(v)
-u^*[\mathcal{E}(\rho)]\,\mathcal{R}_{(\mathcal{E},\rho)}(u).
\end{align}
The remaining two equations \eqref{eq:R-v} and \eqref{eq:R-u} then determine $\mathcal{R}_{(\mathcal{E},\rho)}(v)$ and $\mathcal{R}_{(\mathcal{E},\rho)}(u)$ uniquely, thereby completing the reconstruction of the retrodiction map. We present the explicit solution in Appendix~\ref{app:retrodiction_solution}. With these solutions, the action of the retrodiction map on an arbitrary $\tau \in \St_{\mathbb{C}}(B)$ is given by
\begin{align}
\mathcal{R}_{(\E,\rho)}(\tau)
{=}
\operatorname{Tr}(\tau)\,\mathcal{R}_{(\E,\rho)}(\omega)
{+}
v^*(\tau)\,\mathcal{R}_{(\E,\rho)}(v)
{+}
u^*(\tau)\,\mathcal{R}_{(\E,\rho)}(u).
\end{align}

While the general solution for the retrodiction map is rather bulky, some special cases provide useful intuition. First, consider the constant channel in boxworld, defined by
\begin{align}
\mathcal{E}(\rho)=\operatorname{Tr}(\rho)\,\sigma \qquad \forall \rho \in \mathsf{St}_{\mathbb{C}}(A).
\end{align}
That is, this channel discards all information about the input state and outputs the fixed state $\sigma \in \St(B)$. For this channel, the retrodiction map becomes (see Appendix~\ref{App:constant_chan_retro} for the proof)
\begin{align}
\mathcal{R}_{(\mathcal{E},\rho)}(\omega)=\rho, \qquad
\mathcal{R}_{(\mathcal{E},\rho)}(v)=0, \qquad
\mathcal{R}_{(\mathcal{E},\rho)}(u)=0.
\end{align}
Consequently, for any arbitrary state $\tau \in \mathsf{St}(B)$,
\begin{align}
\mathcal{R}_{(\mathcal{E},\rho)}(\tau)
=
\operatorname{Tr}(\tau)\,\rho.
\end{align}
Thus, the retrodiction map simply discards the input and prepares the prior state $\rho$. This is exactly the Bayesian update rule: when the channel is completely information-degrading, the observation provides no information to update the prior.

As a second example, consider a reversible transformation $\mathcal{U} \in \mathsf{Chan}(A \to B)$. There are two ways to obtain the retrodiction map. A conceptual, theory-independent way is to use the covariance of the broadcasting map. Since $\mathcal{B}$ is covariant under reversible transformations, we have
\begin{align}
(\mathcal{U} \otimes \id) \mathcal{B}(\rho)
&=
(\id \otimes \mathcal{U}^{-1})
(\mathcal{U} \otimes \mathcal{U}) \mathcal{B}(\rho)
\nonumber\\
&=
(\id \otimes \mathcal{U}^{-1})
\mathcal{B}(\mathcal{U}(\rho)).
\end{align}
Comparing this with the defining relation for the retrodiction map, namely
\begin{align}
(\id \otimes \mathcal{R}_{(\mathcal{U}, \rho)}) \mathcal{B}(\mathcal{U}(\rho))
=
(\mathcal{U} \otimes \id) \mathcal{B}(\rho),
\end{align}
This indicates $\mathcal{R}_{(\mathcal{U}, \rho)} = \mathcal{U}^{-1}$ is one solution for a fixed $\rho$. Whether this is the unique solution depends on the specific structure of the multi-time state. To show the uniqueness in the boxworld scenario, we provide a more systematic derivation in Appendix~\ref{app:reversible_retrodiction}. This result is conceptually consistent: if the evolution from $A$ to $B$ is reversible, then no information is lost, and the retrodicted state should be obtained by applying the inverse transformation directly, independent of any prior state $\rho$.

\subsubsection{Boxworld conditionability}

It is worth noting that the state-over-time construction in boxworld also satisfies a notion of conditionability similar to what was introduced in Refs.~\cite {parzygnat2023time,lie2024quantum,lie2025multipartite} for the quantum theory. Specifically, there exists a state-rendering map $\Theta_\rho \in \mathsf{Transf}_{\mathbb{C}}(A \to A)$ such that for all $\mathcal{E} \in \mathsf{Chan}(A \to B)$ and a fixed $\rho \in \mathsf{St}(A)$,
\begin{align}
\mathcal{E} \star \rho = (\id \otimes \Theta_\rho) \circ (\mathcal{E} \star \omega) = (\mathcal{E} \otimes \Theta_\rho) \mathcal{B}[\omega].
\label{Eq:boxworld_conditionability_def}
\end{align}
In particular, choosing $\mathcal{E} = \id$ gives
\begin{align}
(\id \otimes \Theta_\rho) \circ \mathcal{B}(\omega) = \mathcal{B}(\rho).
\label{Eq:boxworld_conditionability}
\end{align}
One can solve for this state-rendering map, showing that for any state $\sigma \in \mathsf{St}(A)$, the map acts as
\begin{align}
\Theta_\rho(\sigma) = \sigma + \rho + \bigl[8 \big(v^*(\sigma) v^*(\rho) + u^*(\sigma) u^*(\rho)\big) - 1\bigr] \omega.
\end{align}

Since a thorough exploration of this conditionability property, and its interpretation in boxworld is beyond the scope of the present work, we leave a detailed investigation of its consequences, including a comparison with the quantum case, for future research.

\subsubsection{Joint broadcasting map}

We now characterize the broadcasting map for the parallel composition of two square-bit systems $A$ and $E$. We first extend
Axiom~\ref{axm:single_system_broadcasting} as:

\begin{highlightbox}
\begin{axm}[Bipartite broadcasting map] \label{axm:joint_broadcasting} The broadcasting map $\mathcal{B} \in \mathsf{Transf}_{\mathbb{C}}[A \otimes E \to (A'' \otimes A' )  \otimes (E'' \otimes E')]$  should satisfy:

\begin{enumerate}
\item \textbf{Covariance:} For every reversible transformation $\mathcal{U}\in \Chan(A\otimes E \to A\otimes E)$,
\begin{align}
\mathcal{B} \circ \mathcal{U} = (\mathcal{U} \otimes \mathcal{U}) \circ \mathcal{B}.
\end{align}
Here, we implicitly reorder the tensor factors to match the output ordering of $\mathcal{B}$ (i.e., $A'' \otimes A' \otimes E'' \otimes E'$ rather than $A''\otimes E''\otimes A'\otimes E'$), keeping track of the systems throughout.
  
\item \textbf{Classical consistency:} For any classical basis $\{\rho_{ij} \in \St(AE)\}_{\ i,j \in \{0,1\}}$ with distinguishing measurement $\{e_{ij}\in \eff(AE) \}_{\ i,j \in \{0,1\}}$ satisfying $\sum_{i,j} e_{ij}=\Tr_{AE}=\Tr_{A}\otimes \Tr_E$ and $e_{ij}(\rho_{kl}) = \delta_{i,k} \cdot \delta_{j,l}$, the broadcasting map must satisfy
\begin{align}
(\mathcal{D} \otimes \mathcal{D}) \circ \mathcal{B} \circ \mathcal{D} = \mathcal{B}_{\mathrm{cl}},
\end{align}
where $\mathcal{D}(\rho) = \sum_{i,j} e_{ij} (\rho)\rho_{ij}$ is the dephasing channel, and $\mathcal{B}_{\mathrm{cl}}(\rho) = \sum_{i,j} e_{ij}(\rho)\ \rho_{ij}\otimes\rho_{ij}$ is the corresponding classical broadcasing map.
\end{enumerate}
\end{axm} 
\end{highlightbox}
Note, although the axiom requires covariance under all joint reversible maps, in the boxworld the joint reversible transformations are highly restricted~\cite{Gross2010}. It is known that joint reversible transformations consist only of local reversible maps and permutations of the subsystems. This structure significantly simplifies the characterization of the joint broadcasting map.

Next, regarding classical consistency, one might wonder whether entangled states could serve as elements of the classical basis. In the maximal tensor product of two square bits, there are 8 pure entangled states (the PR-box states that maximally violate the CHSH inequality) in addition to 16 product states. Due to the restricted structure of joint reversible transformations, which cannot map product states to entangled states, each entangled state belongs to a different equivalence class from the product states. This contrasts with quantum theory, where global unitary transformations can map product states to entangled ones. Consequently, the reference-frame argument in Proposition~\ref{Prop:classical_consistency_equivalence_class} does not apply to any potential classical basis that contains an entangled state.

However, the restricted structure of the joint effect space in the maximal tensor product (Eq.~\eqref{Eq:eff_space_max_tens}) rules out entangled states as classical basis elements. To see this, suppose a pure entangled state $\psi\in \St(AE)$ were in a classical basis. Then some global effect $g = \sum_\lambda p_\lambda e_\lambda \otimes f_\lambda$ would satisfy $g(\psi) = 1$. By convexity, there exists an extremal effect $e \otimes f$ with $(e \otimes f)(\psi) = 1$. Taking the complementary effect $\bar{f}_E:=\Tr_E - f_E$, we have $(e \otimes \bar{f})(\psi) \ge 0$, which implies $e(\Tr_E(\psi)) \ge 1$. Since $e$ is an effect, this forces $e(\Tr_E(\psi)) = 1$. But for any PR-box state, the reduced state on either subsystem is the invariant state, thus we have $\Tr_E(\psi) = \omega_A$, and no extremal effect satisfies $e(\omega_A) = 1$. This contradiction shows that a PR-box cannot be part of a classical basis. Thus, only product states $\rho_{ij} = \rho_i \otimes \rho_j$, with $\{\rho_i\}$ and $\{\rho_j\}$ chosen from either of the equivalence classes in Eqs.~\eqref{Eq:diag_orbit}, and \eqref{Eq:adj_orbit}, can belong to a classical basis. Accordingly, we only need to consider product dephasing channels $\mathcal{D} = \mathcal{D}_\alpha \otimes \mathcal{D}_\beta$ and product classical broadcasters $\mathcal{B}_{\mathrm{cl}} = \mathcal{B}_{\mathrm{cl},\alpha} \otimes \mathcal{B}_{\mathrm{cl},\beta}$, where $\alpha, \beta \in \{\text{adj}, \text{diag}\}$.

With this, we have the following result.
\begin{ThmHighlightbox}
\begin{theorem} \label{Thm:two_party_broadcasting}
    If $A$ and $E$ are two square bit systems with $A\otimes E$ being the maximal tensor product, then the joint broadcasting map $\mathcal{B}: \mathsf{Transf}_{\mathbb{C}}[A \otimes E \to (A' \otimes A'') \otimes (E'' \otimes E')]$ is simply the product of single system broadcasting maps $\mathcal{B}_{A}\in \mathsf{Transf}_{\mathbb{C}}(A \to A'' \otimes A')$ and $\mathcal{B}_E \in  \mathsf{Transf}_{\mathbb{C}}(E \to E'' \otimes E')$ as in Theorem~\ref{thm:unique-broadcasting-map}:
    \begin{align}
        \mathcal{B} = \mathcal{B}_{A}\otimes \mathcal{B}_E.
    \end{align}
\end{theorem}
\end{ThmHighlightbox}
See the proof in Appendix~\ref{App:proof_two_party_broadcasting}. Subsequently, due to Theorem~\ref{Thm:StructureThm_MarkovQSOT}, for any state $\sigma\in \St(A_1E_1)$ and channel $\mathcal
{F}\in \Chan(A_1\otimes E_1\to A_2\otimes E_2)$ we have 
\begin{align}
    \mathcal{F}\star \sigma = (\id_{A_1'} \otimes \mathcal{F}_{A_1''E_1''\to A_2E_2} \otimes \id_{E'_1} )\circ (\mathcal{B}_{A_1}\otimes B_{E_1})(\sigma).
\end{align}
In particular, for all product channels $\mathcal{F}=\E \otimes \N$ with $\E:A_1\to A_2$ and $\N:E_1\to E_2$, and all product states $\sigma_{A_1E_1}=\rho_{A_1}\otimes \tau_{E_1}$ we have
\begin{align}
    (\E\otimes \N)\star (\rho\otimes \tau)=(\E \star \rho)\otimes (\N\star \tau).
\end{align}
Thus Axiom~\ref{axm:spatial-marginal-consistency} naturally holds, and all subsequent results--- Theorem~\ref{Thm:NMQSOT}, Theorem~\ref{thm:semicausal}, Corollary~\ref{cor:semicausality}---hold in square-bit boxworld with maximal tensor product.

Clearly, the axiomatic framework developed in the preceding sections can be applied to the study of multi-time states in boxworld, an OPT, allowing us to characterize the unique state over time associated with general evolutions, including those involving temporal memory. We now turn to quantum theory and apply the same axiomatic framework to characterize multi-time states in the quantum setting.

\section{Multi-Time States in Quantum Theory and Quantum Process Framework}\label{sec:MTSinQuantumTheory}

In the recent past, many different candidates for a quantum state over time have been proposed, such as Ohya compound state~\cite{ohya1983note}, Leifer-Spekkens (LS) quantum state over time~\cite{leifer2013towards}, pseudo-density matrix (PDM)~\cite{fitzsimons2015quantum}, Sutter-Tomamichel-Harrow (STH) state~\cite{sutter2016strengthened,Parzygnat2023}, Fullwood-Parzygnat (FP) quantum state over time~\cite{fullwood2022quantum}. While these constructions (see Ref.~\cite{parzygnat2023time} for a review) find uses in various areas~\cite{leifer2006quantum, parzygnat2023time, Buscemi_2023, bai2024observational,  fitzsimons2015quantum, comar2026using, Song2025}, they are not mutually compatible, thereby motivating the need for a unified framework. In this direction, a set of desirable properties for quantum states over time was formulated in~\cite{horsman2017can} in the form of axioms. There, a valid two-time quantum state is defined as a function of the initial state and the quantum channel, required to satisfy the following conditions: (a) it should produce the correct marginals, (b) it should be Hermitian, (c) it should be linear in both the state and the channel, (d) it should admit an extension to multi-time scenarios, i.e., satisfy compositionality, and (e) it should reproduce the corresponding classical probability distributions under appropriate conditions. The authors of~\cite{horsman2017can} also established a no-go theorem showing that no valid quantum state over time satisfying all these conditions can exist. This issue was later resolved in~\cite{fullwood2022quantum} by demonstrating that a quantum state over time can indeed be consistently assigned to a dynamical process by carefully formulating the axioms over an appropriate domain. However, whether this construction was the unique quantum state over time satisfying these axioms remained an open question. The uniqueness was subsequently established later in~\cite{parzygnat2024virtual,lie2024quantum}. The next natural step concerned the existence of a multi-time extension of the two-time quantum state, namely, a joint quantum state spanning multiple timelike events. A unique multi-time extension of the two-time quantum state corresponding to a Markovian evolution was obtained in~\cite{lie2025multipartite}. However, a systematic characterization of multi-time states for general quantum evolutions, including non-Markovian dynamics, their connections to the process-matrix framework, and their ability to capture signaling and causal properties, is still missing. We now apply our axiomatic framework to address these gaps related to multi-time states in quantum theory. This allows us to recover and refine several existing results, while also elucidating how multi-time states encode the dynamical, temporal, and causal properties of quantum processes. 

To describe general system-environment dynamics, we adopt the framework of a multi-time process, also known as a quantum comb ~\cite{chiribella09b}, process tensors~\cite{Pollock2018Operational, Pollock2018}, or process matrix~\cite{Giarmatzi_2021, Goswami2021,Goswami2025}. Note, the quantum multi-time process framework can be generalized to OPTs~\cite{Kissinger2019, Wilson2021}. Nevertheless, in quantum theory, the process framework provides a complete description of all possible correlations that can arise when a system is probed at multiple times while interacting with an environment. This includes \emph{non-Markovian dynamics}, in which the environment retains a memory of past interactions that influence the system's future evolution. Here, we aim to explore these features by explicitly employing quantum process matrices to characterize quantum multi-time states. We begin by providing a brief overview of the multi-time process framework.

\begin{figure}
    \centering
    \includegraphics[width=\columnwidth]{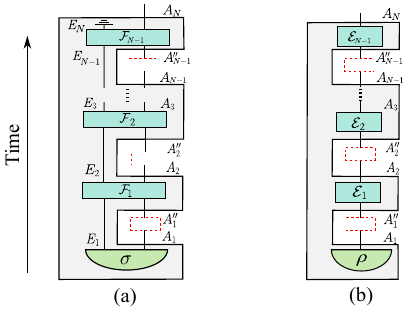}
\caption{\textbf{Multi-time processes.} (a) A general multi-time process, potentially exhibiting non-Markovian dynamics. An initially correlated system-environment state $\sigma_{A_1E_1}$ evolves through a sequence of system-environment interactions: $\mathcal{F}_n:A''_{n}E_n\to A_{n+1}E_{n+1}$ for $n=1,2,\cdots,N-1$, with the final environment $E_N$ being discarded. The orange open slots in both figures represent the probe regions where the experimentalist can perform their quantum operations. The overall dynamics is represented by the process matrix shown in Eq.~\eqref{Eq:non-MarkovianProcess}. (b) A Markovian process, a special case of (a), described by an initial state $\rho_{A_1}$ and a sequence of system-only channels $\E_n : {A_n''\to A_{n+1}}$  for $n=1,2,\cdots,N-1$. In this case, the process matrix factorizes into a product form as in Eq.~\eqref{Eq:MarkovianProcess}.  }  
    \label{fig:markovNonmarkov_N_step}
\end{figure}

In a general multi-time process, as illustrated in Fig.~\ref{fig:markovNonmarkov_N_step}(a), a system and its environment, possibly initially correlated, undergo joint evolution at multiple time steps. Between these system–environment interactions, an experimentalist can probe the system at $N$ well-defined times. Each probe region $n$ (dashed orange slots) consists of an input system $A_n$ and an output system $A_n''$. We assume $A_n\simeq A_n''$ for all $n\in \{1,\cdots,N\}$. The process begins with a quantum state $\sigma_{A_1E_1}$, followed by alternating probe regions (input $A_n$, output $A''_n$) and joint system-environment interaction $\mathcal{F}_{n} \in \Chan( A_n''\otimes E_n\to A_{n+1}\otimes E_{n+1})$ for $n=1,\cdots,N-1$, with the final environment $E_N$ traced out. In this way, the overall dynamics are captured entirely from the system's perspective, with the experimentalist having access only to the designated probe regions.

Mathematically, the multi-time quantum process is represented by a positive semidefinite operator $W \in A_N \otimes A_{N-1}'' \otimes A_{N-1}\otimes \cdots \otimes A_1'' \otimes A_1$ called \emph{process matrix} with 
\begin{align}
    W {=} \Tr_{E_N} \big [\Choi (\mathcal{F}_{N-1}) * \cdots  \Choi(\mathcal{F}_{2}) * \Choi(\mathcal{F}_1) * \sigma\ \big ]. \label{Eq:non-MarkovianProcess}
\end{align}
Here, $\Choi(\mathcal{F}_n) \in A_{n+1}\otimes E_{n+1} \otimes A_n''\otimes E_n$ denotes the Choi representation~\cite{choi_completely_1975,choi75} of the quantum channel $\mathcal{F}_{n}:\ A_n'' E_n\to A_{n+1} E_{n+1}$ where  $\Choi(\mathcal{F}_n):= \sum_{i,j}\mathcal{F}_{n}(\ketbra{i}{j}) \otimes \ketbra{i}{j}$ and $\{\ket{i}\}_i$ being an orthonormal basis of $A_n''E_n$ system. The operation $*$ is the \emph{link product}~\cite{chiribella09b}. It provides a way to combine different quantum operations, e.g., state preparation, evolution through a quantum channel, or measurements, using their Choi representations. For instance, let $A, B$ and $C$ denote the Hilbert spaces associated with three systems and consider two quantum operations with Choi representations $P \in A \otimes B$ and $Q\in B \otimes C$. Their link product is defined as
\begin{align}
 P*Q := \Tr_{B} [(P^{T_B}\otimes \Id _C)\cdot(\Id _A \otimes Q)],   
\end{align}
where $T_B$ represents the partial transpose on the $B$ system. When $P$ and $Q$ have no overlapping system ($B\equiv \emptyset$), then the link product between the operators reduces to a tensor product: $P*Q=P\otimes Q$. 

Suppose an experimentalist intervenes in this process by applying a sequence of quantum instruments $\{e_{x_n}\}_{x_n \in X_n}$ for $n = 1, \dots, N-1$, where each $e_{x_n} \in \mathsf{Transf}(A_n \to A_n'')$ is a CPTNI map associated with a classical outcome $x_n$, and $\sum_{x_n} e_{x_n}$ is a CPTP map. Finally, at the last step, they perform a Positive Operator Valued Measure (POVM) $\{e_{x_N}\}_{x_N \in X_N}$. The probability of the outcome sequence $(x_1, \dots, x_N)$ is then given by the generalized Born rule~\cite{Shrapnel2018}
\begin{align}
p(x_1, \dots, x_N) = W * \bigotimes_{n=1}^N \Choi(e_{x_n}). 
\label{Eq:gen_Born_rule}
\end{align}

The most general multi-time process includes dynamics with memory, also known as \emph{non-Markovian processes}. Here, ``memory'' refers to the situation where information about the system that leaks into the environment during one interaction can later return and influence the system at some future time. This memory effect manifests as temporal correlations in the resulting multi-time process. In contrast, a \emph{Markovian process} exhibits memoryless dynamics, with the system's future evolution depending only on its current state, not on its past. The process matrix for such a Markovian process exhibits a product form, see Fig.~\ref{fig:markovNonmarkov_N_step}(b):
\begin{align}
    W_{\rm Markov}{ =} \Choi(\E_{N-1}) \otimes {\cdots} \otimes \Choi(\E_{2})\otimes \Choi(\E_1) \otimes \rho, \label{Eq:MarkovianProcess}
\end{align}
where the initial state $\rho {\in} \St (A_1)$ and each channel $\E_{n}:A_n''{\to} A_{n+1}$ act only on the system. Mathematically, this product form arises because the constituent components: the initial state and the channel Choi operators, act on distinct, non-overlapping systems. Consequently, their link products simplify to tensor products, in contrast to the general multi-time process where the environment mediates the non-trivial composition. This product structure reflects the memoryless nature of Markovian dynamics: there are no temporal correlations beyond those arising from the sequential application of channels.

\begin{figure}
    \centering
    \includegraphics[width=\columnwidth]{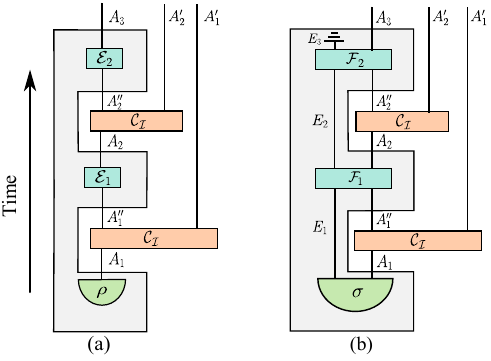}
    \caption{\textbf{ Three-time quantum states.} (a) A three-time quantum state, compatible with the Markovian process in \eqref{Eq:MarkovianProcess} with $N=3$. The state is constructed by inserting broadcasting maps $\mathcal{C}_{\id}$ (see Def.~\ref{def:broadcasting}) into the probe-regions (orange slots in Fig.~\ref{fig:markovNonmarkov_N_step}). As shown in Theorem~\ref{Thm:StructureThm_MarkovQSOT}, this structure follows from Axiom~\ref{axm:PropMTQS}. One can further characterize the broadcasting maps by additional physically motivated axioms. See, e.g., Refs.~\cite{parzygnat2024virtual, xiao2025no}. (b) A general three-time state (see Def.~\ref{Def:QSOTGen}), compatible with a multi-time process in Eq.~\eqref{Eq:non-MarkovianProcess} with $N=3$. As in (a), the state is obtained by fusing broadcasting maps in the probe regions of the general multi-time process. As shown in Theorem~\ref{Thm:NMQSOT}, this structure is derived from Axioms~\ref{axm:PropMTQS} and \ref{axm:spatial-marginal-consistency}.} 
    \label{fig:markovqsot_processmatrix}
\end{figure}

\subsection{On proposed star products for quantum multi-time states}
Before presenting the structure theorems for quantum MTSs and linking them with process matrix framework, we briefly review the various star products proposed so far for constructing states over time in quantum theory, and identify which of them satisfy the properties specified in Axioms~\ref{axm:PropMTQS} and~\ref{axm:spatial-marginal-consistency}. We first consider the non-linear star products used to construct quantum states over time, such as the Ohya compound state~\cite{ohya1983note}, the Leifer--Spekkens (LS) state over time~\cite{leifer2013towards}, the $t$-rotated product~\cite{parzygnat2023time,wilde2015recoverability}, and the STH product~\cite{sutter2016strengthened}. These constructions are non-linear in the input states. As a consequence, for any possible multi-time extension of these states, the reconstructibility via a linear map does not hold. Nevertheless, all of these constructions satisfy the spatio-temporal marginal consistency condition (Axiom~\ref{axm:spatial-marginal-consistency}). To make this explicit, we consider the Leifer-Spekkens (LS) quantum state over time, which, given an initial state $\rho_{A_1} \in \St(\rho_{A_1})$ and a channel $\E: A_1 \to A_2$, is defined as
\begin{align*}
\E \star_{LS} \rho_{A_1}= (\sqrt{\rho_{A_1}} \otimes \Id_{A_2}) \ \jam(\E) \ (\sqrt{\rho_{A_1}} \otimes \Id_{A_2}),
\end{align*}
where $\jam(\E)=\sum_{ij} \ketbra{i}{j}_{A_1} \otimes \E(\ketbra{j}{i})_{A_2}$ denotes the Jamio{\l}kowski matrix on $\mathcal{H}_{A_1}\otimes\mathcal{H}_{A_2}$. This construction satisfies
$\Tr_{A_2}[\E \star_{LS} \rho_{A_1}]=\rho_{A_1}$ and $\Tr_{A_1}[\E \star_{LS} \rho_{A_1}]=\E(\rho_{A_1})$.
More generally, consider an uncorrelated state $\rho_{A_1}\otimes\tau_{E_1} \in \St(A_1 \otimes E_1)$ of the spacelike separated systems $A_1$ and $E_1$, evolving under the channels $\E:A_1 \to A_2$ and $\N: E_1 \to E_2$. In this case,
$(\E\otimes \N)\star_{LS}(\rho_{A_1}\otimes \tau_{E_1})=(\E \star_{LS} \rho_{A_1}) \otimes (\N \star_{LS} \tau_{E_1})$,
since $\jam(\E\otimes \N)= \jam(\E) \otimes \jam(\N)$ and $\sqrt{\rho_{A_1} \otimes \tau_{E_1}}=\sqrt{\rho_{A_1}} \otimes \sqrt{\tau_{E_1}}$. It follows directly that the LS state over time satisfies the spatio-temporal marginal consistency condition in Axiom~\ref{axm:spatial-marginal-consistency}, for instance
\begin{align*}
\Tr_{E_1}[(\E\otimes\mathcal{N})\star_{LS}(\rho_{A_1}\otimes\tau_{E_1})]=(\E\star_{LS}\rho_{A_1})\otimes\mathcal{N}(\tau_{E_1}).   
\end{align*}

We next turn to quantum states over time based on a family of star products, given by
\begin{equation}\label{eq:StarProduct}
    \E\star_p\rho_{A_1}=p[(\rho\otimes\Id)\jam(\E)]+(1-p)[\jam(\E)(\rho\otimes\Id)], ~p\in\mathbb{C}
\end{equation}
where $\rho\in\St{(A_1)}$ is the input state, $\E:A_1\rightarrow A_2$ is a quantum channel, and $\jam(\E)$ is the Jamio{\l}kowski matrix defined earlier. Note that this construction reduces to the symmetric ($p=1/2$), right-bloom ($p=1$), and left-bloom ($p=0$) constructions~\cite{fullwood2025quantum} as special cases. This family of star products is linear in the input state and satisfies both Axioms~\ref{axm:PropMTQS} and~\ref{axm:spatial-marginal-consistency}. To verify the latter, consider an uncorrelated state $\rho_{A_1}\otimes\tau_{E_1} \in \St(A_1 \otimes E_1)$ of the spacelike separated systems $A_1$ and $E_1$, evolving under the channels $\E:A_1 \to A_2$ and $\N: E_1 \to E_2$. For $p\neq 0$ and $p \neq 1$, one readily finds that $(\E\otimes\N)\star_p(\rho_{A_1}\otimes\tau_{E_1})\neq (\E \star_p \rho_{A_1}) \otimes (\N \star_p \tau_{E_1})$. Nevertheless, with $\jam(\E\otimes \N)= \jam(\E) \otimes \jam(\N)$, a straightforward manipulation leads to
\begin{align*}
\Tr_{E_1}[(\E\otimes\N)\star_p(\rho_{A_1}\otimes\tau_{E_1})] = (\E \star_p \rho_{A_1}) \otimes \N(\tau_{E_1}).  
\end{align*}

Note that, by setting $\mathcal{E}=\mathcal{I}$, i.e., the identity channel, the construction given by Eq.~\eqref{eq:StarProduct} yields a broadcasting map, namely, $\mathcal{I}\star_p\rho_{A_1}$. Interestingly, with $p=\frac{1}{2}$, the broadcasting map
\begin{align}
\mathcal{B}_2 (\rho_{A_1})&= \frac{1}{2}\left[(\rho_{A_1}\otimes\Id_{A_2})\jam(\mathcal{I})]+\jam(\mathcal{I})(\rho_{A_1}\otimes\Id_{A_2})\right], \nonumber \\
&=\frac{1}{2}\bigg\{ (\rho_{A_1}\otimes \Id_{A_2}), \ \mathrm{SWAP}\bigg\}.
\label{eq:VirtualBC}
\end{align}
becomes permutation symmetric and unitary covariant~\cite{parzygnat2024virtual,xiao2025no,lie2024quantum}. Here, the second equality follows because $\jam(\id)=\mathrm{SWAP}$ is the swap operator and $\{\cdot,\cdot\}$ denotes the anticommutator. In fact, the map $\mathcal{B}_2$ is \emph{unique} under some reasonable axioms, similar to Axiom~
\ref{axm:single_system_broadcasting}, and also referred to in the literature as virtual quantum broadcasting map~\cite{parzygnat2024virtual, xiao2025no}. Further, the quantum state over time obtained using Eq.~\eqref{eq:StarProduct} with $p=\frac{1}2$ becomes exactly equal to the pseudo-density matrix (PDM) based state over time for qubits~\cite[Appendix~A]{horsman2017can}. It has yet another important property: it admits a valid representation of two-time expectation values for a particular class of observables, known as light-touch observables~\cite{fullwood2024operator}. An observable represented by an operator $M_{A_1}$ is called light-touch if it admits a spectral decomposition of the form $M_{A_1}=\sum_i\lambda_i  P_i$, where the spectra $\{\lambda_i\}$ is either $\{\lambda\}$ or $\{\pm\lambda\}$ for some $\lambda\ge0$, and the $P_i$s are mutually orthogonal projectors. For example, observables associated with Pauli operators are light-touch observables. With this, the two-time expectation value of light-touch observable $M_{A_1}$ and  any observable $N_{A_2}$, for a state $\rho_{A_1}$ evolving under the action of a channel $\E:A_1\to A_2$, is given as
\begin{align*}
\langle M_{A_1}, N_{A_2}\rangle &:= \sum_i \lambda_i \  \Tr[\E(P_i \ \rho_{A_1} \ P_i) \ N_{A_2}] \nonumber \\ 
&= \Tr[\rho_{A_1A_2} \ (M_{A_1} \otimes N_{A_2})],
\end{align*}
where $\rho_{A_1A_2}=\E \star_p \rho_{A_1}$ is the quantum state over time obtained using Eq.~\eqref{eq:StarProduct} with $p=\frac{1}{2}$.

\subsection{Markov and non-Markov compatible quantum MTS}

We now turn to investigate the structures of quantum multi-time states. Before delving into multi-time states compatible with a general multi-time process, it is useful to revisit Theorem~\ref{Thm:StructureThm_MarkovQSOT} within the process matrix formalism. With the multi-time process framework in hand, we can now restate it as follows. 

\begin{CorHighlightbox}
\begin{corollary} \label{Cor:MarkovProc_StTime}
Given a Markovian process $W_{\rm Markov} {=} \bigotimes_{n=1}^{N-1}\Choi(\E_{n}) \otimes \rho$, the corresponding multi-time quantum state $\{\E_n\}_{n=1}^{N-1}\star \rho$ respecting Axiom~\ref{axm:PropMTQS}, has the following structure:
    \begin{align}
        \{\E_n\}_{n=1}^{N-1}\star \rho = W_{\rm Markov} * \left( \bigotimes_{n=1}^{N-1}   \Choi(\mathcal{C}_{\id})_{A_nA_n''A_n'}\right).
    \end{align}
 Here $\C_{\id}:{A_n\to {A''}_n \otimes {A'}_n}$ is a broadcasting map as in Definition~\ref{def:broadcasting}.   
\end{corollary}
\end{CorHighlightbox}

This reformulation follows directly from the structure depicted in Fig.~\ref{fig:markovqsot_processmatrix}(a) and highlights that the Markovian compatible multi-time state is obtained by inserting a broadcasting map into each probe region $n = 1,\cdots,N-1$ of the underlying Markovian process.

Notably, the same unique structure of decomposable quantum multi-time states, resembling the structure presented in Theorem~\ref{Thm:StructureThm_MarkovQSOT}, has also been derived in Ref.~\cite{lie2025multipartite} using a different set of axioms: state linearity and an additional axiom called ``quantum conditionability''.  Interestingly, the decomposable structure already implies reconstructibility, so axioms in Ref.~\cite{lie2025multipartite} give it for free. In our work, we instead start from reconstructibility and compositionality and arrive at the same structure. This convergence suggests the structure is fundamentally natural and robust. Another important distinction lies in the scope of the two frameworks. The proof in Ref.~\cite{lie2025multipartite} is carried out within quantum theory. By contrast, our axioms and their consequences hold for OPTs~\cite{Chiribella_2011, chiribella2010probabilistic, DAriano2016} obeying principles such as local tomography and causality. This generality allows us to construct multi-time states in theories beyond quantum theory, such as boxworld, as explored in Section~\ref{sec:MTSinBoxWorld}. Also, see Appendix~\ref{App:QSOTinGPT} for a discussion.

Quantum conditionability, however, remains important because it appears to generalize the classical Bayesian update rule to quantum theory. Interestingly, the boxworld multi-time states (see Section~\ref{sec:MTSinBoxWorld}) exhibit an analogue of conditionability, that emerges once additional axioms are invoked to characterize the broadcasting map. This seems to suggest conditionability is strictly stronger than Axiom~\ref{axm:PropMTQS}. Having seen it in both quantum and boxworld theory, it would be interesting to know whether it can be derived in a theory-independent way, or axiomatized in operational probabilistic theories to yield the same structure as Theorem~\ref{Thm:StructureThm_MarkovQSOT}. We leave these questions for future work.

Corollary~\ref{Cor:MarkovProc_StTime} hints at a natural generalization: a general multi-time quantum state can be obtained by replacing the Markovian process with a general (multi-time) process, as in Fig.~\ref{fig:markovqsot_processmatrix}(b), thereby unifying the frameworks of multi-time quantum states and process matrices. Now the structure of general quantum multi-time state, as given in Theorem~\ref{Thm:NMQSOT}, can be expressed using process matrices, as given below.

\begin{CorHighlightbox}
\begin{corollary}
 Given a multi-time process with a process matrix $W{=} \Tr_{E_N}\big[\Ast _{n=1}^{N-1} \Choi(\mathcal{F}_n) * \sigma \big]$, with $\mathcal{F}_{n} \in \Chan (A_n''\otimes E_n\to A_{n+1}\otimes E_{n+1})$, the corresponding quantum multi-time state  $\tau_{A_N\cdots A_2A_1}=\Tr_{E_N\cdots E_1}\big[\{\mathcal{F}_n\}_{n=1}^{N-1}\star \sigma \big]$, respecting Axioms~\ref{axm:PropMTQS} and \ref{axm:spatial-marginal-consistency} satisfies
\begin{align}
\tau_{A_N\cdots A_2A_1}=\Tr_{E_N\cdots E_1}\big[\{\mathcal{F}_n\}_{n=1}^{N-1}\star \sigma \big] =W * \left( \bigotimes_{n=1}^{N-1}   \Choi(\mathcal{C}_{\id})\right).   
\end{align}   
\end{corollary}
\end{CorHighlightbox}

Having expressed quantum multi-time states in terms of process matrices, all the structural and conceptual aspects of MTS established through our axiomatic framework in the previous sections can now be reformulated within quantum theory. Conversely, the dynamical features captured by the quantum process-matrix framework can be translated into the language of multi-time states. This correspondence, in particular, opens the possibility of studying exotic causal structures---such as indefinite causal order and causal non-separability---within the multi-time state framework. We leave a systematic investigation of these connections for future work. In the remainder of this section, we instead focus on how temporal quantum memory and causal signaling are manifested in quantum multi-time states, and how these states can be exploited for causal discovery. 

\subsection{Semicausality and non-Markov processes \label{sec:SemiCausalProcessMatrix}}
Let us reinvestigate the connection between semicausality and non-Markov compatibility, discussed in Section~\ref{sec:CausalSignalNonMarkov}, in quantum multi-time states. Consider a bipartite quantum channel $\mathcal{F}:A_I B_I\rightarrow A_O B_O$. To establish the connection between causal structure and multi-time states, we restrict, without loss of generality, to semicausal channels~\cite{Preskill_casual, Eggeling2002Semicuasal}, for which signaling from $A_I$ to $B_O$ is prohibited, i.e., $A_I\not\to B_O$. Given an initial state $\sigma\in\St(A_I\otimes B_I)$, the action of $\mathcal{F}$ induces a quantum two-time state $\mathcal{F}\star\sigma\in A_I'\otimes A_O\otimes B_O\otimes B_I'$, as illustrated in Fig.~\ref{fig:QSOT-Semicausality}(b).

We are interested in understanding how the causal structure of the underlying process is reflected in the corresponding quantum multi-time state. In particular, suppose that we retain only the input system $A_I'$ and the output system $B_O$, while discarding the intermediate systems $A_O$ and $B_I'$. The resulting reduced state $\Tr_{A_OB_I'}(\mathcal{F}\star\sigma)$ raises a natural question: does the absence of signaling from $A_I$ to $B_O$ impose a characteristic constraint on this reduced two-time state? Conversely, can such a constraint on a two-time state be used to infer or detect the causal structure of the underlying quantum process, and more generally, to identify signaling in quantum theory? We have already established in Corollary~\ref{cor:semicausality} that non-positive semidefiniteness of the two-time states $\Tr_{A_OB_I'}(\mathcal{F}\star\sigma)$ implies signaling from $A_I$ to $B_0$. We now highlight that the above idea of detecting signaling through a bipartite channel can be extended to multi-time processes as well. As a concrete application, we show how Theorem~\ref{thm:semicausal} can be adapted to certify memory in a multi-time process next. 

\begin{figure} 
    \centering
    \includegraphics[width=\columnwidth]{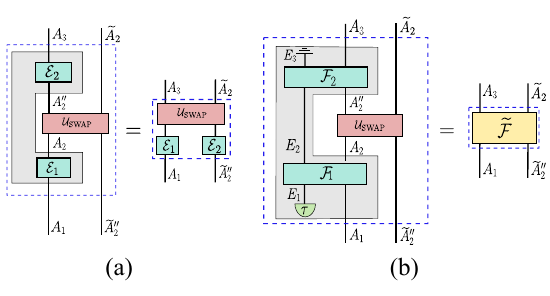}
    \caption{\textbf{Using multi-time process–channel isomorphism to detect memory.} The isomorphism is realized by inserting a SWAP unitary in the probe regions. We illustrate this for a two-time process ($N=2$). (a) A two-time process (see Eq.~\eqref{Eq:semicausality_twoT_process}) is isomorphic to a bipartite channel $\widetilde{\mathcal{F}}$ (see Eq.~\eqref{Eq:bipartite_channel_semicausality}). This bipartite channel can then be used to apply Theorem~\ref{thm:semicausal} to determine whether signaling from $A_1$ to $A_3$ occurs, thereby detecting memory in the dynamics. (b) In the Markovian case, the corresponding bipartite channel is a product channel followed by a SWAP interaction. Here, the reduced multi-time state (see Eq.~\eqref{Eq:semicausality_reduced_QSOT_Markovian}) is always positive. Consequently, the presence of memory is revealed by the failure of positivity in the reduced multi-time state. }
    \label{Fig:memory_semicausality}
\end{figure}

Consider a general three-time process $W\in A_3\otimes A_2'' \otimes A_2\otimes A_1$, constructed from of an initial environment state $\tau_{E_1}$ and two successive system-environment channels $\mathcal{F}_1:A_1\otimes E_1\to A_2\otimes E_2$, and $\mathcal{F}_2:A_2''\otimes E_2\to A_3\otimes E_3$, as depicted in Fig.~\ref{Fig:memory_semicausality}(b):

\begin{align}
 W=\Tr_{E_3}\big[\Choi(\mathcal{F}_2) * \Choi(\mathcal{F}_1) *  \tau\big].  \label{Eq:semicausality_twoT_process}  
\end{align}
For simplicity, we take the initial environment state $\tau_{E_1}$ to be uncorrelated with the system $A_1$. The result can be straightforwardly extended to initially correlated cases. We do not yet specify a particular input state at $A_1$, focusing instead on the structure of the channels and how they connect to  Theorem~\ref{thm:semicausal}. 

In particular, we can see that this three-time process is mathematically isomorphic to a bipartite channel $\widetilde{\mathcal{F}}:A_1\otimes\widetilde{A}_2''\to A_3\otimes\widetilde{A_2}$ where $\widetilde{A}_2''\simeq A_2''$ and $\widetilde{A}_2\simeq A_2$. This isomorphism is implemented by inserting a SWAP unitary: $\mathcal{U}_{\rm SWAP}: A_2\otimes\widetilde{A_2}''\to A_2''\otimes\widetilde{A}_2$, defined by  $\mathcal{U}_{\rm SWAP}(\rho)={\rm SWAP}\  \rho\  {\rm SWAP}$ with ${\rm SWAP}=\sum_{i,j}\ketbra{i}{j}\otimes \ketbra{j}{i}$. As illustrated in Fig.~\ref{Fig:memory_semicausality}(a), this gives
\begin{align}
    \Choi(\widetilde{\mathcal{F}}) =  W * \Choi(\mathcal{U}_{\rm SWAP}). \label{Eq:bipartite_channel_semicausality}
\end{align}

Operationally, the insertion of this Swap unitary is motivated by our aim to detect `memory' in the dynamics. Here, memory manifests as the potential for signaling along the path $A_1\to E_2\to A_3$: information that flows from the system into the environment and later returns. However, there is also a direct signaling path $A_1\to A_2\to A_2''\to A_3$ through the system itself. The SWAP unitary breaks this direct path, allowing us to recast the dynamics as a bipartite channel and isolate the memory effect. Note that inserting a SWAP for this causal break is only a proof-of-principle approach; weaker interventions may suffice, as exemplified in Refs.~\cite{Goswami2021,roy2026practical}.

Now consider the Markovian case, where the process matrix factorize as
\begin{align}
  W_{\rm Markov} = \Choi(\E_2)\otimes \Choi(\E_1),  \label{Eq:semicausality_twoT_Markovian_process}
\end{align}
with channels $\E_1:A_1\to A_2$ and $\E_2: A_2''\to A_3$ ($A_2''\simeq\widetilde{A}_2''$) (see Fig.~\ref{Fig:memory_semicausality}(a)). After inserting the SWAP, the corresponding bipartite channel becomes 
\begin{align}
    \widetilde{F}_{\rm Markov} = \mathcal{U_{\rm SWAP}}\circ (\E_1 \otimes \E_2).
\end{align}
Clearly, here $A_1$ can not signal to $A_3$. Equivalently, $\E_2$ emerges as the marginal channel of $\widetilde{F}_{\rm Markov}$: $\Tr_{\widetilde{A}_2}\circ \widetilde{F}_{\rm Markov} = \Tr_{A_1}\otimes \E_2$. 

As a consequence, Theorem~\ref{thm:semicausal} dictates that for any input state $\sigma_{A_1\widetilde{A}_{2}''}$,  
\begin{align}
    \Tr_{\widetilde{A_2}\widetilde{A_2}''} \big(\widetilde{\mathcal{F}}_{\rm Markov} \star \sigma\big) = (\id_{A_1}\otimes \E_2)(\sigma).  \label{Eq:semicausality_reduced_QSOT_Markovian}
\end{align}
Therefore, by Corollary~\ref{cor:semicausality}, whenever the bipartite channel $\widetilde{F}$ obtained from an underlying dynamics yields a reduced two-time state that is not positive semidefinite, then the dynamics must be non-Markovian. This establishes a connection between our axiomatic framework and the detection of memory effects.

\subsection{Causal discovery via semidefinite programming}

\begin{figure}
    \centering
    \includegraphics[width=1\linewidth]{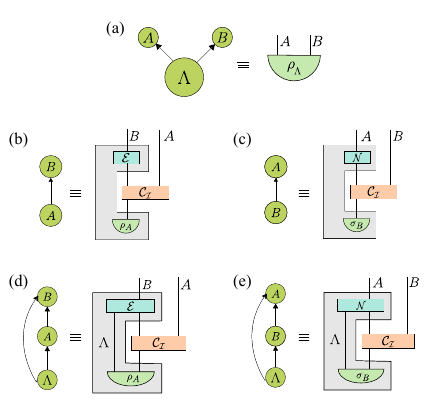}
\caption{\textbf{A bipartite spatio-temporal state compatible with different causal models.} Consider a bipartite spatio-temporal state $\tau_{AB}$.
(a) \textit{Common cause:} A latent variable $\Lambda$ influences both $A$ and $B$. In this case, $\tau_{AB}$ is simply a bipartite physical state.
\textit{Direct cause:} (b) causal order $A \prec B$; (c) causal order $B \prec A$. By Theorem~\ref{Thm:StructureThm_MarkovQSOT}, $\tau_{AB}$ is compatible with a Markovian process in either case.
\textit{A combination of direct and common cause:} (d) causal order $A \prec B$; (e) causal order $B \prec A$. By Theorem~\ref{Thm:NMQSOT}, $\tau_{AB}$ arises from a non-Markovian process in either case.
The task of causal discovery is to identify which causal model gives rise to a given spatio-temporal state. In a ``causally neutral'' framework such as classical probability theory,  a bipartite joint distribution is compatible with all of the above causal models, making it impossible to infer the underlying causal model from the data alone. In more general theories, such as quantum theory, however, a ``causally neutral'' spatio-temporal state is not always possible, allowing us to rule out certain causal models based on the spatio-temporal state $\tau_{AB}$. }
    \label{fig:causalmodels}
\end{figure}

As an application of the correspondence established between quantum multi-time states and process matrices, we now investigate how this framework can be used to infer causal relationships among temporal events, a task commonly referred to as causal discovery.  A question in causal discovery is whether we can figure out a causal model just from the observed correlations~\cite{pearl_2009,Wood2015,Giarmatzi2018}. In our setting, this question becomes, given a multi-time quantum state, what causal model could have produced it? As Fig.~\ref{fig:causalmodels} illustrates, even a simple two-time quantum state can potentially arise from multiple causal structures. The task, then, is to find all causal models consistent with the corresponding state over time and rule out those that are not.

As shown in Refs.~\cite{costa2016,barrett_2019,Barrett_2021}, the underlying causal model can be uniquely represented by a process matrix with appropriate linear constraints. In this language, our question is given an $N$-time quantum state $\rho_{A_N\cdots A_2A_1}$, is it compatible with a general $N$-time quantum process that respects a well-defined causal order? More precisely, denoting the collective input ($A_i$) and output $(A''_i)$ systems of the probing region $i$ as $\mathbb{A}_i$, a causal order is given by a permutation $\sigma$ such that $\mathbb{A}_{\sigma(1)}\preceq \mathbb{A}_{\sigma(2)} \cdots \preceq \mathbb{A}_{\sigma(N)}$, where $\mathbb{A}_{i}\preceq \mathbb{A}_{j}$ means that the probing region $i$ causally precedes probing region $j$, i.e., no information can flow from $j$ to $i$. For notational convenience, from now on we discuss the scenario with the causal order $\mathbb{A}_{1}\preceq \mathbb{A}_{2} \cdots \preceq \mathbb{A}_{N}$. The same arguments hold for any other causal order by relabeling.

As shown in Eq.~\eqref{Eq:non-MarkovianProcess}, such a process $W$ can be represented by an initial system-environment state together with a sequence of joint channels. Alternatively, $W$ can also be characterized directly by a set of linear constraints. To state these constraints compactly, we use the following notation~\cite{araujo2015witnessing,Bavaresco2021}: for any subsystem $X$, let
\begin{align}
    _{X}W :=\Tr_{X}W \otimes \frac{\Id _X}{d_X},
\end{align}
where $d_X$ is the dimension of $X$. Using this notation, a general multi-time process $W$ consistent with the causal order $\mathbb{A}_{1}\preceq \mathbb{A}_{2} \cdots \preceq \mathbb{A}_{N}$ satisfies the following constraints for all $1< n \le N$:

\begin{align}
    &W \ge 0, \\
    &_{A_NA''_{N-1}\cdots A_n}W =     _{A_NA''_{N-1}\cdots A_n A''_{n-1}}W, \\
    &\Tr(W) = d_{A_1''}d_{A_2''}\cdots d_{A''_{N-1}}.
\end{align}
Operationally, the first operator inequality ensures the CP-preserving property of the process: any set of completely positive maps (possibly with an ancillary system) is transformed to a completely positive map. The second constraint encodes the causality constraint, the systems at causal future cannot influence the past systems. The third constraint fixes the overall normalization. 

One can verify that the process matrix in Eq.~\eqref{Eq:non-MarkovianProcess}  satisfies these constraints due to the trace preserving property of the channels. Conversely, it is known that any operator satisfying the above constraints represents a valid $N$-time quantum process~\cite{chiribella09b}. This characterization thus provides an algorithmic way to search over multi-time processes with a given causal structure. 

We denote the set of multi-time processes satisfying the above constraints by $\mathcal{W}^{\mathbb{A}_{1}\preceq \mathbb{A}_{2} \cdots \preceq \mathbb{A}_{N}}$.

Now let $\rho_{A_N\cdots A_2A_1}$ be a general $N$-time quantum state.  From Theorem~\ref{Thm:NMQSOT}, a multi-time state compatible with a process $W\in \mathcal{W}^{\mathbb{A}_{1}\preceq \mathbb{A}_{2} \cdots \preceq \mathbb{A}_{N}}$ is given by 
\begin{align}
    \rho_{A_N \cdots A_2A_1} = W * \left( \bigotimes_{i=1}^{N-1} \Choi(\mathcal{C}_{\id})_{A_i A_i'' A_i'} \right),
\end{align}
where $\mathcal{C}_{\id}$ is a broadcasting map that defines the star product. For example, we can choose the virtual broadcasting map $\mathcal{B}_2$ from~\cite{parzygnat2024virtual} that reproduces the Jordan product-state over time for the two-time scenario. The causal inference problem then becomes given $\rho$, find a process $W\in \mathcal{W}^{\mathbb{A}_{1}\preceq \mathbb{A}_{2} \cdots \preceq \mathbb{A}_{N}}$ that satisfies this equation.

To solve this feasibility problem, we can frame it as an optimization that minimizes the distance between $\rho$ and the multi-time state generated by a process from the set $\mathcal{W}^{\mathbb{A}_{1}\preceq \mathbb{A}_{2} \cdots \preceq \mathbb{A}_{N}}$. A natural choice is to minimize the trace norm distance:

\begin{align}
    &\text{minimize} \quad \frac{1}{2} \left\| \  \rho - W * \left( \bigotimes_{i=1}^{N-1} \Choi(\mathcal{C}_{\id})_{A_i A_i'' A_i'} \right)\  \right\|_1 \\
    &\text{subject to} \quad W \in \mathcal{W}^{\mathbb{A}_{1}\preceq \mathbb{A}_{2} \cdots \preceq \mathbb{A}_{N}}.
\end{align}

This can be reformulated as the following Semidefinite Program (SDP) by introducing auxiliary variables $Y$ and $Z$~\cite{watrous2009semidefinite,watrous2013simpler}:

\begin{align} \min_{W, Y, Z} \quad & \Tr(Y + Z) \\ \text{subject to} \quad & Y - Z = \rho - W * \left( \bigotimes_{i=1}^{N-1} \Choi(\mathcal{C}_{\id})_{A_i A_i'' A_i'} \right), \\ & Y \ge 0, \quad Z \ge 0, \\ & 
W \in {\mathcal{W}^{\mathbb{A}_{1}\preceq \mathbb{A}_{2} \cdots \preceq \mathbb{A}_{N}}}.
\end{align}

If the optimal value of this SDP is zero, then the multi-time state $\rho$ is compatible with the causal order $\mathbb{A}_{1}\preceq \mathbb{A}_{2} \cdots \preceq \mathbb{A}_{N}$, and the optimal $W$ serves as a witness to this fact.

\section{Conclusions \label{sec:Conclusion}}

In this work, we address two major gaps in the theory of multi-time states: the lack of an axiomatic characterization of multi-time states corresponding to non-Markovian dynamics, and the absence of a general framework for spatio-temporal states beyond quantum theory itself. To bridge these gaps, we provide a set of operational principles formulated within the broader framework of operational probabilistic theories (OPTs)~\cite{chiribella2010probabilistic,Chiribella_2011,DAriano2016} with finite-dimensional convex state spaces that satisfy causality (i.e., a unique discarding effect)~\cite{Chiribella_2011,chiribella2010probabilistic} and local tomography~\cite{hardy2001quantum, hardy2011reformulating,chiribella2010probabilistic,Barrett2007,DAriano2016, Chiribella2021}. Specifically, we have introduced three operationally motivated axioms---compositionality, reconstructibility, and spatio-temporal marginal consistency---and have shown that these axioms uniquely determine the structure of multi-time states in both Markovian and non-Markovian settings. Our results establish a connection between multi-time states and the well-established framework of multi-time processes (such as quantum combs, process tensors, process matrices in quantum theory), which provide a complete description of open system dynamics, including both Markovian and non-Markovian evolutions~\cite{Pollock2018Operational}.

For Markovian dynamics, our structure theorem implies that a Markov compatible multi-time state is obtained by inserting broadcasting maps into the intervention slots of a Markov multi-time process: the initial physical state undergoes a broadcasting map, one output evolves through a local channel, whose output is then fed into the next broadcasting map, and the structure repeats sequentially. Given this structure, characterizing Markov compatible multi-time states reduces to identifying the admissible broadcasting maps, which are typically fixed by additional operational axioms~\cite{parzygnat2024virtual}. This establishes broadcasting maps as the fundamental building blocks of the structure of Markov compatible multi-time states, consistent with a previously obtained result in quantum theory~\cite{lie2024quantum,lie2025multipartite}. Our contribution is to show that this structure holds more generally within the framework of operational probabilistic theories, without relying on any specific features of quantum theory.

Our axiomatic framework extends naturally beyond the Markovian regime. By imposing an additional axiom of spatio-temporal marginal consistency, we establish a structure theorem for general multi-time states, including those compatible with non-Markovian dynamics. Operationally, both Markovian and non-Markovian dynamics are instances of general multi-time processes, in which a joint initial system-environment state evolves through a sequence of system-environment interactions, followed by discarding the inaccessible environmental degrees of freedom at the final time step, and an intervention slot on the system between successive interactions. We show that a multi-time state compatible with such a general process can be constructed by inserting suitable broadcasting maps at each intervention slot of the underlying multi-time process, provided the multi-time state satisfies spatio-temporal marginal consistency together with compositionality and reconstructibility. These structure theorems establish, for the first time, an explicit axiomatic characterization of multi-time states arising from arbitrary system-environment dynamics. Furthermore, as in the Markovian case, our axioms are not restricted to quantum theory but apply to the broader class of operational probabilistic theories. This has foundational implications: it provides a concrete structure for multi-time states and potentially their associated Bayesian inference rules across a wide range of physical theories. Such a framework can help clarify which features of Bayesian calculus are fundamental and arise from theory-independent operational principles.

In addition to its foundational implications, our axiomatic framework also leads to several concrete applications. One of our results establishes a fundamental connection between the properties of a spatio-temporal state---arising from a joint initial physical state and a subsequent bipartite channel---and signaling through that channel. Specifically, we show that for any spatio-temporal state satisfying reconstructibility, compositionality, and spatio-temporal marginal consistency, if the reduced spatio-temporal state, obtained by discarding the relevant intermediate systems, is not a physical state of the theory, then this provides a sufficient condition for certifying signaling through the bipartite channel. In quantum theory, for Hermitian spatio-temporal states, this property translates to negativity of a particular marginal spatio-temporal state, a special case of which was investigated in the pseudo-density matrix framework~\cite{Liu2025}. In our work, we demonstrate that this connection is not a peculiar feature of quantum theory or any particular formalism for spatio-temporal states, but rather emerges directly from the operational principles underlying our framework. Furthermore, we also show that this property can be extended to multi-time processes and used to certify memory in dynamics. This observation could lead to experimental tests to detect memory in a non-Markovian process, analogous to tests used in causal inference~\cite{liu2024experimentaldemonstrationquantumcausal}. 

We also show that the correspondence between multi-time states and processes naturally leads to an application in causal discovery. Finding a causal model from observed correlations is a central goal in the causality community~\cite{pearl_2009,Wood2015,Giarmatzi2018}. Since the framework of multi-time states generalizes the notion of classical joint probability distributions in time, it is natural to ask: given a multi-time state, which underlying causal model, captured by a multi-time process~\cite{Spekkens_cause,costa2016,Barrett_2021}, could explain it? In general, the inferred causal model need not be unique. As shown in Ref.~\cite{Song2025}, certain multi-time states can be compatible with both a bipartite state and a Markov compatible two-time state. Similarly, we have shown that the same multi-time state can be compatible with both Markovian and non-Markovian dynamics. This ambiguity is not universal, however, due to the inherent asymmetry between spacelike and timelike correlations in nonclassical theories. The goal of a causal discovery algorithm is therefore to rule out incompatible causal models, rather than to identify a unique one. We demonstrate a proof-of-principle approach to this question in the quantum case. Specifically, given a multi-time state, we ask whether there exists a multi-time process with a particular causal order of interest that could explain it. We show that this question can be addressed via a semidefinite program (SDP). Our structure theorems provide the foundation for this formulation, offering a systematic framework for quantum causal discovery from multi-time states. 

Finally, to justify that our framework indeed extends beyond quantum theory, we develop, for the first time, a complete framework of multi-time states and Bayesian retrodiction in the square-bit boxworld. As a first step, we characterize the broadcasting map there using operational principles analogous to those in Ref.~\cite{parzygnat2024virtual}, which had been left unspecified in the general structure theorems. Since the square-bit boxworld satisfies the assumptions of our OPT framework, the structure theorems apply directly whenever a multi-time state satisfies reconstructibility, compositionality, and spatio-temporal marginal consistency. Moreover, the unique broadcasting map we identify yields a unique multi-time state in this theory, and all the applications discussed above extend straightforwardly to boxworld. This example concretely establishes our axiomatic framework in a Hilbert-space-free setting, without relying on the tools of quantum theory.

Our results raise several intriguing questions and open up new directions for future research, which we discuss in the following. Although we found a unique structure for multi-time states in a broad class of OPTs, the broadcasting map itself was left unspecified. In quantum theory and boxworld, however, we have concrete examples of unique broadcasting maps, both derived from similar operational axioms~\cite{parzygnat2024virtual}. This naturally raises the question: are these axioms sufficiently universal to single out a broadcasting map in any OPT? Several mathematical tools are already available to explore this question. It is known that reversible transformations in any finite-dimensional OPT with closed-convex state space form a compact Lie group~\cite{Gianelli2026}, and hence admit an irreducible decomposition. This means that maps satisfying covariance conditions under reversible transformations can, in principle, be characterized. The remaining question is whether classical consistency and permutation invariance are sufficient to force a unique solution. A positive answer could lead to a unique state over time for any OPT satisfying our assumptions.

Another open question is whether our framework can be extended beyond locally tomographic OPTs. In principle, one could attempt to define multi-time states in some of these theories, such as fermionic theories, which are simply a superselection of quantum theory, by projecting quantum multi-time states onto the relevant subspaces. However, it would be more satisfying to characterize multi-time states directly within the OPT framework, without relying on an embedding into quantum theory. This remains an interesting direction for future work.

One of the main motivations for the quantum state-over-time formalism was to develop a rigorous, non-commutative Bayesian calculus within quantum theory. This effort is partly inspired by the success of Bayesian methods in classical settings, with applications in signal processing, machine learning, and state estimation~\cite{Sarkka2013-mn, Bishop2006-qp}. In quantum theory, similar ideas have led to Bayesian phase estimation~\cite{Wiebe2016}, quantum state smoothing~\cite{Tsang2009, guevara2015smoothing, Liu2025_retrodictive,liu2026unifying}, and quantum belief propagation~\cite{Hastings2007}, to name a few. A systematic exploration of these applications through the unified lens of spatio-temporal states could provide a rigorous, device-independent, and theory-independent foundation for existing protocols, and may also uncover new Bayesian tools that have no classical analog, similar to Ref.~\cite{Liu_2026_proper_unproper}.

Another interesting direction is to distinguish two types of multi-time states: those arising from a probabilistic mixture of Markovian processes (a class of classical memory processes), and those with genuine quantum memory. In the multi-time process framework, considerable progress has been made in distinguishing these~\cite{Giarmatzi_2021, taranto2024characterising, Goswami2025}. A natural question is whether this distinction is also visible at the level of multi-time states. One possible approach is to adapt the concept of squashed non-Markovianity~\cite{Gangwar2025}, originally developed for tripartite physical states and later extended to dynamical settings~\cite{Buscemi2025}. This could provide a quantitative way to distinguish different types of non-Markovianity in multi-time states.

Finally, an interesting endeavor would be to go beyond multi-time processes and consider more exotic scenarios, such as processes with indefinite causal order~\cite{oreshkov12,chiribella09} or indefinite input-output directions~\cite{chiribella2022quantum}. In the pseudo-density matrix (PDM) framework, some initial steps have been taken to connect PDMs with indefinite causal order, showing that PDM negativity can be activated in such scenarios~\cite{Liu_2024}. An axiomatic framework could provide a deeper understanding of this phenomenon and potentially generalize such negativity activation across different physical theories.

\begin{acknowledgements}
This research was completed while K.G.\ was a postdoc at HKU. K.G.\ was
supported by the Hong Kong Research Grant Council (RGC) through grant
No.~17307520 and by the John Templeton Foundation through grant 62312,
``The Quantum Information Structure of Spacetime'' (qiss.fr). A.B. thankfully acknowledges financial support from IISER Mohali through the institute-funded PhD programme. All the authors align with the manifesto of \emph{Quantum Scientists for Disarmament} \cite{Disarmament2026} and oppose any use of the results of this work for military or surveillance purposes.\\

\smallskip

\noindent \textbf{AI Usage:} The conceptualization, proof strategy, and all technical content are entirely generated by the authors, all of whom, to the best of our knowledge, are humans. We used standard AI tools to polish the writing and improve styles. No conceptual or mathematical content was AI-generated. KG used DeepSeek-V4.1-Flash to generate the Python code~\cite{Goswami_GlobalBroadcastingCode_2026}, which he subsequently modified.
\end{acknowledgements}

\appendix

\section{Towards a formulation of state over time in general physical theories}\label{App:QSOTinGPT}

We now introduce the framework of a broader class of physical theories on which
our framework for multi-time states rests. Generally, such theories are formalized within the framework of general probabilistic theories (GPTs)\cite{hardy2001quantum,Barrett2007,HARDY2013,Plvala2023}, which provide a unified setting for studying quantum theory and beyond. Here, we focus on a special class of GPTs, known as operational probabilistic theories (OPTs)~\cite{chiribella2010probabilistic,Chiribella_2011,DAriano2016}, which is inspired by category theory. In this theory, physical transformations and their sequential composition are treated as primitives, and physical systems are equipped with parallel composition rules. We begin with a brief introduction to OPT. For conciseness, we introduce OPT in algebraic notation; a comprehensive diagrammatic introduction is given in Refs.~\cite{chiribella2010probabilistic,Chiribella_2011,DAriano2016}.

\subsection{OPT framework}

\noindent \textbf{Physical systems:} The physical systems in OPT determine the input and output of all physical transformations. In classical theory, physical systems are random variables; in quantum theory, they are represented by Hilbert spaces. For two physical systems $A$ and $B$, the composite system (or their parallel composition) is denoted by $A\otimes B$. The parallel composition is associative: $A\otimes (B\otimes C) = (A\otimes B)\otimes C$ for all systems $A,B,C$. To keep notation concise, we write $AB$ for $A\otimes B$ whenever there is no risk of confusion. This parallel composition rule is not unique across theories: different choices, such as the minimal or maximal tensor product~\cite{Plvala2023} or the wedge product~\cite{Vidal2021}, can yield theories with distinct physical features. In quantum theory, this composition rule is isomorphic to the algebraic tensor product. Every OPT contains a distinguished `trivial system' $I$ satisfying $A\otimes I = I \otimes A = A$ for all systems $A$. It is isomorphic to a one-dimensional vector space. 

\medskip

\noindent{\textbf{Transformation and tests}:} The set of physical transformations from an input system $A_I$ to an output system $A_O$ is denoted by $\mathsf{Transf}(A_I\to A_O)$. Parallel composition of two transformations $\N \in \mathsf{Transf}(A_I\to A_O)$ and $\M \in \mathsf{Transf}(B_I\to B_O)$ results in a new transformation $\N \otimes \M \in \mathsf{Transf}(A_I \otimes B_I\to A_O \otimes B_O)$. On the other hand, transformations $\N \in \mathsf{Transf}(A_1\to A_2)$ and $\M \in \mathsf{Transf}(A_2\to A_3)$ can be sequentially composed to get a new transformation $\M \circ \N \in \mathsf{Transf}(A_1\to A_3)$. Just like the parallel composition, the sequential composition `$\circ$' is also associative. Every system $A$ has an identity transformation $\id \in \mathsf{Transf}(A\to A)$ satisfying $\N \circ \id_A = \id_B \circ \N$ for all $\N \in \mathsf{Transf}(A\to B)$ and for all systems $A$, $B$. Moreover, for any two systems $A$ and $B$, there is a SWAP transformation $\mathcal{U}_{\rm{SWAP}} \in \mathsf{Transf}(A\otimes B \to B \otimes A)$ that exchanges the order of the systems. It satisfies $\mathcal{U}_{\rm{SWAP}} \circ (\N \otimes \M) = (\M \otimes \N) \circ \mathcal{U}_{\rm{SWAP}} $ for all $\N \in \mathsf{Transf}(A \to A)$ and $\M \in \mathsf{Transf}(B \to B)$. Together, systems and transformations form a symmetric monoidal category, where systems are objects, transformations are morphisms, sequential composition is the categorical composition, and the parallel composition gives the monoidal structure, with the SWAP map providing symmetry.
In quantum theory, transformations are represented by completely positive trace non-increasing (CPTNI) linear maps acting on the space of bounded operators on a Hilbert space.

\smallskip

A test represents a scenario where an experimentalist performs a collection of transformations $\{\N_x\}_{x\in X}$ with each $\N_x \in \mathsf{Transf}(A \to B)$ associated with a classical outcome $x$. In quantum theory, this corresponds to a quantum instrument. The set of all tests from $A$ to $B$ is represented by $\mathsf{Tests}(A \to B)$. Tests are composed in parallel as $\{\N_x\}_{x\in X} \otimes \{\M_y\}_{y\in Y} = \{\N_x \otimes \M_y\}_{(x,y)\in X\times Y}$ and sequence as 
$\{\N_x\}_{x\in X} \circ \{\M_y\}_{y\in Y}=\{\N_x \circ \M_y\}_{(x,y)\in X\times Y}$. A single-outcome test corresponds to a \emph{deterministic transformation}. For example, coarse-graining a test $\{\N_x\}_{x\in X}$ to a single outcome yields a deterministic transformation $\sum_{x\in X}\N_x$. Deterministic transformations, however, are primitive and need not necessarily arise from coarse-graining tests (e.g., waveplates in optics). The set of all deterministic transformations between $A$ and $B$ is denoted by $\Chan(A\to B)$, following the quantum theory convention where deterministic transformations are known as quantum channels and are given by completely positive trace preserving (CPTP) linear maps.

\medskip

\noindent{\textbf{States, effects, measurements and scalars}:} Operationally, a state denotes the process of initializing a system.  In OPT, a state $\rho$ of system $A$ is a transformation from the trivial system $I$ to $A$, i.e. $\rho \in \mathsf{Transf}(I\to A)$, with the set of all states denoted by $\St _1(A):=\mathsf{Transf}(I\to A)$. In quantum theory, these correspond to subnormalized density operators, represented by positive operators with trace less than or equal to one. Deterministic states in $A$ are elements of $\St(A):=\Chan(I \to A)$. In quantum theory, these are positive operators with unit trace (density operators) \footnote{Note, our notation differs slightly from the OPT literature~\cite{DAriano2016}. In OPT, the state space includes subnormalized states corresponding to probabilistic preparations, whereas in most GPT~\cite{Plvala2023} and quantum theory, the state space refers to deterministic (normalized) states. To avoid confusion, we denote the deterministic states by $\St(A)$ and the OPT state space by $\St_1(A)$.}.

\smallskip

On the other hand, effects in the system $A$ are the elements of the effect space $\eff(A):=\mathsf{Transf(A\to I)}$, which are transformations from the system $A$ to the trivial systems $I$. Similarly, a collection of effects $\{e_x\}_X$ in $A$ is an element of $\mathsf{Tests}(A\to I)$ representing a measurement, and generalizing the positive operator valued measure (POVM) in the quantum theory. The set of deterministic effects on $A$ are denoted by $\mathsf{DetEff}(A):=\Chan(A\to I)$. An important principle is the \emph{Causality} which states that for every system $A$ there is only one and unique element in $\mathsf{DetEff}(A)$. Quantum theory satisfies this principle, and the deterministic effect is the trace map $\Tr$. In our work, we assume the causality principle, and denote the deterministic effect in $A$ by $\Tr_A$, or simply $\Tr$ when there is no confusion.

\smallskip

The elements of $\mathsf{Transf}(I\to I)$ are called scalars. They are crucial to the probabilistic structure of the OPT, as their sequential and parallel compositions commute.

\medskip

\noindent \textbf{Probabilistic structure and quotienting:} The probabilistic structure of OPT is captured by a map $\mathsf{Prob}: \mathsf{Transf}(I \to I) \to [0,1] \subset \mathbb{R}$, , which assigns probabilities to scalars and satisfies

\begin{align}
  &\sum_{x\in X}\mathsf{Prob}(s_x)=1,\ \forall \{s_x\}_{x\in X} \in \mathsf{Tests}(I\to I),\\
  & \mathsf{Prob}(s\otimes t) = \mathsf{Prob}(s).\mathsf{Prob}(t), \ \forall s,t \in \mathsf{Transf}(I\to I).
\end{align}
Thus, for a state $\rho\in \St(A)$ and measurements $\{e_x\}_{x\in X}\in \mathsf{Tests}(A\to I)$, their composition $ \{e_x(\rho) := e_x\circ \rho\}_{x\in X} \in \mathsf{Tests(I\to I)}$  yields the measurement outcome probability $p(x) = \mathsf{Prob}(e_x (\rho) )$. In quantum theory, this is the Born rule~\cite{born_1926}.

Given $\Tr$ is a deterministic effect, for all $\rho \in \St(A)$, $\Tr(\rho):=\Tr \circ \rho$ is a deterministic scalar, hence $\mathsf{Prob}(\Tr(\rho))=1$. Conversely, thanks to the causality axiom, any probabilistic state $\tilde{\rho}\in \St _1(A)$ can be uniquely normalized to a deterministic state $\rho \in \St(A)$ via $\tilde{\rho}= \mathsf{{Prob}(\Tr(\tilde{\rho}))\  .\  \rho}$.

The probabilistic structure also induces an equivalence relation on transformations. Two transformations $\N,\M \in \mathsf{Transf}(A\to B)$ are equivalent if, for all auxiliary system $R$, every state $\rho\in \St(AR)$ and effects $e\in \eff(BR)$, 
\begin{align}
    \mathsf{Prob}( e \circ (\N \otimes \id_R) \circ \rho) =  \mathsf{Prob}( e \circ (\M \otimes \id_R) \circ \rho). \label{Eq:Equivalent_transforamtion}
\end{align}
By quotienting out this equivalence relation, that is, grouping together all operationally indistinguishable transformations, yields the quotient OPT, which retains only physically meaningful distinctions. Throughout this work, we work in the quotient OPT, and write $\N=\M$ whenever they belong to the same equivalence class under \eqref{Eq:Equivalent_transforamtion}.  Also, from now on, we drop $\mathsf{Prob}$, implicitly treating scalars $\mathsf{Transf}(I\to I)$ as their associated probabilities. For example, $\Tr(\rho)$ denotes $\mathsf{Prob}(\Tr \circ \rho)$. 

\medskip

\noindent \textbf{Multi-time process in OPT:} The notion of multi-time processes has been investigated within the OPT framework under the name of `Higher Order Causal Category (HOCC)'~\cite{Kissinger2019}, or more generally `Higher Order Process Theory (HOPT)'~\cite{Wilson2021}.   Explicitly, consider an $N$-time process with causal order $\mathbb{A}_1\preceq \mathbb{A}_2 \preceq \cdots \preceq \mathbb{A}_N$.  The process is constituted by an initial system–environment state $\sigma \in \St(A_1 E_1)$, followed by a  sequence of deterministic system–environment transformations $\mathcal{F}_n \in \Chan(A_n E_n \to A_{n+1} E_{n+1})$ for $n=1,\cdots,N-1$. At each probe region, an experimentalist can intervene by performing tests $\{\tau_{x_n}\}_{x_n \in X_n}$, where $\{\tau_{x_n}\}_{x_n \in X_n}\in \mathsf{Tests}(A_n\to A_n)$ for $n=1.\cdots, N-1$ and $\{\tau_{x_N}\}_{x_N \in X_N}\in \mathsf{Tests}(A_N\to I)$. The resulting joint probability of the temporal sequence of outcomes $(x_1,\cdots, x_N)\in X_1 \times \cdots \times X_N$ is given by

\begin{align}
&p(x_1, \dots, x_N)
 \nonumber \\
 &=
(\Tr_{E_N} \otimes \tau_{x_N})
\circ \mathcal{F}_{N-1}
\circ \cdots
\circ \mathcal{F}_1
\circ (\id_{E_1} \otimes \tau_{x_1})
\circ \sigma.
\end{align}
All our results in quantum theory can be succinctly expressed within the multi-time process framework of HOPT.

\medskip
\noindent \textbf{Vector space structure in OPT:
} For any system 
$A$, we can define the real vector space $\St_\mathbb{R}(A) := \operatorname{Span}_{\mathbb{R}}(\St(A))$ generated by the state space, and $\eff_\mathbb{R}(A) := \operatorname{Span}_{\mathbb{R}}(\eff(A)) $generated by the effect space. Their complexifications are then given by the standard construction $\St_\mathbb{C}(A) := \St_\mathbb{R}(A) \oplus i \St_\mathbb{R}(A)$ and $\eff_\mathbb{C}(A) := \eff_\mathbb{R}(A) \oplus i \eff_\mathbb{R}(A)$. $\St_{\mathbb{R}}(A)$ and $\eff_\mathbb{R}(A)$ are respectively generated by the positive convex cones $\St_+(A) := \left\{ \sum_i a_i \rho_i \mid a_i \ge 0,\ \rho_i \in \St(A) \right\}$, and $\eff_+(A) := \left\{ \sum_i a_i e_i \mid a_i \ge 0,\ e_i \in \eff(A) \right\}$ via $\St_\mathbb{R}(A) = \St_+(A) - \St_+(A)$ and $\eff_\mathbb{R}(A) = \eff_+(A) - \eff_+(A)$. 

The parallel composition of states $\rho \otimes \sigma$, for $\rho \in \St(A)$ and $\sigma \in \St(B)$, span $\St_\mathbb{R}(A) \otimes \St_\mathbb{R}(B)$. As these are valid states in $\St(A\otimes B)$, we have $\St_\mathbb{R}(A) \otimes \St_\mathbb{R}(B) \subseteq \St_\mathbb{R}(A \otimes B)$. When this inclusion is an equality, the theory is said to satisfy \textit{local tomography}, also known as \textit{local discriminability}~\cite{hardy2001quantum, hardy2011reformulating,chiribella2010probabilistic,Barrett2007,DAriano2016, Chiribella2021}. In this work, we assume local tomography.

In general, to characterize a transformation $\N\in \mathsf{Transf}(A\to B)$ its action must be checked on all states of $A\otimes R$ for every auxiliary system $R$, as in Eq.~\eqref{Eq:Equivalent_transforamtion}. Each such extension gives rise to a family of linear maps $\widehat{\mathcal{N} \otimes \id_R} \in \mathsf{Lin}( \St_\mathbb{R}(A \otimes R), \St_\mathbb{R}(B \otimes R))$, where $\mathsf{Lin}(V,W)$ denotes the vector space of all linear maps from vector spaces $V$ to $W$. Thus, a transformation $\N$ is represented by the direct sum of its local extensions
\begin{align}
L_{\mathcal{N}} := \bigoplus_{R} \widehat{\mathcal{N} \otimes \id_R},
\end{align}
and the vector space of transformations satisfies
\begin{align}
\mathsf{Transf}_{\mathbb{R}}(A \to B) \simeq \bigoplus_{R} \mathsf{Lin}\big(\St_\mathbb{R}(A \otimes R),\, \St_\mathbb{R}(B \otimes R)\big). 
\end{align}
When local tomography is respected, however, life is simpler: the transformation $\N$ is uniquely determined by its action on a single system $A$. In that case, 

\begin{align}
\mathsf{Transf}_{\mathbb{R}}(A \to B) \simeq \mathsf{Lin}\big(\St_\mathbb{R}(A),\, \St_\mathbb{R}(B)\big).
\end{align}
Subsequently, one can define the complexification of the transformation space: $\mathsf{Transf}_{\mathbb{C}}(A\to B):=\mathsf{Transf}_{\mathbb{R}}(A \to B)+ i \mathsf{Transf}_{\mathbb{R}}(A \to B)$. 

In quantum theory, $\St_\mathbb{C}(A)$ and $\eff_\mathbb{C}(A)$ both correspond to the space of bounded operators $\mathcal{L}(\H _A)$ on a Hilbert space $\H _A$, and $\mathsf{Transf}_{\mathbb{C}}(A\to B)$
is the space of linear maps $\mathsf{Lin}(\mathcal{L}(\H _A), \mathcal{L}(\H _B))$.

\subsection{Assumptions on OPT }\label{Sec:OPT_Assumptions}

\noindent \textbf{Finite dimensionality, convexity and closure}.
For all systems $A$ and $B$, we assume that the vector space $\mathsf{Transf}_{\mathbb{C}}(A\to B)$ is finite dimensional. Moreover, the set of physical transformations $\mathsf{Transf}(A\to B)$ is assumed to be convex and closed in an operationally meaningful topology. In OPT, such a topology is typically induced by the operational norm~\cite{chiribella2010probabilistic}, which reduces to the trace norm in quantum theory. We adopt this topology without further specification.

\medskip

\noindent \textbf{Causality Axiom}. Classical and quantum theories admit a unique way to discard systems. In classical theory, discarding $B$ from a join distribution $P_{AB}$ gives the marginal $P_A(a)=\sum_bP_{AB}(a,b)$. In quantum theory, it corresponds to the partial trace: $\rho_{A}=\Tr_{B}\rho_{AB}$. In OPT,  the existence of a unique discarding operation is not guaranteed; it is imposed as the \emph{Causality axiom} \cite{Chiribella_2011,chiribella2010probabilistic}, which is intricately linked to the no-signaling principle for bipartite states. It ensures that the marginal state of a subsystem is well-defined and independent of how the rest of the system is discarded.  We restrict our discussion to theories satisfying this axiom and denote the discarding map on system $B$ by $\Tr_{B}$, following the quantum notation.

\medskip

\noindent \textbf{Local Tomography:
} We also impose the axiom of \emph{local tomography} (or \emph{local discriminability})~\cite{hardy2001quantum, hardy2011reformulating,chiribella2010probabilistic,Barrett2007,DAriano2016, Chiribella2021}.  It states that a composite state is uniquely determined by the statistics of all local measurements on its subsystems. Equivalently, two joint states $\rho, \sigma \in \St_1(A \otimes R)$ are equal if and only if they yield the same measurement statistics for every pair of local effects $e \in \eff(A)$ and $f \in \eff(R)$: 
\begin{align}
 &(e_A \otimes f_R)(\rho_{AR}) = (e_A \otimes f_R)(\sigma_{AR}) \quad \forall e \in \eff(A), \ f \in \eff(R) \nonumber \\ &\Longrightarrow \quad \rho_{AR} = \sigma_{AR}.   
\end{align}

Local tomography significantly simplifies the structure of the theory's state spaces. For instance, it implies that the underlying vector space $\St_{\mathbb{R}}(AR)$ for a composite system $A\otimes R$ is isomorphic to the tensor product of the local vector spaces: $\St_{\mathbb{R}}(AR)=\St_{\mathbb{R}}(A)\otimes \St_{\mathbb{R}}(R)$. In contrast, without local tomography, we only have a strict inclusion $\St_{\mathbb{R}}(AR)\subsetneq \St_{\mathbb{R}}(A)\otimes \St_{\mathbb{R}}(R)$. 

As a consequence of local tomography, any composite state $\tau_{AR}$ can be written as an affine combination of product states: 
\begin{align}
 \tau_{AR}:=\sum_k c_k \ \rho_{k}\otimes \sigma_{k} \label{Eq:affineState},   
\end{align}
where $\rho_{k}\in \St(A)$ and $\sigma_{k}\in \St(R)$, and $c_k \in \mathbb{R}$ with $\sum_k c_k=1$. 

Moreover, local tomography also implies that two physical transformations $\mathcal{N}, \mathcal{M}: \mathsf{Transf}(A_I\to A_O)$ are equal if and only if they agree on all input states---without needing to check their action on ancillary inputs as in Eq.~\eqref{Eq:Equivalent_transforamtion}. That is, 

\begin{align}
    \mathcal{N} = \mathcal{M} \quad \Longleftrightarrow \quad \mathcal{N}(\rho) = \mathcal{M}(\rho) \quad \forall \rho \in \St(A_I)\label{Eq:local_input}.
\end{align}
Eqs.~\eqref{Eq:affineState} and \eqref{Eq:local_input} immediately extend the above argument to transformations with multiple inputs and outputs: product input states are sufficient to distinguish any two such transformations. For any $\mathcal{N}, \mathcal{M}: \mathsf{Transf}(A_1 \otimes \cdots \otimes A_n \to B_1 \otimes \cdots \otimes B_m)$,
\begin{align}
    \mathcal{N} {=} \mathcal{M}\  {\Longleftrightarrow}\  \mathcal{N}(\rho_1{\otimes} \rho_2{\otimes} \cdots{\otimes}\rho_n) {=} \mathcal{M}(\rho_1{\otimes} \rho_2{\otimes} \cdots{\otimes}\rho_n), \label{Eq:Equality_with_product}
\end{align}
for all states $\rho_i\in \St(A_i)$. It is easy to see that the above property holds for  any linear map in $\mathsf{Transf}_{\mathbb{C}}(A_1 \otimes \cdots \otimes A_n \to B_1 \otimes \cdots \otimes B_m)$:

\begin{lemma}\label{lem:local-tomography-linear}
In a locally tomographic OPT, for any two linear maps $\mathcal{R}, \mathcal{S} \in  \mathsf{Transf}_{\mathbb{C}}(A_1 \otimes \cdots \otimes A_n \to B_1 \otimes \cdots \otimes B_m)$,
\begin{align}
\mathcal{R} = \mathcal{S} \quad \Longleftrightarrow \quad \mathcal{R}(\rho_1 \otimes \cdots \otimes \rho_n) = \mathcal{S}(\rho_1 \otimes \cdots \otimes \rho_n)
\end{align}
for all product states $\rho_k \in \St(A_k)$.
\end{lemma}

\begin{proof}
The implication $\Rightarrow$ is trivial. For the other direction ($\Leftarrow$), 
let the set of physical transformations $\{\mathcal{N}_i\}_i$ form a basis of the vector space $\mathsf{Transf}_{\mathbb{C}}(A_1 \otimes \cdots \otimes A_n \to B_1 \otimes \cdots \otimes B_m)$. Hence, any linear map $\mathcal{R}$ can be written as $\mathcal{R} = \sum_i \alpha_i \mathcal{N}_i$, and similarly $\mathcal{S} = \sum_i \beta_i \mathcal{N}_i$, with $\alpha_i,\beta_i \in \mathbb{C}$. Suppose they agree on all product states, then
\begin{align}
  \sum_i (\alpha_i - \beta_i) \mathcal{N}_i(\rho_1 \otimes \cdots \otimes \rho_n) = 0,  
\end{align}
for all product states $\rho_k \in \St(A_k)$. Since the theory is locally tomographic, any state of the composite system can be written as an affine combination of product states. Hence, for any arbitrary ancillary system $R$ and any joint state $\sigma \in \St(A_1 \otimes \cdots \otimes A_n \otimes R)$,
\begin{align}
\sum_i (\alpha_i - \beta_i) (\mathcal{N}_i \otimes \id_R)(\sigma) = 0.
\end{align}
This implies that $\sum_i (\alpha_i - \beta_i) \mathcal{N}_i = 0$ as a linear map. Since $\{\mathcal{N}_i\}_i$ form a basis, we must have $\alpha_i = \beta_i$ for all $i$. Therefore, $\mathcal{R} = \mathcal{S}$.
\end{proof}
Lemma~\ref{lem:local-tomography-linear} is a key tool that underpins several propositions in the main text. Local tomography is a non-trivial axiom. It is violated by several physical theories, including fermionic~\cite{DAriano2014_Fermionic,DAriano2014_Feynman}, anyonic~\cite{vidal2024entanglementasymm}, real quantum~\cite{wootters1990local,Hardy_2011,Aleksandrova2013}, and quaternionic quantum theories~\cite{hardy2009foliableoperation,Barnum2014,Barnum2014_HigherOrder}. Conversely, it holds in classical theory, deterministic classical theory~\cite{chiribella2010probabilistic, DAriano2020}, standard quantum theory, and theories based on minimal and maximal tensor products (e.g., the PR-box model)~\cite{Popescu1994,Barrett2007,Plvala2023}. The principle has also been tested experimentally, and the results are consistent with its validity, showing no violation~\cite{lismer2025experimentaltest}.

\section{Multi-time states in OPT}

The vector space structure of OPT discussed in the previous section, along with all the notation introduced there, extends directly to multi-time states as defined in Definitions~\ref{def:MTQS} and~\ref{Def:QSOTGen}, the temporal extension map in Definition~\ref{Def:tempExt}, and the corresponding Axioms~\ref{axm:PropMTQS} and \ref{axm:spatial-marginal-consistency}. Moreover, under the assumptions in Section~\ref{Sec:OPT_Assumptions}, Theorems~\ref{Thm:StructureThm_MarkovQSOT}, \ref{Thm:NMQSOT}, and \ref{thm:semicausal} continue to hold in this broader framework. Note, in the OPT context, the multi-time quantum process in Theorem~\ref{Thm:NMQSOT} is the multi-time process ($n$-comb) in the HOCC framework~\cite[Definition 6.6]{Kissinger2019}, and the semicausal quantum channel in Theorem~\ref{thm:semicausal} becomes a semicausal deterministic transformation in OPT, referred to as a `one-way signaling process' in Ref.~\cite[Definition 2.9]{Kissinger2019}.

To establish our results in full generality, all statements and proofs of the propositions and theorems in Secs.~\ref{sec:AxiomaticFramework},
\ref{sec:MarkovCaopatible}, and~\ref{sec:GenMTQSnonMarkov} of the main text are formulated in the OPT setting.

\subsection{Proof of Proposition~\ref{Prop:stLinEquivLinTemExt}} \label{App:proof_stLinEquivLinTemExt}
We first state a lemma that applies to OPTs regardless of whether local tomography is assumed. 

\begin{lemma}\label{lem:convex-base-extension-composite}
 Let $f: \St(A) \to W$, with $W$ being a vector space, be a convex-linear map, which has the property that for every auxiliary system $R$, the map $f \otimes \id_R$ is convex-linear on $\St(A \otimes R)$. Then $f$ extends uniquely to a real-linear map $F: \St_{\mathbb{R}}(A) \to W$ such that for any auxiliary system $R$, this extension satisfies
\begin{align}
(F \otimes \id_R)(v) = \sum_i a_i (f \otimes \id_R)(\rho_i)
\end{align}
for any $v = \sum_i a_i \rho_i \in \St_{\mathbb{R}}(A\otimes R)$.
\end{lemma}

\begin{proof}
Fix an arbitrary auxiliary system $R$. Assume two different decompositions
\begin{align}
&v = \sum_i a_i \rho_i = \sum_j b_j \sigma_j,
\end{align}
with $\rho_i, \sigma_j \in \St(A \otimes R)$ and $a_i,b_i \in \mathbb{R}$. Applying, the unique deterministic effect, fixed by the causality axiom, $\Tr_{AR} := \Tr_A \otimes \Tr_R$, which extends to a linear functional on $\St_{\mathbb{R}}(AR)$, we have
\begin{align}
\Tr_{AR}(v) {=} \sum_i a_i = \sum_j b_j. \label{Eq:GPT_Tr_on_span_composite}
\end{align}

We now define the extension of $f \otimes \id_R$. For any $v \in \St_{\mathbb{R}}(AR)$, choose a decomposition $v = \sum_i a_i \rho_i$ with $\rho_i \in \St(A R)$, and define
\begin{align}
(F \otimes \id_R)(v) := \sum_i a_i (f \otimes \id_R)(\rho_i).
\end{align}

\medskip
\noindent \textbf{Well-definedness.} 
We must show that $(F \otimes \id_R)(v)$ does not depend on the chosen decomposition. Suppose $v$ has two different decompositions
\begin{align}
&v = \sum_i a_i \rho_i = \sum_j b_j \sigma_j, \nonumber \\
&{\Rightarrow} \quad \sum_{i: a_i \ge 0} a_i \rho_i + \sum_{j: b_j < 0} |b_j| \sigma_j
= \sum_{j: b_j \ge 0} b_j \sigma_j + \sum_{i: a_i < 0} |a_i| \rho_i. \label{Eq:GPT_pos_neg_rearrangement_composite}
\end{align}
Here in the second equation, we separated and rearranged the non-negative ($a_i\ge 0$, $b_i\ge 0$) and negative ($a_i< 0$, $b_i <0$) coefficients. Furthermore, rearranging these non-negative and negative coefficients in Eq.~\eqref{Eq:GPT_Tr_on_span_composite}, we define
\begin{align}
\lambda := \sum_{i: a_i \ge 0} a_i + \sum_{j: b_j < 0} |b_j|
= \sum_{j: b_j \ge 0} b_j + \sum_{i: a_i < 0} |a_i|.
\end{align}
Dividing both sides of Eq.~\eqref{Eq:GPT_pos_neg_rearrangement_composite} by $\lambda$ gives a convex combination of states on both sides. Since $f \otimes \id_R$ is convex-linear, applying it to these convex combinations yields
\begin{align}
&\sum_{i: a_i \ge 0} \frac{a_i}{\lambda} (f \otimes \id_R)(\rho_i)
+ \sum_{j: b_j < 0} \frac{|b_j|}{\lambda} (f \otimes \id_R)(\sigma_j) \nonumber \\
&{=} \sum_{j: b_j \ge 0} \frac{b_j}{\lambda} (f \otimes \id_R)(\sigma_j)
+ \sum_{i: a_i < 0} \frac{|a_i|}{\lambda} (f \otimes \id_R)(\rho_i).
\end{align}
Canceling $\lambda$ and rearranging the rest gives
\begin{align}
\sum_i a_i (f \otimes \id_R)(\rho_i) = \sum_j b_j (f \otimes \id_R)(\sigma_j).
\end{align}
Thus $(F \otimes \id_R)(v)$ is independent of the decomposition, so it is well-defined.

\medskip
\noindent  Linearity of this map can be easily verified. For uniqueness, we have the following:

\medskip

\noindent \textbf{Uniqueness.} 
Suppose $G: \St_{\mathbb{R}}(A) \to W$ is another real-linear extension of $f$. Then, for any $v = \sum_i a_i \rho_i \in \St_{\mathbb{R}}(AR)$,
\begin{align}
(G \otimes \id_R)(v) &= \sum_i a_i (G \otimes \id_R)(\rho_i) \nonumber \\ 
&=\sum_i a_i (f \otimes \id_R)(\rho_i) = (F \otimes \id_R)(v).
\end{align}
Thus $G\otimes \id_R$ and $F \otimes \id_R$ agree on all elements of $\St_{\mathbb{R}}(AR)$, and all extensions $R$. Hence $G = F$.
\end{proof}

\medskip

\noindent \textbf{Complexification.} 
Having established the unique real-linear extension $F: \St_{\mathbb{R}}(A) \to W$, we can naturally extend it to the complex vector space $\St_\mathbb{C}(A)$, with the target space similarly complexified to $W^{\mathbb{C}} := W \oplus i W$. For any $v = v_1 + i v_2 \in \St_\mathbb{C}(A)$, with $v_1, v_2 \in \St_\mathbb{R}(A)$, we define the complex-linear extension $F^{\mathbb{C}}: \St_\mathbb{C}(A)\to W^{\mathbb{C}}$ by $F^{\mathbb{C}}(v) := F(v_1) + i F(v_2).$
The well-definedness, linearity, and uniqueness of $F^{\mathbb{C}}$ follow directly from the corresponding properties of $F$ on the real vector space.

Note that the above proof does not require local tomography, except that the argument only applies to convex-linear maps that stay convex-linear when applied locally alongside an arbitrary auxiliary system. Clearly, this property holds for every convex-linear map in a locally tomographic OPT. Consequently, for locally tomographic theories, it suffices to consider the proof with $R$ taken to be the trivial system, as in the quantum case treated in Ref.~\cite{lie2024quantum}.

\medskip

\noindent \textbf{Proof of Proposition~\ref{Prop:stLinEquivLinTemExt} in an OPT satisfying assumptions in Subsec. \ref{Sec:OPT_Assumptions}} Clearly, when a temporal extension map exists as in Definition~\ref{Def:tempExt}, then the star product must be convex-linear in states. Conversely, for a convex-linear star product, we define a convex-linear map $C_{\E_n}$ such that $C_{\E_n}(\rho) = \{\E_n\}\star \rho$. In an OPT, satisfying assumptions in Subsec. \ref{Sec:OPT_Assumptions},  this map satisfies the conditions of Lemma~\ref{lem:convex-base-extension-composite} and therefore extends uniquely to a linear map on the full span.

\subsection{Proof of Proposition~\ref{Pro:ReconAndComposition}} \label{App:Pro:ReconAndComposition}

\begin{figure}
    \centering
    \includegraphics[width=1\linewidth]{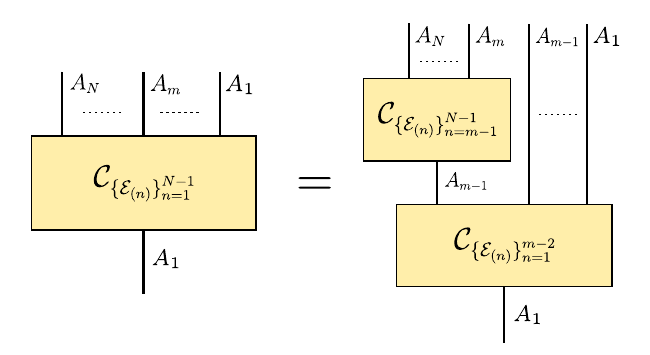}
    \caption{\textbf{Decomposable structure of the temporal extension map.} The axiom of Reconstructibility imposes above structure on the temporal extension map (Eq.~\eqref{Eq:chain-likeRecon3}). This decomposition mirrors the Markov property that future evolution is obtained by acting only on the present.}
    \label{fig:Reconchainlikeproposition}
\end{figure}

\noindent \textbf{Proof of part 1: } Consider a sequence of channels $\{ \E_{N-1}, \cdots, \E_2, \E_1 \} \equiv \{ \E_n \}_{n=1}^{N-1}$ with $\E_n \in \Chan (A_n \to A_{n+1})$. Suppose there exists a linear reconstructing map, $\mathcal{R}_m \in \mathsf{Transf}_{\mathbb{C}}(A_m\rightarrow A_N\cdots A_{n+1}A_m)$, such that
\begin{equation}
    \mathcal{R}_m\circ\Tr_{A_N\cdots A_{m+1}}\left[\{\E_n\}_{n=1}^{N-1}\star\rho_{A_1}\right ]=\{\E_n\}_{n=1}^{N-1}\star\rho_{A_1}, \forall m
\end{equation}
Now, choosing $m=1$, i.e., applying $\mathcal{R}_1$ after discarding $A_2,\cdots,A_N$ (applying unique deterministic effect $\Tr_{A_2,\cdots,A_N}$), we have, 
\begin{align}
\mathcal{R}_1\circ\Tr_{A_N\cdots A_2}\left[\{\E_n\}_{n=1}^{N-1}\star\rho_{A_1}\right ]&=\{\E_n\}_{n=1}^{N-1}\star\rho_{A_1},\\
\implies \mathcal{R}_1(\rho_{A_1})&=\{\E_n\}_{n=1}^{N-1}\star\rho_{A_1}. \label{Eq:R1->star} 
\end{align}
The second equality is due to the marginal consistency of the star product (see Definition~\ref{def:MTQS}). Since $\mathcal{R}_1$ is linear, the above equation implies the star product is convex-linear in states. However, due to Proposition~\ref{App:proof_stLinEquivLinTemExt}, convex linearity of star product uniquely identifies a temporal extension map $\C_{\{ \E_n \}_{n=1}^{N-1}} \in \mathsf{Transf}_{\mathbb{C}}( A_1 \to A_N \cdots A_1)$ such that
\begin{align}
\mathcal{C}_{\{\mathcal{E}_n\}_{n=1}^{N-1}}(\rho_{A_1})
=
\{\mathcal{E}_n\}_{n=1}^{N-1} \star \rho_{A_1}. \label{Eq:temp->star}
\end{align}

Thus, for all $m< N$, we can express the star product in terms of the temporal extension map. Now take any arbitrary cut $m$ ($m<N$). Then, for all $\rho \in \St(A_1)$, we have,
\begin{align}
     \mathcal{R}_{m-1}\circ\Tr_{A_N\cdots A_{m}}\left[\{\E_n\}_{n=1}^{N-1}\star\rho_{A_1}\right ]&=\{\E_n\}_{n=1}^{N-1}\star\rho_{A_1},\nonumber\\
     \implies \mathcal{R}_{m-1}\circ\left[\{{\E_n\}_{n=1}^{m-2}\star\rho_{A_1}}\right]&=\{\E_n\}_{n=1}^{N-1}\star\rho_{A_1},\nonumber\\
     \implies \mathcal{R}_{m-1}\circ\C_{\{\E_n\}_{n=1}^{m-2}}(\rho_{A_1})&=\C_{\{\E_n\}_{n=1}^{N-1}}(\rho_{A_1}),\nonumber\\
     \implies \mathcal{R}_{m-1}\circ\C_{\{\E_n\}_{n=1}^{m-2}}&=\C_{\{\E_n\}_{n=1}^{N-1}}.
\end{align}

The first equality is due to the reconstructibility axiom, the second one is due to the marginal consistency of the star product (see Definition~\ref{def:MTQS}), and in the third equality, we express the star products in terms of the temporal extension maps (see Definition~\ref{Def:tempExt}). The fourth equality is due to local tomography (see Lemma~\ref{lem:local-tomography-linear}). Discarding $A_1, A_2,\cdots,A_{m-1}$ (i.e. applying the unique deterministic effect $\Tr_{A_1, A_2,\cdots,A_{m-1}}$) on both sides of the above equation, and applying marginal consistency of the temporal extension maps (Eq.~\eqref{Eq:D1}), we get,
\begin{equation}
    \mathcal{R}_{m-1}\circ\E_{m-2}\cdots\circ\E_1=\C_{\{\E_n\}_{n=m-1}^{N-1}}\circ\E_{m-2}\cdots\circ\E_1
\end{equation}
If the reconstructing map is independent of \emph{all possible} past channels $\{\mathcal{E}_1, \dots, \mathcal{E}_{n-1}\}$, as stipulated in Axiom~\ref{axm:PropMTQS}, then choosing reversible channels, the reconstructing map $\mathcal{R}_{m-1}$ is uniquely determined and must equal the temporal extension map $\mathcal{C}_{\{\E_n\}_{n=m-1}^{N-1}}$ itself. Hence, we have for all $m<N$,
\begin{equation}
    \mathcal{C}_{\{\E_n\}_{n=m-1}^{N-1}}\circ\C_{\{\E_n\}_{n=1}^{m-2}}=\C_{\{\E_n\}_{n=1}^{N-1}}\label{Eq:chain-likeRecon3}
\end{equation}
See Fig.~\ref{fig:Reconchainlikeproposition} for a pictorial depiction. Applying Eq.~\eqref{Eq:chain-likeRecon3} recursively for
$m = N-1, N-2, \dots, 2$, and using the definition of the star product, we
obtain the single chain shown in part~1 of proposition~\ref{Pro:ReconAndComposition}, completing the proof
of part~1.

\medskip

\noindent \textbf{Proof of part 2: }  Consider the particular case of $N=3$ of a temporal extension map, i.e., we consider the map $\C_{\{\E_1,\E_2\}} \in \mathsf{Transf}_{\mathbb{C}}(A_1\to A_3\otimes A_2\otimes A_1)$. Due to \eqref{Eq:chain-likeRecon3}, we have for all systems $A_1,A_2$ and $A_3$ and channels ${\E_1}: A_1\to A_2$ and $\E_2: A_2\to A_3$
\begin{align}
 \C_{\{\E_1,\E_2\}} =  \C_{\E_2}\circ \C_{\E_1}. 
\end{align}
Subsequently, discarding $A_2$ on both sides, and assuming compositionality (Q3 in Axiom ~\ref{axm:PropMTQS}) holds, with the temporal extension map satisfying Eq.~\eqref{Eq:D2}, we have the following chain of equations:
\begin{align}
 &\Tr_{A_2}\circ \C_{\{\E_1,\E_2\}} =  [\Tr_{A_2}\circ \C_{\E_2}]\circ \C_{\E_1} \qquad \forall \E_1, \E_2\nonumber \\
 & \implies \C_{\{\E_2\circ\E_1\}} =  {\E_2}\circ \C_{\E_1}, \forall \E_1, \E_2. 
\end{align}
In particular, choosing $\E_1=\id$ so that $\mathcal{C}_{\mathcal{I}}:A_1\rightarrow A_1''\otimes A_1'$ (where $A_1\simeq A_1'\simeq A_1''$) and $\E_2 =\E$ being an arbitrary channel, we have for all channels $\E$
\begin{align}
 & \C_{\E} =  (\E_{A_1''\rightarrow A_2} \otimes \id)\circ \C_{\id}. 
\end{align}
This proves part~2 of proposition~\ref{Pro:ReconAndComposition}.

\section{Multi-time state in the square-bit boxworld}

\subsection{Proof of Theorem~\ref{thm:unique-broadcasting-map}} \label{App:unique-broadcasting-map}

We structure the proof as follows. First, we establish the representation-theoretic decomposition of the relevant state spaces. Second, we translate the covariance condition into an invariance problem in a tensor product space. Third, we characterize all invariant vectors in this tensor product. Finally, we translate the resulting invariant vectors back to linear maps and impose the remaining axioms to single out the unique broadcasting map.

\subsubsection*{Step 1: Vector Representation of Linear Maps}

We first introduce some useful facts from linear algebra, which are closely related to the Choi-Jamiołkowski isomorphism in quantum theory. For the sake of completeness, we elaborate on them in some detail.

We use the standard isomorphism $\mathsf{Lin}(V,W) \cong V^* \otimes W$, where $V^*$ is the dual space of $V$. For bases $\{v_i\}$ of $V$ and $\{w_j\}$ of $W$, a linear map $\mathcal{N}$ with $\mathcal{N}(v_i) = \sum_j N_{ji} w_j$ corresponds to the vector
\begin{align*}
X_{\mathcal{N}} = \sum_i v_i^* \otimes \mathcal{N}(v_i) = \sum_{i,j} N_{ji} \, v_i^* \otimes w_j \in V^* \otimes W.
\end{align*}
Conversely, any vector $X = \sum_{i,j} X_{ji} \, v_i^* \otimes w_j$ defines $\mathcal{N}_X \in \mathsf{Lin}(V,W)$ by $\mathcal{N}_X(\tilde{v}) = \sum_{i,j} v_i^*(\tilde{v})X_{ji} w_j$ for any $\tilde{v}\in V$. We now introduce a useful lemma.

\begin{lemma} \label{Lem:dual_map_Lemma}
For linear maps $\Phi_1: U \to V$, $\mathcal{N}: V \to W$, and $\Phi_2: W \to Z$, the vector representation of the composition $\Phi_2 \circ \mathcal{N} \circ \Phi_1: U \to Z$ is given by
\begin{align*}
X_{\Phi_2 \circ \mathcal{N} \circ \Phi_1} = (\Phi_1^* \otimes \Phi_2)(X_{\mathcal{N}}),
\end{align*}
where $\Phi_1^*: V^* \to U^*$ is the dual map defined by $[\Phi_1^*(\psi)](\mu) = \psi(\Phi_1(\mu))$ for all $\psi \in V^*$ and $\mu \in U$.
\end{lemma}

\begin{proof}
Let $\{\mu_i\}$, $\{v_j\}$, and $\{w_k\}$ be bases of $U$, $V$, and $W$, respectively. Write $\Phi_1(\mu_i) = \sum_j (\Phi_1)_{ji} v_j$ and $\mathcal{N}(v_j) = \sum_k N_{kj} w_k$. Then $X_{\mathcal{N}} = \sum_{j,k} N_{kj} v_j^* \otimes w_k$. The dual map $\Phi_1^*$ acts on the dual basis by $\Phi_1^*(v_j^*) = \sum_i (\Phi_1)_{ji} \mu_i^*$; this follows since $[\Phi_1^*(v_j^*)](\mu_i) = v_j^*(\Phi_1(\mu_i)) = (\Phi_1)_{ji}$. Applying $\Phi_1^* \otimes \Phi_2$ to $X_{\mathcal{N}}$ gives
\begin{align*}
(\Phi_1^* \otimes \Phi_2)(X_{\mathcal{N}}) 
&= \sum_{j,k} N_{kj} \Phi_1^*(v_j^*) \otimes \Phi_2(w_k) \\
&= \sum_{j,k} N_{kj} \left(\sum_i (\Phi_1)_{ji} \mu_i^*\right) \otimes \Phi_2(w_k) \\
&= \sum_{i,k} \left(\sum_j (\Phi_1)_{ji} N_{kj}\right) \mu_i^* \otimes \Phi_2(w_k).
\end{align*}
But $(\Phi_2 \circ \mathcal{N} \circ \Phi_1)(\mu_i) = \Phi_2\left(\sum_j (\Phi_1)_{ji} \mathcal{N}(v_j)\right) = \sum_{j,k} (\Phi_1)_{ji} N_{kj} \Phi_2(w_k)$, so the coefficients match those of $X_{\Phi_2 \circ \mathcal{N} \circ \Phi_1}$. Hence the identity holds.
\end{proof}

\subsubsection*{Step 2: Applying covariance condition }

Next, we note that the covariance condition of the broadcasting map $\mathcal{B}$ is only to be checked on the generators $R$ and $M$ of $D_4$. In particular, we have

\begin{lemma}\label{lem:map_to_invariant_vector}
The covariance condition of the broadcasting map $\mathcal{B}$ is equivalent to invariance of its associated vector $X_{\mathcal{B}}$ under the action of the generators of $D_4$:
\begin{align}
&(R \otimes R) \circ \mathcal{B}  = \mathcal{B} \circ R
 {\iff }
[(R^{-1})^* \otimes R \otimes R](X_{\mathcal{B}}) = X_{\mathcal{B}}, \\
&(M \otimes M) \circ \mathcal{B}  {=} \mathcal{B} \circ M
{\iff}
[(M^{-1})^* {\otimes} M {\otimes} M](X_{\mathcal{B}}){ =} X_{\mathcal{B}}.
\end{align}

\end{lemma}

\begin{proof}
Covariance under rotation: $(R \otimes R) \circ \mathcal{B} = \mathcal{B} \circ R$ is equivalent to $(R \otimes R) \circ \mathcal{B} \circ R^{-1} = \mathcal{B} $. By Lemma~\ref{Lem:dual_map_Lemma}, the vector representation of the composition $(R \otimes R) \circ \mathcal{B} \circ R^{-1}$ is
\begin{align*}
X_{(R \otimes R) \circ \mathcal{B} \circ R^{-1}}
= [(R^{-1})^* \otimes (R \otimes R)](X_{\mathcal{B}})
= [(R^{-1})^* \otimes R \otimes R](X_{\mathcal{B}}).
\end{align*}
Since the vector representation is an isomorphism, the equality of maps is equivalent to the equality of their associated vectors. The proof for $M$ is identical.
\end{proof}
Consequently, $\mathcal{B}$ is $D_4$-covariant if and only if $X_{\mathcal{B}}$ is invariant under the action of both generators, and hence under the whole group $D_4$.

\subsubsection*{Step 3: Irreducible Decomposition of $\mathsf{St}_{\mathbb{C}}(A)$}

Recall that the state space $\mathsf{St}_{\mathbb{C}}(A)$ of a square bit has basis $\{\rho_{00}, \rho_{10}, \rho_{01}\}$, where $\rho_{00}, \rho_{10}, \rho_{01}$ are pure states of the square bit in boxworld. We find it convenient to work in a different basis to find the irreducible subrepresentations of the reversible transformations:
\begin{align}
\omega &:= \frac{\rho_{01} + \rho_{10}}{2}, \\
v &:= \rho_{01} - \rho_{10}, \\
u &:= \rho_{00} - \rho_{11},
\end{align}
where $\rho_{11}$ is the fourth pure state of the square bit, determined by the relation $\rho_{11} = \rho_{01} + \rho_{10} - \rho_{00}$. Then $\mathsf{St}_{\mathbb{C}}(A) = \operatorname{Span}_{\mathbb{C}}\{\omega, v, u\}$.

\begin{lemma}
The vector space $\mathsf{St}_{\mathbb{C}}(A)$ decomposes into irreducible subrepresentations as
\begin{align*}
\mathsf{St}_{\mathbb{C}}(A) \cong \mathbf{1} \oplus E,
\end{align*}
where 
\begin{align*}
\mathbf{1} := \operatorname{Span}_{\mathbb{C}}\{\omega\}, \qquad E := \operatorname{Span}_{\mathbb{C}}\{u,v\}.
\end{align*}
\end{lemma}

\begin{proof}
We first show that $\mathbf{1}$ and $E$ are subrepresentations by verifying that they are closed under the action of the generators $R$ and $M$ of $D_4$. From the action of $R$ and $M$ on the pure states $\rho_{00}, \rho_{10}, \rho_{01}, \rho_{11}$, one obtains
\begin{align}
R(\omega) &= \omega, & R(u) &= -v, & R(v) &= u, \\
M(\omega) &= \omega, & M(u) &= -u, & M(v) &= v.
\end{align}
Thus $\omega$ is fixed by both generators, so $\mathbf{1}$ is a one-dimensional subrepresentation. Moreover, $R$ and $M$ map $E = \operatorname{Span}\{u,v\}$ to itself, so $E$ is also a subrepresentation.

It remains to show that these subrepresentations are irreducible. Since $\mathbf{1}$ is one-dimensional, it is trivially irreducible. For $E$, suppose there were a one-dimensional invariant subspace $\operatorname{Span}_{\mathbb{C}}\{au + bv\}$ with $a,b \in \mathbb{C}$. If $b = 0$, then $R(au) = aR(u) = -av \notin \operatorname{Span}_{\mathbb{C}}\{u\}$, so the subspace would not be invariant; similarly, if $a = 0$, then $R(bv) = bu \notin \operatorname{Span}_{\mathbb{C}}\{v\}$. Hence both $a$ and $b$ must be non-zero. Now invariance under $M$ requires $M(au + bv) = \lambda(au + bv)$ for some $\lambda \in \mathbb{C}$. Using $M(u) = -u$ and $M(v) = v$, we get
\begin{align*}
-au + bv = \lambda au + \lambda bv.
\end{align*}
Since $u$ and $v$ are linearly independent, comparing coefficients yields $-a = \lambda a$ and $b = \lambda b$. As $a,b \neq 0$, this forces $\lambda = -1$ from the first equation and $\lambda = 1$ from the second, a contradiction. Thus no one-dimensional invariant subspace exists, so $E$ is irreducible.
\end{proof}

\subsubsection*{Step 4: Characterizing Invariant Vectors}

Hence, by $\mathsf{St}_{\mathbb{C}}(A) \cong \mathbf{1} \oplus E$, the invariant vector 
$X_{\mathcal{B}} \in \mathsf{St}_{\mathbb{C}}(A)^* \otimes \mathsf{St}_{\mathbb{C}}(A) \otimes \mathsf{St}_{\mathbb{C}}(A)$ 
is equivalently an invariant vector in
$(\mathbf{1} \oplus E)^* \otimes (\mathbf{1} \oplus E) \otimes (\mathbf{1} \oplus E)$, where $(\mathbf{1} \oplus E)^* = \operatorname{Span}_{\mathbb{C}}\{\omega^*, u^*, v^*\}$ is the dual space, with dual basis satisfying
\begin{align}
&\omega^*(\omega) = v^*(v) = u^*(u) = 1, \\
&\omega^*(v) = \omega^*(u) = v^*(\omega) = u^*(\omega) = v^*(u) = u^*(v) = 0.
\end{align}
In terms of the effects of the square bit, these dual basis elements are given by
\begin{align}
\omega^* = \operatorname{Tr}, \qquad
v^* = \frac{1}{2}(e_y - e_x), \qquad
u^* = \frac{1}{2}(\bar{e}_x - e_y).
\end{align}
The action of $(R^{-1})^* = (R^*)^{-1}$ and $(M^{-1})^* = M^*$ on this dual basis is given by
\begin{align}
(R^{-1})^*(\omega^*) &= \omega^*, &
(R^{-1})^*(u^*) &= -v^*, &
(R^{-1})^*(v^*) &= u^*, \\
(M^{-1})^*(\omega^*) &= \omega^*, &
(M^{-1})^*(u^*) &= -u^*, &
(M^{-1})^*(v^*) &= v^*.
\end{align}
Expanding the tensor product $(\mathbf{1} \oplus E)^* \otimes (\mathbf{1} \oplus E) \otimes (\mathbf{1} \oplus E)$ yields eight components. We first eliminate those that do not contain any non-zero invariant vectors.
\begin{lemma} 
The components $E^* \otimes \mathbf{1} \otimes \mathbf{1}$, $\mathbf{1}^* \otimes E \otimes \mathbf{1}$, $\mathbf{1}^* \otimes \mathbf{1} \otimes E$, and $E^* \otimes E \otimes E$ contain no non-zero invariant vectors under the $D_4$-action.
\end{lemma}

\begin{proof}
Finding the invariant vector is equivalent to looking for simultaneous eigenvectors of $R$ and $M$ with eigenvalue $1$. With this view, note that on $E = \operatorname{Span}_{\mathbb{C}}\{u,v\}$, the rotation $R$ acts as
$R(u) = -v, \qquad R(v) = u$.
A simple calculation reveals that the eigenvalues of $R$ on $E$ are $\pm i$, with eigenvectors
\begin{align*}
R(u + iv) = i(u + iv), \qquad R(u - iv) = -i(u - iv).
\end{align*}
Similarly, on $E^*$, the dual map $(R^{-1})^*$ has the same eigenvalues $\pm i$, with
\begin{align*}
(R^{-1})^*(u^* + iv^*) = i(u^* + iv^*), \quad (R^{-1})^*(u^* - iv^*) {=} -i(u^* - iv^*).
\end{align*}
Thus neither $R$ on $E$ nor $(R^{-1})^*$ on $E^*$ has $+1$ as an eigenvalue. Consequently, any tensor component containing a single factor of $E$ or $E^*$ while the remaining factors are trivial — i.e., $E^* \otimes \mathbf{1} \otimes \mathbf{1}$, $\mathbf{1}^* \otimes E \otimes \mathbf{1}$, and $\mathbf{1}^* \otimes \mathbf{1} \otimes E$ — cannot contain a vector fixed by $R$, since the eigenvalue of $R$ on such a component is simply the eigenvalue on the single non-trivial factor, which is $\pm i \neq 1$.

For the remaining component $E^* \otimes E \otimes E$, the covariance under rotation translates to the action $(R^{-1})^* \otimes R \otimes R$. The eigenvalues of this map are products of one eigenvalue from $(R^{-1})^*$ and two from $R$, each $\pm i$. Hence every eigenvalue is of the form $(\pm i)^3 = \pm i$, and again $1$ is not an eigenvalue. Therefore no non-zero vector in this component fixed by the rotations exists.
\end{proof}

Thus only four components can contain non-zero invariant vectors:
\begin{align*}
\mathbf{1}^* \otimes \mathbf{1} \otimes \mathbf{1}, \quad
\mathbf{1}^* \otimes E \otimes E, \quad
E^* \otimes E \otimes \mathbf{1}, \quad
E^* \otimes \mathbf{1} \otimes E.
\end{align*}

\begin{lemma}\label{Lem:invariant_Vector}
The invariant subspaces of $E \otimes E$ under $R \otimes R$ and $M \otimes M$, and of $E^* \otimes E$ under $(R^{-1})^* \otimes R$ and $(M^{-1})^* \otimes M$, are each one-dimensional, spanned respectively by
\begin{align}
u \otimes u + v \otimes v \quad \text{and} \quad u^* \otimes u + v^* \otimes v.
\end{align}
\end{lemma}

\begin{proof}
We first analyze $E \otimes E$. On $E$, the eigenvectors of $R$ are $R(u+iv) = i(u+iv)$, and $R(u-iv) = -i(u-iv)$.  Thus the $+1$-eigenspace of $R \otimes R$ on $E \otimes E$ is spanned by
\begin{align}
(u+iv) \otimes (u-iv) \quad \text{and} \quad (u-iv) \otimes (u+iv).
\end{align}
A general $+1$-eigenvector of $R \otimes R$ is therefore
\begin{align}
\alpha (u+iv) \otimes (u-iv) + \beta (u-iv) \otimes (u+iv),
\end{align}
for some $\alpha, \beta \in \mathbb{C}$. Expanding this expression gives
\begin{align}
(\alpha+\beta)(u \otimes u + v \otimes v) + i(\alpha-\beta)(u \otimes v - v \otimes u).
\end{align}
Applying $M \otimes M$ to this vector, using $M(u) = -u$ and $M(v) = v$, yields
\begin{align}
(\alpha+\beta)(u \otimes u + v \otimes v) - i(\alpha-\beta)(u \otimes v - v \otimes u).
\end{align}
For this to be invariant under $M$, we require $\alpha-\beta = 0$, so $\alpha = \beta$. Thus the invariant subspace is one-dimensional, spanned by
\begin{align}
u \otimes u + v \otimes v.
\end{align}

The argument for $E^* \otimes E$ is analogous and gives the one-dimensional invariant subspace spanned by 
\begin{align}
u^* \otimes u + v^* \otimes v.
\end{align}
\end{proof}

With Lemma~\ref{Lem:invariant_Vector} in hand, we can now explicitly write down the invariant vectors in each of the four contributing components. Collecting these contributions yields the following characterization.

\begin{corollary} \label{Cor:invariant_vector}
The most general invariant vector in 
$(\mathbf{1} \oplus E)^* \otimes (\mathbf{1} \oplus E) \otimes (\mathbf{1} \oplus E)$
is given by
\begin{align}
X_{\mathcal{B}} 
&= \alpha \, \omega^* \otimes \omega \otimes \omega \nonumber \\
&\quad + \beta \, \omega^* \otimes (u \otimes u + v \otimes v) \nonumber \\
&\quad + \gamma \, (u^* \otimes u + v^* \otimes v) \otimes \omega \nonumber \\
&\quad + \delta \, (u^* \otimes \omega \otimes u + v^* \otimes \omega \otimes v),
\end{align}
with $\alpha, \beta, \gamma, \delta \in \mathbb{C}$.
\end{corollary}

\begin{proof}
The component $\mathbf{1}^* \otimes \mathbf{1} \otimes \mathbf{1}$ is one-dimensional and clearly invariant, contributing $\alpha \, \omega^* \otimes \omega \otimes \omega$. For $\mathbf{1}^* \otimes E \otimes E$, Lemma~\ref{Lem:invariant_Vector} gives the invariant subspace spanned by $u \otimes u + v \otimes v$, yielding $\beta \, \omega^* \otimes (u \otimes u + v \otimes v)$. Similarly, for $E^* \otimes E \otimes \mathbf{1}$ and $E^* \otimes \mathbf{1} \otimes E$, the same lemma gives the invariant subspaces spanned by $u^* \otimes u + v^* \otimes v$, contributing $\gamma \, (u^* \otimes u + v^* \otimes v) \otimes \omega$ and $\delta \, (u^* \otimes \omega \otimes u + v^* \otimes \omega \otimes v)$, respectively. Combining these yields the claimed expression.
\end{proof}

\subsubsection*{Step 5: Translating back to linear maps}

We now translate the invariant vector $X_{\mathcal{B}}$ back to a linear map via the isomorphism established in Step 1.

\begin{lemma} \label{Lem:covariant_map}
The most general covariant map $\mathcal{B} \in \mathsf{Transf}_{\mathbb{C}}(A \to A'' \otimes A')$ arising from the invariant vector $X_{\mathcal{B}}$ in Corollary~\ref{Cor:invariant_vector} is given, for any $\rho \in \mathsf{St}_{\mathbb{C}}(A)$, by
\begin{align}
\mathcal{B}(\rho) 
&= \tilde{\alpha} \, \operatorname{Tr}(\rho) \, \omega \otimes \omega \nonumber \\
&\quad + \beta \, \operatorname{Tr}(\rho) \, (u \otimes u + v \otimes v) \nonumber \\
&\quad + \gamma \, \rho \otimes \omega \nonumber \\
&\quad + \delta \, \omega \otimes \rho,
\end{align}
with $\tilde{\alpha}, \beta, \gamma, \delta \in \mathbb{C}$.
\end{lemma}

\begin{proof}
Under the isomorphism $\mathsf{Lin}(V,W) \cong V^* \otimes W$, a vector $X = \sum_i v_i^* \otimes w_i$ corresponds to the linear map $\mathcal{N}(\rho) = \sum_i v_i^*(\rho) w_i$. Applying this to the expression for $X_{\mathcal{B}}$ in Corollary~\ref{Cor:invariant_vector} yields
\begin{align}
\mathcal{B}(\rho) 
&= \alpha \, \omega^*(\rho) \, \omega \otimes \omega \nonumber \\
&\quad + \beta \, \omega^*(\rho) \, (u \otimes u + v \otimes v) \nonumber \\
&\quad + \gamma \, (u^*(\rho) u + v^*(\rho) v) \otimes \omega \nonumber \\
&\quad + \delta \, \omega \otimes (u^*(\rho) u + v^*(\rho) v).
\end{align}
Expanding $\rho = \operatorname{Tr}(\rho) \, \omega + u^*(\rho) u + v^*(\rho) v$, with $\omega^* = \operatorname{Tr}$, and using $u^*(\rho) u + v^*(\rho) v = \rho - \operatorname{Tr}(\rho) \, \omega$, we obtain
\begin{align}
\mathcal{B}(\rho) 
&= \tilde{\alpha} \, \operatorname{Tr}(\rho) \, \omega \otimes \omega
+ \beta \, \operatorname{Tr}(\rho) \, (u \otimes u + v \otimes v) \nonumber \\
&\quad + \gamma \, \rho \otimes \omega
+ \delta \, \omega \otimes \rho,
\end{align}
where $\tilde{\alpha} := \alpha - \gamma - \delta$.
\end{proof}

\subsubsection*{Step 6: Imposing the classical consistency}

By Proposition~\ref{Prop:classical_consistency_equivalence_class}, it suffices to verify classical consistency for a single representative of each equivalence class. We begin with the adjacent orbit.

\medskip

\begin{lemma}
Classical consistency on the adjacent orbit uniquely fixes all coefficients in the covariant map, yielding
\begin{align}
\tilde{\alpha} = -1, \qquad \beta = \frac18, \qquad \gamma = 1, \qquad \delta = 1.
\end{align}
\end{lemma}

\begin{proof}
Choose the representative $\mathcal{C}_{\mathrm{adj}} = \{\rho_{00}, \rho_{10}\}$, distinguished by the effect $e_x$. The corresponding dephasing channel and classical broadcasting map are
\begin{align}
\mathcal{D}_{\mathrm{adj}}(\rho) &= e_x(\rho)\rho_{10} + \bar{e}_x(\rho)\rho_{00}, \\
\mathcal{B}_{\mathrm{cl}}^{\mathrm{adj}}(\rho) &= e_x(\rho)\rho_{10}\otimes\rho_{10} + \bar{e}_x(\rho)\rho_{00}\otimes\rho_{00}.
\end{align}

The most general covariant map from Lemma~\ref{Lem:covariant_map} is
\begin{align}
\mathcal{B}(\rho) 
&= \tilde{\alpha} \, \operatorname{Tr}(\rho) \, \omega \otimes \omega \nonumber \\
&\quad + \beta \, \operatorname{Tr}(\rho) \, (u \otimes u + v \otimes v) \nonumber \\
&\quad + \gamma \, \rho \otimes \omega \nonumber \\
&\quad + \delta \, \omega \otimes \rho.
\end{align}

Since $\mathcal{D}_{\mathrm{adj}}$ fixes the basis states $\rho_{00}$ and $\rho_{10}$, it is sufficient to impose classical consistency on these two states.

We first evaluate on $\rho_{00}$. Using $\operatorname{Tr}(\rho_{00})=1$, we have
\begin{align}
\mathcal{B}(\rho_{00}) 
&= \tilde{\alpha} \, \omega \otimes \omega 
+ \beta \, (u \otimes u + v \otimes v) 
+ \gamma \, \rho_{00} \otimes \omega 
+ \delta \, \omega \otimes \rho_{00}.
\end{align}

We now compute the action of $\mathcal{D}_{\mathrm{adj}}$ on the basis elements $\{\omega, v, u\}$:
\begin{align}
\mathcal{D}_{\mathrm{adj}}(\omega) &= \frac12(\rho_{00}+\rho_{10}), \\
\mathcal{D}_{\mathrm{adj}}(v) &= \mathcal{D}_{\mathrm{adj}}(u) = \rho_{00}-\rho_{10}.
\end{align}

Applying $\mathcal{D}_{\mathrm{adj}} \otimes \mathcal{D}_{\mathrm{adj}}$ to $\mathcal{B}(\rho_{00})$ and expanding in the basis $\{\rho_{00}, \rho_{10}\}$ gives
\begin{align}
&(\mathcal{D}_{\mathrm{adj}} \otimes \mathcal{D}_{\mathrm{adj}})\mathcal{B}(\rho_{00}) \nonumber \\
&= \left(\frac{\tilde{\alpha}}{4} + 2\beta + \frac{\gamma + \delta}{2}\right) \rho_{00}\otimes\rho_{00} \nonumber \\
&\quad + \left(\frac{\tilde{\alpha}}{4} - 2\beta + \frac{\gamma}{2}\right) \rho_{00}\otimes\rho_{10} \nonumber \\
&\quad + \left(\frac{\tilde{\alpha}}{4} - 2\beta + \frac{\delta}{2}\right) \rho_{10}\otimes\rho_{00} \nonumber \\
&\quad + \left(\frac{\tilde{\alpha}}{4} + 2\beta\right) \rho_{10}\otimes\rho_{10}.
\end{align}

Classical consistency demands
\begin{align}
(\mathcal{D}_{\mathrm{adj}} \otimes \mathcal{D}_{\mathrm{adj}})\mathcal{B}(\rho_{00}) = \rho_{00}\otimes\rho_{00}.
\end{align}
Since $\{\rho_{00}\otimes\rho_{00}, \rho_{00}\otimes\rho_{10}, \rho_{10}\otimes\rho_{00}, \rho_{10}\otimes\rho_{10}\}$ are linearly independent, we may compare coefficients. This yields
\begin{align}
\frac{\tilde{\alpha}}{4} + 2\beta + \frac{\gamma + \delta}{2} &= 1,  \\
\frac{\tilde{\alpha}}{4} - 2\beta + \frac{\gamma}{2} &= 0,  \\
\frac{\tilde{\alpha}}{4} - 2\beta + \frac{\delta}{2} &= 0,  \\
\frac{\tilde{\alpha}}{4} + 2\beta &= 0. 
\end{align}

Solving these equations, we find that classical consistency on the adjacent orbit uniquely fixes all coefficients:
\begin{align}
\tilde{\alpha} = -1, \qquad \beta = \frac18, \qquad \gamma = 1, \qquad \delta = 1.
\end{align}

An identical calculation for $\rho_{10}$ with the condition 
\begin{align}
(\mathcal{D}_{\mathrm{diag}} \otimes \mathcal{D}_{\mathrm{diag}})\mathcal{B}(\rho_{10}) = \rho_{10}\otimes\rho_{10}
\end{align}
yields the identical result.
\end{proof}

We now explore the diagonal orbit and show that classical consistency on this orbit yields the same system of equations.

\begin{lemma}
Classical consistency on the diagonal orbit uniquely fixes all coefficients in the covariant map, yielding
\begin{align}
\tilde{\alpha} = -1, \qquad \beta = \frac18, \qquad \gamma = 1, \qquad \delta = 1.
\end{align}
\end{lemma}

\begin{proof}
Choose the representative $\mathcal{C}_{\mathrm{diag}} = \{\rho_{01}, \rho_{10}\}$, distinguished by the effect $e_x$. The corresponding dephasing channel and classical broadcasting map are
\begin{align}
\mathcal{D}_{\mathrm{diag}}(\rho) &= e_x(\rho)\rho_{10} + \bar{e}_x(\rho)\rho_{01}, \\
\mathcal{B}_{\mathrm{cl}}^{\mathrm{diag}}(\rho) &= e_x(\rho)\rho_{10}\otimes\rho_{10} + \bar{e}_x(\rho)\rho_{01}\otimes\rho_{01}.
\end{align}

As before, it is sufficient to impose classical consistency on the basis states $\rho_{01}$ and $\rho_{10}$. We first compute the action of $\mathcal{D}_{\mathrm{diag}}$ on the basis elements $\{\omega, v, u\}$:
\begin{align}
\mathcal{D}_{\mathrm{diag}}(\omega) &= \frac12(\mathcal{D}_{\mathrm{diag}}(\rho_{01}) + \mathcal{D}_{\mathrm{diag}}(\rho_{10}))
= \frac12(\rho_{01}+\rho_{10}) = \omega, \\
\mathcal{D}_{\mathrm{diag}}(v) &= \mathcal{D}_{\mathrm{diag}}(\rho_{01}) - \mathcal{D}_{\mathrm{diag}}(\rho_{10}) = \rho_{01}-\rho_{10} = v, \\
\mathcal{D}_{\mathrm{diag}}(u) &= \mathcal{D}_{\mathrm{diag}}(\rho_{00}) - \mathcal{D}_{\mathrm{diag}}(\rho_{11}) = \rho_{01}-\rho_{10} = v.
\end{align}
Thus $\mathcal{D}_{\mathrm{diag}}(u) = \mathcal{D}_{\mathrm{diag}}(v) = v$.

Applying $\mathcal{D}_{\mathrm{diag}} \otimes \mathcal{D}_{\mathrm{diag}}$ to $\mathcal{B}(\rho_{10})$, using $\operatorname{Tr}(\rho_{10})=1$, and expanding in the basis $\{\rho_{01}, \rho_{10}\}$ gives
\begin{align}
&(\mathcal{D}_{\mathrm{diag}} \otimes \mathcal{D}_{\mathrm{diag}})\mathcal{B}(\rho_{10}) \nonumber \\
&= \left(\frac{\tilde{\alpha}}{4} + 2\beta + \frac{\gamma + \delta}{2}\right) \rho_{10}\otimes\rho_{10} \nonumber \\
&\quad + \left(\frac{\tilde{\alpha}}{4} - 2\beta + \frac{\delta}{2}\right) \rho_{01}\otimes\rho_{10} \nonumber \\
&\quad + \left(\frac{\tilde{\alpha}}{4} - 2\beta + \frac{\gamma}{2}\right) \rho_{10}\otimes\rho_{01} \nonumber \\
&\quad + \left(\frac{\tilde{\alpha}}{4} + 2\beta\right) \rho_{01}\otimes\rho_{01}.
\end{align}

Imposing classical consistency:
\begin{align}
(\mathcal{D}_{\mathrm{diag}} \otimes \mathcal{D}_{\mathrm{diag}})\mathcal{B}(\rho_{01}) = \rho_{01}\otimes\rho_{01}.
\end{align}
and comparing coefficients gives the same system of equations and the same solutions as in the adjacent case.

\end{proof}

Consequently, the broadcasting map satisfying Axiom~\ref{axm:single_system_broadcasting} is uniquely determined by
\begin{align}
\mathcal{B}(\rho) 
= \rho \otimes \omega + \omega \otimes \rho - \operatorname{Tr}(\rho) \, \omega \otimes \omega
+ \frac18 \operatorname{Tr}(\rho) \, (u \otimes u + v \otimes v)
\end{align}
for all $\rho \in \mathsf{St}_{\mathbb{C}}(A)$.
This proves Theorem~\ref{thm:unique-broadcasting-map}.

\subsection{Bayesian retrodiction map in the boxworld}

\subsubsection{Explicit solution for the Bayesian retrodiction map}
\label{app:retrodiction_solution}

Here, we provide the explicit solution for the retrodiction map $\mathcal{R}_{(\mathcal{E},\rho)}$ as in Eq.~\ref{Eq:Retrodiction_definition}. Let us define
\begin{align}
\Delta
:=
8\bigl(u^*[\mathcal{E}(\rho)]\bigr)^2
+
8\bigl(v^*[\mathcal{E}(\rho)]\bigr)^2
-1. \label{Eq:Delta_retro}
\end{align}

The action of the retrodiction map on the basis elements $\{\omega,v,u\}$ of $\operatorname{St}_{\mathbb{C}}(A)$ is as follows.

\begin{widetext}
\begin{align}
\mathcal{R}_{(\mathcal{E},\rho)}(\omega)
={}&
\frac{1}{\Delta}
\Bigg\{
\omega
\Big[
8\,u^*[\mathcal{E}(\rho)]\,
u^*[\mathcal{E}(\rho-\omega)]
+
8\,v^*[\mathcal{E}(\rho)]\,
v^*[\mathcal{E}(\rho-\omega)]
\Big] \nonumber \\
&\quad
+\rho
\Big[
8\,u^*[\mathcal{E}(\rho)]\,
u^*[\mathcal{E}(\omega)]
+
8\,v^*[\mathcal{E}(\rho)]\,
v^*[\mathcal{E}(\omega)]
-1
\Big] \nonumber \\
&\quad
+u
\Big[
u^*[\mathcal{E}(u)]\,u^*[\mathcal{E}(\rho)]
+
v^*[\mathcal{E}(u)]\,v^*[\mathcal{E}(\rho)]
\Big] \nonumber \\
&\quad
+v
\Big[
u^*[\mathcal{E}(v)]\,u^*[\mathcal{E}(\rho)]
+
v^*[\mathcal{E}(v)]\,v^*[\mathcal{E}(\rho)]
\Big]
\Bigg\}, \label{Eq:R_omega_sol}
\end{align}

Using the expression for $\mathcal{R}_{(\E,\rho)}(\omega)$ already obtained, we can rearrange Eqs.~\eqref{eq:R-v} and \eqref{eq:R-u} to write the remaining components directly:

\begin{align}
\mathcal{R}_{(\mathcal{E},\rho)}(v)
&=
8\,v^*[\mathcal{E}(\omega)]\,(\rho-\omega)
+v^*[\mathcal{E}(v)]\,v
+v^*[\mathcal{E}(u)]\,u
-8\,v^*[\mathcal{E}(\rho)]\,\left(\mathcal{R}_{(\mathcal{E},\rho)}(\omega) - \ \omega\right), \label{Eq:R_v_sol}
\\[4pt]
\mathcal{R}_{(\mathcal{E},\rho)}(u)
&=
8\,u^*[\mathcal{E}(\omega)]\,(\rho-\omega)
+u^*[\mathcal{E}(v)]\,v
+u^*[\mathcal{E}(u)]\,u
-8\,u^*[\mathcal{E}(\rho)]\,\left(\mathcal{R}_{(\mathcal{E},\rho)}(\omega) - \ \omega\right). \label{Eq:R_u_sol}
\end{align}

\end{widetext}

Since $\{\omega,v,u\}$ is a basis of $\operatorname{St}_{\mathbb{C}}(A)$, the retrodiction map is completely determined by these three expressions.

\subsubsection{Retrodiction for Constant Channels in Boxworld}\label{App:constant_chan_retro}

Consider the constant channel that outputs a fixed state $\sigma \in \mathsf{St}(B)$ regardless of the input:
\begin{align}
\mathcal{E}(\rho) = \operatorname{Tr}(\rho) \, \sigma, \qquad \forall \rho \in \mathsf{St}(A).
\end{align}
Expressing the output state in the basis $\{\omega, u, v\}$, we have
\begin{align}
\sigma = \omega + r_u u + r_v v,
\end{align}
where $r_u = u^*(\sigma)$ and $r_v = v^*(\sigma)$. For this channel, the denominator $\Delta$, as in Eq.~\eqref{Eq:Delta_retro} becomes
\begin{align}
\Delta := 8\bigl(u^*[\mathcal{E}(\rho)]\bigr)^2 + 8\bigl(v^*[\mathcal{E}(\rho)]\bigr)^2 - 1
= 8(r_u^2 + r_v^2) - 1.
\end{align}

We now evaluate the retrodiction map on the basis elements $\{\omega, v, u\}$.

\medskip

\noindent\textbf{Action on $\omega$:} 
Since $\mathcal{E}(\rho) = \operatorname{Tr}(\rho) \, \sigma$, we have $\mathcal{E}(\rho - \omega) = \mathcal{E}(u) = \mathcal{E}(v) = 0$. Thus all coefficients of $\omega$, $u$, and $v$ in the general expression for $\mathcal{R}_{(\mathcal{E},\rho)}(\omega)$ in Eq.~\eqref{Eq:R_omega_sol} vanish. On the other hand, $u^*[\mathcal{E}(\rho)] = u^*[\mathcal{E}(\omega)] = r_u$ and $v^*[\mathcal{E}(\rho)] = v^*[\mathcal{E}(\omega)] = r_v$, so the coefficient of $\rho$ is $8(r_u^2 + r_v^2) - 1 = \Delta$. Hence
\begin{align}
\mathcal{R}_{(\mathcal{E},\rho)}(\omega)
=
\frac{1}{\Delta} \, (\rho \, \Delta)
=
\rho.
\end{align}

\medskip

\noindent\textbf{Action on $v$:} 
Using Eq.~\eqref{Eq:R_v_sol} together with $\mathcal{R}_{(\mathcal{E},\rho)}(\omega) = \rho$, we obtain
\begin{align}
\mathcal{R}_{(\mathcal{E},\rho)}(v)
&=
8 r_v (\rho - \omega) - 8 r_v (\rho - \omega)
=
0.
\end{align}

\medskip

\noindent\textbf{Action on $u$:} 
Similarly, from Eq.~\eqref{Eq:R_u_sol},
\begin{align}
\mathcal{R}_{(\mathcal{E},\rho)}(u)
&=
8 r_u (\rho - \omega) - 8 r_u (\rho - \omega)
=
0.
\end{align}

\medskip

Thus for any $\tau \in \mathsf{St}_{\mathbb{C}}(B)$, with
\begin{align}
\tau = \operatorname{Tr}(\tau)\, \omega + \tau_v\, v + \tau_u\, u,
\end{align}
where $\tau_v = v^*(\tau)$ and $\tau_u = u^*(\tau)$, we have,
\begin{align}
\mathcal{R}_{(\mathcal{E},\rho)}(\tau)
&=
\operatorname{Tr}(\tau)\, \mathcal{R}_{(\mathcal{E},\rho)}(\omega)
+ \tau_v\, \mathcal{R}_{(\mathcal{E},\rho)}(v)
+ \tau_u\, \mathcal{R}_{(\mathcal{E},\rho)}(u)
\nonumber\\
&=
\operatorname{Tr}(\tau) \, \rho.
\end{align}

Thus, for a constant channel, the retrodiction map simply discards the input and prepares the prior state $\rho$.

\subsubsection{Retrodiction for Reversible Transformations in Boxworld}
\label{app:reversible_retrodiction}

Let $\mathcal{E} = \mathcal{U}$, where $\mathcal{U}$ is a reversible transformation of the square bit. We show explicitly that the corresponding retrodiction map is
\begin{align}
\mathcal{R}_{(\mathcal{U}, \rho)} = \mathcal{U}^{-1}.
\end{align}

Since the reversible transformations of the square bit form the group $D_4$, and the representation decomposes as
\begin{align}
\mathsf{St}_{\mathbb{C}}(A)
=
\operatorname{Span}_{\mathbb{C}}\{\omega\}
\oplus
\operatorname{Span}_{\mathbb{C}}\{u,v\},
\end{align}
every reversible transformation $\mathcal{U}$ acts trivially on the first summand and orthogonally on the second. Thus
\begin{align}
\mathcal{U}(\omega)=\omega,
\end{align}
and, writing the restriction of $\mathcal{U}$ to $E=\operatorname{Span}_{\mathbb{C}}\{u,v\}$ as $\mathcal{U}'$, we have
\begin{align}
\mathcal{U} = 1 \oplus \mathcal{U}'.
\end{align}

The generators $R$ and $M$ act on the basis $\{u,v\}$ as
\begin{align}
R(u)=-v, \quad R(v)=u, \qquad
M(u)=-u, \quad M(v)=v.
\end{align}
Therefore their matrix representations on $\operatorname{Span}_{\mathbb{C}}\{u,v\}$ are
\begin{align}
R'
=
\begin{pmatrix}
0 & -1\\
1 & 0
\end{pmatrix}, \qquad
M'
=
\begin{pmatrix}
-1 & 0\\
0 & 1
\end{pmatrix}.
\end{align}
Both matrices are orthogonal, and hence
\begin{align}
(R')^{-1}=(R')^{T}, \qquad (M')^{-1}=(M')^{T}.
\end{align}
Since $R$ and $M$ generate $D_4$, the same property holds for every $\mathcal{U} \in D_4$:
\begin{align}
(\mathcal{U}')^{-1}=(\mathcal{U}')^{T}.
\end{align}
Consequently,
\begin{align}
\mathcal{U}^{-1}
=
1 \oplus (\mathcal{U}')^{-1}
=
1 \oplus (\mathcal{U}')^{T}.
\end{align}

Now let us consider the matrix representation 

\begin{align}
\mathcal{U}'
=
\begin{pmatrix}
U_{uu} & U_{uv}\\
U_{vu} & U_{vv}
\end{pmatrix},
\end{align}
so that
\begin{align}
\mathcal{U}(u)=U_{uu}u+U_{uv}v, \qquad
\mathcal{U}(v)=U_{vu}u+U_{vv}v. \label{Eq:U(u)_and_U(v)}
\end{align}

Furthermore, since $\mathcal{U}'.(\mathcal{U}')^T  = \Id$, we have
\begin{align}
U_{uu}^{2}+U_{uv}^{2}=1, \quad
U_{vu}^{2}+U_{vv}^{2}=1, \quad
U_{uu}U_{vu}+U_{uv}U_{vv}=0.\label{Eq:U'_ortho}
\end{align}
Now let us consider an arbitrary state $\rho \in \St(A)$: 

\begin{align}
\rho = \omega + r_u u + r_v v.
\end{align}
Then we have
\begin{align}
\mathcal{U}(\rho) &= \omega + r_u \mathcal{U}(u) + r_v \mathcal{U}(v), \nonumber \\
&=\omega + (r_v U_{vu} + r_u U_{uu})u + (r_v U_{vv} + r_u U_{uv})v\label{Eq:U_rho} 
\end{align}

We now evaluate $\mathcal{R}_{\mathcal{U},\rho}(\omega)$ using the general expression given in Eq.~\eqref{Eq:R_omega_sol}. Because $\mathcal{U}(\omega) = \omega$, we have
\begin{align}
u^*[\mathcal{U}(\omega)] = v^*[\mathcal{U}(\omega)] = 0.
\end{align}
Furthermore, we show that the action of $\mathcal{U}$ on $\mathrm{Span}_{\mathbb{C}}(u,v)$ results in the following,
\begin{align}
u^*[\mathcal{U}(u)]u^*[\mathcal{U}(\rho)] + v^*[\mathcal{U}(u)]v^*[\mathcal{U}(\rho)] &= u^*(\rho), \\
u^*[\mathcal{U}(v)]u^*[\mathcal{U}(\rho)] + v^*[\mathcal{U}(v)]v^*[\mathcal{U}(\rho)] &= v^*(\rho). \label{Eq:v*(rho)}
\end{align}

To see this, using Eqs.~\eqref{Eq:U(u)_and_U(v)} and \eqref{Eq:U_rho}, we have 
\begin{align}
  &u^*[\mathcal{U}(u)]u^*[\mathcal{U}(\rho)] + v^*[\mathcal{U}(u)]v^*[\mathcal{U}(\rho)] \nonumber \\
  &= U_{uu}(r_v U_{vu} + r_u U_{uu}) + U_{uv}( r_vU_{vv} + r_u U_{uv} ) \nonumber \\
  &=r_v(U_{uu}U_{vv}+ U_{uv}U_{vv}) + r_u(U_{uu}^2 + U_{uv}^2)\nonumber \\
  &= r_u = u^*(\rho).
\end{align}
Here, in the third equality, we have used Eq.~\eqref{Eq:U'_ortho}. Eq.~\eqref{Eq:v*(rho)} can be proved in the exact same method. 

Now, note that
\begin{align}
\Delta
&=
8\bigl(u^*[\mathcal{U}(\rho)]\bigr)^2 + 8\bigl(v^*[\mathcal{U}(\rho)]\bigr)^2 - 1
=
8(r_u^2+r_v^2)-1,
\end{align}
with these steps, the expression for $\mathcal{R}_{(\mathcal{U},\rho)}(\omega)$ simplifies to
\begin{align}
\mathcal{R}_{(\mathcal{U},\rho)}(\omega)
=
\frac{1}{\Delta}
\Big[
8(r_u^2+r_v^2)\omega
-\rho
+r_u u+r_v v
\Big].
\end{align}
But $\rho = \omega + r_u u + r_v v$, and therefore
\begin{align}
8(r_u^2+r_v^2)\omega - \rho + r_u u + r_v v
=
\bigl(8(r_u^2+r_v^2)-1\bigr)\omega
=
\Delta\,\omega.
\end{align}
Hence
\begin{align}
\mathcal{R}_{(\mathcal{U}, \rho)}(\omega)=\omega.
\end{align}

We next determine the action on the nontrivial subspace $\mathrm{Span}_{\mathbb{C}}(u,v)$. Replacing $\mathcal{U}(\omega)=\omega$ and $\mathcal{R}_{(\mathcal{U}, \rho)}(\omega)=\omega$ in Eqs.~\eqref{Eq:R_v_sol} and ~\eqref{Eq:R_u_sol} we obtain

\begin{align}
\mathcal{R}_{(\mathcal{U},\rho)}(v)
&=
v^*[\mathcal{U}(v)]\,v
+v^*[\mathcal{U}(u)]\,u, \nonumber \\
&=U_{vv}v + U_{uv}u
\\[4pt]
\mathcal{R}_{(\mathcal{U},\rho)}(u)
&=u^*[\mathcal{U}(v)]\,v
+u^*[\mathcal{U}(u)]\,u \nonumber \\
&=U_{vu}v+U_{uu}u
\end{align}
 
Comparing these equations with $\mathcal{U}^{-1}=1\oplus (\mathcal{U'})^T$, where 
\begin{align}
\mathcal{U}^{-1}(u) = U_{uu}u+U_{vu}v, \qquad
\mathcal{U}^{-1}(v) = U_{uv}u+U_{vv}v,
\end{align}
we have
\begin{align}
\mathcal{R}_{(\mathcal{U}, \rho)}(v)  = \mathcal{U}^{-1}(v), \\
\mathcal{R}_{(\mathcal{U}, \rho)}(u) = \mathcal{U}^{-1}(u).
\end{align}

Together with the result for $\omega$, this completely determines the retrodiction map:
\begin{align}
\mathcal{R}_{(\mathcal{U}, \rho)}=\mathcal{U}^{-1}.
\end{align}

\subsection{Proof of Theorem~\ref{Thm:two_party_broadcasting}} \label{App:proof_two_party_broadcasting}
We structure the proof as follows. First, we establish the group-theoretic structure of the global transformations for two systems. We then use the invariant vectors from the single-system case to characterize the two-system covariant maps. Finally, we impose the broadcasting and classical consistency axioms to obtain the unique solution.
\subsubsection*{Step 1: Group-theoretic structure of global reversible transformations for two square-bits}

In boxworld with the maximal tensor product, the global reversible transformations of composite square-bit systems consist only of local reversible transformations and permutations of the subsystems~\cite{Gross2010}. For two systems $A$ and $E$, the local transformations act on the tensor product state space $\mathsf{St}_{\mathbb{C}}(A) \otimes \mathsf{St}_{\mathbb{C}}(E)$ via the representation
\begin{align}
(g_1, g_2) \mapsto \pi(g_1) \otimes \pi(g_2), \qquad g_1, g_2 \in D_4,
\end{align}
where $\pi$ is the standard representation of $D_4$ on a single square-bit system. Thus the local transformations are generated by the direct product group $D_4 \times D_4$. Together with the two-element permutation group $S_2$, the full global reversible transformation group is generated by $D_4 \times D_4$ and $S_2$. In terms of group product, this means
\begin{align}
\mathsf{U}_{\mathrm{global}} &= (D_4 \times D_4) S_2 \nonumber \\
&:= \{ (g_1, g_2) \cdot \sigma \mid (g_1, g_2) \in D_4 \times D_4,\; \sigma \in S_2 \}.
\end{align}

More precisely, we can write this group product as a semidirect product, which is defined as 

\begin{definition}\label{def:semidirectProduct}
  Let $G$ be a group, $N$ a normal subgroup of $G$, and $H$ a subgroup of $G$. Then $G$ is called the \emph{semidirect product} of $N$ and $H$, denoted
\begin{align}
G = N \rtimes H,
\end{align}
if the following holds:
\begin{align}
G = NH:= \{ nh \mid n \in N,\; h \in H \}, \qquad N \cap H = \{e\}.
\end{align}
\end{definition} 

We will make use of the following elementary fact about internal products of groups.

\begin{corollary}\label{cor:unique_decomp_product}
Let $G$ be a group with subgroups $N$ and $H$ such that $G = NH$ and $N \cap H = \{e\}$. Then every element $g \in G$ can be written uniquely as $g = nh$ for some $n \in N$ and $h \in H$.
\end{corollary}

\begin{proof}
Existence follows immediately from $G = NH$. For uniqueness, suppose $g = n_1 h_1 = n_2 h_2$ with $n_1, n_2 \in N$ and $h_1, h_2 \in H$. Then $n_2^{-1} n_1 = h_2 h_1^{-1}$. The left-hand side is in $N$, and the right-hand side is in $H$, so both lie in $N \cap H = \{e\}$. Hence $n_2^{-1} n_1= h_2 h_1^{-1} = e$, so $n_1 = n_2$ and $h_1 = h_2$.
\end{proof}

With this definition, we can identify the group-theoretic structure of the full global transformation group for two square-bit systems.

\begin{corollary}
The full global transformation group for two square-bit systems is the semidirect product
\begin{align}
\mathsf{U}_{\mathrm{global}} = (D_4 \times D_4) \rtimes S_2.
\end{align}
\end{corollary}

\begin{proof}
First, note that $D_4 \times D_4$ is a normal subgroup of $\mathsf{U}_{\mathrm{global}}$. To see this, let $g = (h_1, h_2) \cdot \sigma \in \mathsf{U}_{\mathrm{global}}$ be arbitrary, and let $(g_1, g_2) \in D_4 \times D_4$. Then
\begin{align}
g (g_1,g_2) g^{-1} 
&= (h_1, h_2) \sigma (g_1, g_2) \sigma^{-1} (h_1^{-1}, h_2^{-1}) \\
&= (h_1, h_2) (g_{\sigma(1)}, g_{\sigma(2)}) (h_1^{-1}, h_2^{-1}) \\
&= (h_1 g_{\sigma(1)} h_1^{-1}, \; h_2 g_{\sigma(2)} h_2^{-1}) \in D_4 \times D_4.
\end{align}
Thus $D_4 \times D_4$ is the normal subgroup of $\mathsf{U}_{\mathrm{global}}$. 

Second, $(D_4 \times D_4) \cap S_2 = \{e\}$, since the only permutation that acts trivially on both subsystems is the identity.

Hence, by Definition~\ref{def:semidirectProduct},
\begin{align}
\mathsf{U}_{\mathrm{global}} = (D_4 \times D_4) \rtimes S_2.
\end{align}
\end{proof}

We now characterize the global covariant map 
$\mathcal{B} \in \mathsf{Transf}_{\mathbb{C}}(A \otimes E \to A'' \otimes A' \otimes E'' \otimes E')$ 
satisfying, for all $\mathcal{U} \in \mathsf{U}_{\mathrm{global}}$,
\begin{align}
\mathcal{B} \circ \mathcal{U} = (\mathcal{U} \otimes \mathcal{U}) \circ \mathcal{B}.
\end{align}
Following the same reasoning as in Lemma~\ref{lem:map_to_invariant_vector}, we associate to $\mathcal{B}$ an invariant vector 
$X_{\mathcal{B}} \in \mathsf{St}_{\mathbb{C}}(A)^* \otimes \mathsf{St}_{\mathbb{C}}(A'') \otimes \mathsf{St}_{\mathbb{C}}(A') \otimes \mathsf{St}_{\mathbb{C}}(E)^* \otimes \mathsf{St}_{\mathbb{C}}(E'') \otimes \mathsf{St}_{\mathbb{C}}(E')$ 
such that
\begin{align}
((\mathcal{U}^{-1})^* \otimes \mathcal{U} \otimes \mathcal{U})(X_{\mathcal{B}}) = X_{\mathcal{B}},
\end{align}
for all $\mathcal{U} \in \mathsf{U}_{\mathrm{global}}$.
Due to the semidirect product structure $\mathsf{U}_{\mathrm{global}} = (D_4 \times D_4) \rtimes S_2$, the characterization of invariant vectors can be broken down into two steps: first, we characterize the fixed-point subspace under the local transformations $D_4 \times D_4$; then, we impose invariance under the remaining $S_2$ action.

Before proceeding, we first justify this method by showing that the projector onto the fixed-point subspace of $\mathsf{U}_{\mathrm{global}}$ decomposes into the projectors onto the fixed-point subspaces of $D_4 \times D_4$ and $S_2$. We show this in general group-theoretic terms.

First, note that for any finite group $G$ acting on a vector space $V$, the projector onto the fixed-point subspace, $V^G:=\{v \in V \mid \pi(g)v = v, \forall g \in G\}$, is given by the twirling operation
\begin{align}
P_G := \frac{1}{|G|} \sum_{g \in G} \pi(g).
\end{align}
Here, $|G|$ is the number of elements in the group $G$. Now note that, due to the product formula, for any group $G = NH$,
\begin{align}
|G| = \frac{|N| \cdot |H|}{|N \cap H|}.
\end{align}
In particular, when $N \cap H = \{e\}$, we have $|G| = |N| \cdot |H|$. With this and due to Corollary~\ref{cor:unique_decomp_product}, which states that any group element $g \in G$ can be uniquely written as $g = nh$ for $n \in N$ and $h \in H$, we show that the projector onto the fixed-point subspace $V^G$ factors as follows.

\begin{lemma} \label{lem:twirling_decomp}
Let $G$ be a finite group with subgroups $N$ and $H$ such that $G = NH$ and $N \cap H = \{e\}$. Let $V$ be a finite-dimensional representation of $G$, and let
\begin{align}
&P_G := \frac{1}{|G|} \sum_{g \in G} \pi(g), \qquad
P_N := \frac{1}{|N|} \sum_{n \in N} \pi(n), \\
&P_H := \frac{1}{|H|} \sum_{h \in H} \pi(h)
\end{align}
be the twirling projectors onto the fixed-point subspaces $V^G$, $V^N$, and $V^H$, respectively. Then
\begin{align}
P_G = P_N \circ P_H.
\end{align}
Furthermore, if $N$ is a normal subgroup of $G$, then the projectors commute and
\begin{align}
P_G = P_N \circ P_H = P_H \circ P_N.
\end{align}
\end{lemma}

\begin{proof}
  Since $G = NH$ and $N \cap H = \{e\}$, every element $g \in G$ has a unique decomposition $g = nh$ with $n \in N$ and $h \in H$. With this, using the product formula, $|G| = |N| \cdot |H|$, and the homomorphism property $\pi(nh) = \pi(n)\pi(h)$,
\begin{align}
P_G 
&= \frac{1}{|G|} \sum_{g \in G} \pi(g) \\
&= \frac{1}{|N||H|} \sum_{n \in N} \sum_{h \in H} \pi(nh) \\
&= \frac{1}{|N||H|} \sum_{n \in N} \sum_{h \in H} \pi(n) \pi(h) \\
&= \left( \frac{1}{|N|} \sum_{n \in N} \pi(n) \right) \circ \left( \frac{1}{|H|} \sum_{h \in H} \pi(h) \right) \\
&= P_N \circ P_H.
\end{align}
Furthermore, if $N$ is normal, then for any $n \in N$ and $h \in H$, we have $nh = h n'$ for some $n' \in N$. Replacing $\pi(nh) = \pi(h n') = \pi(h)\pi(n')$ in the above calculation yields $P_G = P_H \circ P_N$. Thus
\begin{align}
P_G = P_N \circ P_H = P_H \circ P_N.
\end{align}

\end{proof}

Since $\mathsf{U}_{\mathrm{global}}$, and more generally any semidirect product, satisfies the conditions of Lemma~\ref{lem:twirling_decomp}, we can first characterize the fixed-vector subspace under the normal subgroup $D_4 \times D_4$, and then impose invariance under the remaining $S_2$ action.

\subsubsection*{Step 2: Invariant vectors under local reversible transformations}

We first establish a general group-theoretic fact that will be important in our case. Let $(\pi_1, V_1)$ and  $(\pi_2, V_2)$ be representations of groups $G_1$ and  $G_2$ respectively. Let us denote the fixed-point subspace, i.e., the subspace spanned by the invariant vectors, by
\begin{align}
&V_1^{G_1} := \{ v \in V_1 \mid \pi_1(g_1)v = v \text{ for all } g_1 \in G_1 \}, \ \,\text{and} \\
&V_2^{G_2} := \{ w \in V_2 \mid \pi_2(g_2)w = w \text{ for all } g_2 \in G_2 \}.
\end{align}
With this notation, we have the following lemma.

\begin{lemma} \label{Lem:invariant_vector_space_product_group}
For finite-dimensional representations $(\pi_1, V_1)$ of $G_1$ and $(\pi_2, V_2)$ of $G_2$, the fixed-point subspace of the tensor product representation $V_1 \otimes V_2$ of $G_1 \times G_2$ satisfies
\begin{align}
(V_1 \otimes V_2)^{G_1 \times G_2} = V_1^{G_1} \otimes V_2^{G_2},
\end{align}
where the action of $G_1 \times G_2$ on $V_1 \otimes V_2$ is given by $(g_1, g_2) \mapsto \pi_1(g_1) \otimes \pi_2(g_2)$.
\end{lemma}

\begin{proof}
($\supseteq$): Let $v \in V_1^{G_1}$ and $w \in V_2^{G_2}$. Then for any $(g_1, g_2) \in G_1 \times G_2$,
\begin{align*}
(\pi_1(g_1) \otimes \pi_2(g_2))(v \otimes w) 
&= \pi_1(g_1)v \otimes \pi_2(g_2)w \\
&= v \otimes w.
\end{align*}
Thus $v \otimes w \in (V_1 \otimes V_2)^{G_1 \times G_2}$, therefore $(V_1 \otimes V_2)^{G_1 \times G_2} \supseteq V_1^{G_1} \otimes V_2^{G_2}$.

\medskip

\noindent ($\subseteq$): Let $X \in (V_1 \otimes V_2)^{G_1 \times G_2}$ be an arbitrary invariant vector. Choose bases $\{v_i\}_{i=1}^{\dim V_1}$ of $V_1$ and $\{w_j\}_{j=1}^{\dim V_2}$ of $V_2$. Then $X$ has a unique expansion
\begin{align}
X = \sum_{i,j} c_{ij} \, v_i \otimes w_j
\end{align}
for some coefficients $c_{ij} \in \mathbb{C}$. Since $X$ is invariant under the full group $G_1 \times G_2$, it is in particular invariant under the subgroup $G_1 \times \{e_2\}$, where $e_2$ is the identity element of $G_2$. Thus, for any $g_1 \in G_1$,
\begin{align*}
&(\pi_1(g_1) \otimes \id_{V_2})(X) = X,\\ 
&\implies \sum_{i,j} c_{ij} \, \pi_1(g_1)v_i \otimes w_j 
= \sum_{i,j} c_{ij} \, v_i \otimes w_j.
\end{align*}
Since $\{w_j\}$ is a basis of $V_2$, the coefficients of each $w_j$ must match. Hence, for each fixed $j$,
\begin{align}
\sum_i c_{ij} \, \pi_1(g_1)v_i = \sum_i c_{ij} \, v_i \quad \forall g_1 \in G_1.
\end{align}
For each $j$, define
$x_j := \sum_i c_{ij} v_i \in V_1.$ Then according to the above equation,
$x_j \in V_1^{G_1} $ for each $j$. Consequently, we can rewrite $X$ as
\begin{align}
X = \sum_j x_j \otimes w_j, \qquad x_j \in V_1^{G_1}.
\end{align}

Now let $\{a_k\}_{k=1}^{\dim V_1^{G_1}}$ be a basis of the invariant subspace $V_1^{G_1}$. Since each $x_j \in V_1^{G_1}$, there exist coefficients $\alpha_{kj} \in \mathbb{C}$ such that
\begin{align}
x_j = \sum_k \alpha_{kj} a_k.
\end{align}
Substituting this into the expression for $X$, we obtain
\begin{align}
X = \sum_{k,j} \alpha_{kj} \, a_k \otimes w_j 
= \sum_k a_k \otimes \left( \sum_j \alpha_{kj} w_j \right).
\end{align}
For each $k$, define
$z_k := \sum_j \alpha_{kj} w_j \in V_2$. Then
$X = \sum_k a_k \otimes z_k$. Now, since $X$ is invariant under the full group, it is in particular invariant under the subgroup $\{e_1\} \times G_2$. Thus, for any $g_2 \in G_2$,
\begin{align*}
&(\id_{V_1} \otimes \pi_2(g_2))(X) = X, \nonumber \\
&\implies \sum_k a_k \otimes \pi_2(g_2) z_k = \sum_k a_k \otimes z_k.
\end{align*}
Since $\{a_k\}$ is a basis of $V_1^{G_1}$, the coefficients of each $a_k$ must match:
\begin{align}
\pi_2(g_2) z_k = z_k \quad \forall g_2 \in G_2,
\end{align}
for each $k$. Therefore,
$z_k \in V_2^{G_2}$.
Hence, we have expressed $X$ as
\begin{align}
X = \sum_k a_k \otimes z_k,
\end{align}
where $a_k \in V_1^{G_1}$ and $z_k \in V_2^{G_2}$. This shows that $X \in V_1^{G_1} \otimes V_2^{G_2}$. Therefore,$(V_1 \otimes V_2)^{G_1 \times G_2} \subseteq V_1^{G_1} \otimes V_2^{G_2}$. Combining both inclusions proves the lemma.
\end{proof}

We can now characterize the invariant vectors under local reversible transformations $D_4 \times D_4$. As an immediate consequence of Lemma~\ref{Lem:invariant_vector_space_product_group}, we note that the fixed-point subspace of $D_4 \times D_4$ is simply the tensor product of the fixed-point subspaces of each individual $D_4$ factor already characterized in Lemma~\ref{Lem:invariant_Vector}, leading to the following:

\begin{lemma}\label{lem:D4xD4_Invariance}
    The $D_4 \times D_4$-invariant subspace of 
$\mathsf{St}_{\mathbb{C}}(A)^* \otimes \mathsf{St}_{\mathbb{C}}(A'') \otimes \mathsf{St}_{\mathbb{C}}(A')\otimes \St_{\mathbb{C}}(E)^*\otimes \St_{\mathbb{C}}(E'') \otimes \St_{\mathbb{C}}(E') $
is spanned by vectors of the form
\begin{align}
X_{\mathcal{B}} = \sum_{ i,j \in \{1,2,3,4\}} a_{ij} \, L_i \otimes L_j,
\end{align}
where
\begin{align}
L_1 &= (u^* \otimes u + v^* \otimes v) \otimes \omega,\\  L_2 &= u^* \otimes \omega \otimes u + v^* \otimes \omega \otimes v,
 \\
L_3 &= \omega^* \otimes \omega \otimes \omega, \\
L_4 &= \omega^* \otimes (u \otimes u + v \otimes v).
\end{align}
are the $D_4$-invariant vectors characterized in Lemma~\ref{Lem:invariant_Vector}.
\end{lemma}

\subsubsection*{Step 3: Invariant vectors under the SWAP operation}

Having characterized the invariant vectors under $D_4 \times D_4$ in Lemma~\ref{lem:D4xD4_Invariance}, we now impose the remaining covariance condition associated with the $S_2$ swap operation.
 
 For notational convenience, let $\mathcal{S}_{AE}$ denote the swap transformation $\mathcal{U}_{\mathrm{SWAP}}\in \Chan(A\otimes E \to E\otimes A)$, that exchanges the systems $A$ and $E$. With this, we have the following.

\begin{lemma}\label{SWAP_Invariant_X_B}
Let $X_{\mathcal{B}} = \sum_{i,j} a_{ij} L_i \otimes L_j$ be a $D_4 \times D_4$-invariant vector as characterized in Lemma~\ref{lem:D4xD4_Invariance}. The swap-invariance condition $\widetilde{\mathcal{S}} \circ  X_{\mathcal{B}} = X_{\mathcal{B}}$, where $\widetilde{\mathcal{S}} := \mathcal{S}_{A^*E^*}^* \otimes \mathcal{S}_{A''E''} \otimes \mathcal{S}_{A'E'}$, forces the coefficient matrix $(a_{ij})$ to be symmetric, i.e., $a_{ij} = a_{ji}$ for all $i,j$.
\end{lemma}

\begin{proof}
 Suppose the vector $X_{\mathcal{B}} \in \mathsf{St}_{\mathbb{C}}(A)^* \otimes \mathsf{St}_{\mathbb{C}}(A'') \otimes \mathsf{St}_{\mathbb{C}}(A') \otimes \mathsf{St}_{\mathbb{C}}(E)^* \otimes \mathsf{St}_{\mathbb{C}}(E'') \otimes \mathsf{St}_{\mathbb{C}}(E')$ satisfies the swap-invariance condition
\begin{align}
\widetilde{\mathcal{S}} \circ X_{\mathcal{B}} = X_{\mathcal{B}},
\end{align}
where
\begin{align}
\widetilde{\mathcal{S}} := \mathcal{S}_{A^*E^*}^* \otimes \mathcal{S}_{A''E''} \otimes \mathcal{S}_{A'E'}.
\end{align}
Here, we have assumed the necessary rearrangements of the systems in $\widetilde{\mathcal{S}}$ to match the order of the systems in $X_{\mathcal{B}}$. The map $\mathcal{S}_{A^*E^*}^*$ is the dual map of the swap on the dual spaces. Since the dual of the swap is again the swap, i.e., $\mathcal{S}_{A^*E^*}^* = \mathcal{S}_{A^*E^*}$, we simply have $\mathcal{S}_{A^*E^*}^*$ swapping the dual systems $A^*$ and $E^*$.

Recall from Lemma~\ref{lem:D4xD4_Invariance} that the $D_4 \times D_4$-invariant vectors are of the form
\begin{align}
X_{\mathcal{B}} = \sum_{i,j} a_{ij} \, L_i \otimes L_j,
\end{align}
where $L_i$ lives on the $A$-systems ($A^* \otimes A'' \otimes A'$) and $L_j$ lives on the $E$-systems ($E^* \otimes E'' \otimes E'$). Under the swap transformation, the entire $A$-system is exchanged with the $E$-system, so
\begin{align}
\widetilde{\mathcal{S}} \circ  (L_i \otimes L_j) = L_j \otimes L_i.
\end{align}

Applying $\widetilde{\mathcal{S}}$ to $X_{\mathcal{B}}$ gives
\begin{align}
\widetilde{\mathcal{S}} \circ  X_{\mathcal{B}}
= \sum_{i,j} a_{ij} \, L_j \otimes L_i
= \sum_{i,j} a_{ji} \, L_i \otimes L_j.
\end{align}

The swap-invariance condition $\widetilde{\mathcal{S}} \circ X_{\mathcal{B}} = X_{\mathcal{B}}$ then yields
\begin{align}
\sum_{i,j} a_{ji} \, L_i \otimes L_j = \sum_{i,j} a_{ij} \, L_i \otimes L_j.
\end{align}

Since the vectors $L_i \otimes L_j$ are linearly independent, we may compare coefficients to obtain
\begin{align}
a_{ji} = a_{ij} \quad \text{for all } i,j.
\end{align}

Thus the coefficient matrix $(a_{ij})$ is symmetric.
\end{proof}

\subsubsection*{Step 4: Converting invariant vectors to linear maps}

Converting these invariant vectors under $(D_4 \times D_4) \rtimes S_2$ back to linear maps via the same procedure of rearranging the coefficients as in Lemma~\ref{Lem:covariant_map}, we immediately obtain the following.

\begin{lemma}
The global $(D_4 \times D_4) \rtimes S_2$-covariant map 
$\mathcal{B} \in \mathsf{Transf}_{\mathbb{C}}(A \otimes E \to A'' \otimes A' \otimes E'' \otimes E')$ 
satisfies, for all $\rho \in \mathsf{St}_{\mathbb{C}}(A)$ and $\sigma \in \mathsf{St}_{\mathbb{C}}(E)$,
\begin{align}
\mathcal{B}(\rho \otimes \sigma) = \sum_{i,j \in \{1,2,3,4\}} \tilde{a}_{ij} \, \mathcal{N}_i(\rho) \otimes \mathcal{N}_j(\sigma),
\end{align}
where $\tilde{a}_{ij} = \tilde{a}_{ji} \in \mathbb{C}$, and for any system $C \in \{A, E\}$, the linear maps 
$\mathcal{N}_i \in \mathsf{Transf}_{\mathbb{C}}(C \to C'' \otimes C')$ 
are given, for all $\tau \in \mathsf{St}_{\mathbb{C}}(C)$, by
\begin{align}
&\mathcal{N}_1(\tau) = \tau \otimes \omega, \\
&\mathcal{N}_2(\tau) = \omega \otimes \tau, \\
&\mathcal{N}_3(\tau) = \operatorname{Tr}(\tau) \, \omega \otimes \omega, \\
&\mathcal{N}_4(\tau) = \operatorname{Tr}(\tau) \, (u \otimes u + v \otimes v).
\end{align}
\end{lemma}

\begin{proof}
By the same procedure as in the proof of Lemma~\ref{Lem:covariant_map}, the invariant vectors $L_i$ in Lemma~\ref{lem:D4xD4_Invariance} correspond to linear maps $\mathcal{N}_i$ as follows:
\begin{align}
L_3 &\longleftrightarrow \mathcal{N}_3, &
L_4 &\longleftrightarrow \mathcal{N}_4, \\
L_1 &\longleftrightarrow \mathcal{N}_1 - \mathcal{N}_3, &
L_2 &\longleftrightarrow \mathcal{N}_2 - \mathcal{N}_3.
\end{align}

or equivalently, $\mathbf{L} \longleftrightarrow M \boldsymbol{\mathcal{N}}$, where $\mathbf{L}=(L_1\  L_2\  L_3\  L_4)^T$, $\boldsymbol{\mathcal{N}}=(\N_1\  \N_2\  \N_3\ \N_4)^T$, and
\begin{align}
M := 
\begin{pmatrix}
1 & 0 & -1 & 0 \\
0 & 1 & -1 & 0 \\
0 & 0 & 1 & 0 \\
0 & 0 & 0 & 1
\end{pmatrix}.
\end{align}

Using this isomorphism, we obtain the transformed map as
\begin{align}
\mathcal{B}(\rho \otimes \sigma)
&= \sum_{i,j} a_{ij} \left( \sum_p M_{ip} \mathcal{N}_p(\rho) \right) \otimes \left( \sum_q M_{jq} \mathcal{N}_q(\sigma) \right) \\
&= \sum_{p,q} \left( \sum_{i,j} M_{ip} a_{ij} M_{jq} \right) \mathcal{N}_p(\rho) \otimes \mathcal{N}_q(\sigma).
\end{align}

Thus the reparameterized coefficients are given by
\begin{align}
\tilde{a}_{pq} = \sum_{i,j} M_{ip} a_{ij} M_{jq},
\end{align}
which in matrix form is $\tilde{A} = M^T A M$, where $A = (a_{ij})$ and $\tilde{A} = (\tilde{a}_{ij})$. Since Lemma~\ref{SWAP_Invariant_X_B} gives $a_{ij} = a_{ji}$, we have $A = A^T$. Therefore,
\begin{align}
\tilde{A}^T = (M^T A M)^T  = M^T A^T M = M^T A M = \tilde{A}.
\end{align}

Hence $\tilde{a}_{ij} = \tilde{a}_{ji}$ for all $i,j$. 

\end{proof}

\subsubsection*{Step 5: Imposing the remaining Axioms}

Subsequently, we evaluate the 10 independent coefficients $\tilde{a}_{ij}$ (with $i \leq j$) by imposing the classical consistency axiom constraints as in Axiom~\ref{axm:joint_broadcasting}. This yields a system of linear equations for the coefficients. We extract and solve these equations using standard symbolic computation (see Ref.~\cite{Goswami_GlobalBroadcastingCode_2026} for the code). The resulting solution uniquely determines the broadcasting map.

\medskip

\noindent \textbf{Classical consistency:} Here, we impose the classical consistency for the adjacent-adjacent case, which yields the following equations:

\begin{align}
\frac{\tilde{a}_{11}}{4} + \frac{\tilde{a}_{12}}{2} + \frac{\tilde{a}_{13}}{4} + 2\tilde{a}_{14} + \frac{\tilde{a}_{22}}{4} + \frac{\tilde{a}_{23}}{4} & \nonumber \\
+ 2\tilde{a}_{24} + \frac{\tilde{a}_{33}}{16} + \tilde{a}_{34} + 4\tilde{a}_{44} &= 1,  \\[8pt]
\frac{\tilde{a}_{11}}{2} + \frac{3\tilde{a}_{12}}{4} + \frac{\tilde{a}_{13}}{2} + 2\tilde{a}_{14} + \frac{\tilde{a}_{22}}{4}& \nonumber \\
+ \frac{3\tilde{a}_{23}}{8} + \tilde{a}_{24} + \frac{\tilde{a}_{33}}{8} + \tilde{a}_{34} &= 1,  \\[8pt]
\frac{\tilde{a}_{11}}{4} + \frac{3\tilde{a}_{12}}{4} + \frac{3\tilde{a}_{13}}{8} + \tilde{a}_{14} + \frac{\tilde{a}_{22}}{2}& \nonumber \\
+ \frac{\tilde{a}_{23}}{2} + 2\tilde{a}_{24} + \frac{\tilde{a}_{33}}{8} + \tilde{a}_{34} &= 1,  \\[8pt]
\frac{\tilde{a}_{11}}{2} + \tilde{a}_{12} + \frac{3\tilde{a}_{13}}{4} + 2\tilde{a}_{14} + \frac{\tilde{a}_{22}}{2}& \nonumber \\+ \frac{3\tilde{a}_{23}}{4} + 2\tilde{a}_{24} + \frac{\tilde{a}_{33}}{4} + 2\tilde{a}_{34} &= 1,  \\[8pt]
\tilde{a}_{11} + \tilde{a}_{12} + \tilde{a}_{13} + \frac{\tilde{a}_{22}}{4} + \frac{\tilde{a}_{23}}{2} + \frac{\tilde{a}_{33}}{4} &= 1,  \\[8pt]
\frac{\tilde{a}_{11}}{2} + \frac{5\tilde{a}_{12}}{4} + \frac{3\tilde{a}_{13}}{4} + \frac{\tilde{a}_{22}}{2} + \frac{3\tilde{a}_{23}}{4} + \frac{\tilde{a}_{33}}{4} &= 1,  \\[8pt]
\tilde{a}_{11} + \frac{3\tilde{a}_{12}}{2} + \frac{3\tilde{a}_{13}}{2} + \frac{\tilde{a}_{22}}{2} + \tilde{a}_{23} + \frac{\tilde{a}_{33}}{2} &= 1,  \\[8pt]
\frac{\tilde{a}_{11}}{4} + \tilde{a}_{12} + \frac{\tilde{a}_{13}}{2} + \tilde{a}_{22} + \tilde{a}_{23} + \frac{\tilde{a}_{33}}{4} &= 1,  \\[8pt]
\frac{\tilde{a}_{11}}{2} + \frac{3\tilde{a}_{12}}{2} + \tilde{a}_{13} + \tilde{a}_{22} + \frac{3\tilde{a}_{23}}{2} + \frac{\tilde{a}_{33}}{2} &= 1,  \\[8pt]
\tilde{a}_{11} + 2\tilde{a}_{12} + 2\tilde{a}_{13} + \tilde{a}_{22} + 2\tilde{a}_{23} + \tilde{a}_{33} &= 1. 
\end{align}

Solving this system of equations yields the unique solution for the adjacent-adjacent case:
\begin{align}
\tilde{a}_{11} &= 1, &
\tilde{a}_{12} &= 1, &
\tilde{a}_{13} &= -1, &
\tilde{a}_{14} &= \frac{1}{8}, \nonumber \\
\tilde{a}_{22} &= 1, &
\tilde{a}_{23} &= -1, &
\tilde{a}_{24} &= \frac{1}{8}, &
\tilde{a}_{33} &= 1, \nonumber \\
\tilde{a}_{34} &= -\frac{1}{8}, &
\tilde{a}_{44} &= \frac{1}{64}. \label{Eq:solutions}
\end{align}

Imposing the remaining classical consistency conditions---adjacent-diagonal, diagonal-adjacent, and diagonal-diagonal---yields the identical solution, confirming that no further constraints arise from these cases.

\subsubsection*{Step 6: Factorization and final form}

The solutions in Eq.~\eqref{Eq:solutions} imply the coefficients $\tilde{a}_{ij}$ factorize as $\tilde{a}_{ij} = b_i b_j$, with
\begin{align}
b_1 = 1, \qquad b_2 = 1, \qquad b_3 = -1, \qquad b_4 = \frac{1}{8}.
\end{align}
Consequently, the global broadcasting map decomposes as the parallel composition of two identical single-system maps:
\begin{align}
\mathcal{B}(\rho \otimes \sigma)
&= \left( \sum_{i=1}^4 b_i \mathcal{N}_i(\rho) \right)
\otimes
\left( \sum_{j=1}^4 b_j \mathcal{N}_j(\sigma) \right) \nonumber \\
&= \mathcal{B}_{A}(\rho) \otimes \mathcal{B}_{E}(\sigma),
\end{align}
where for $C\in\{A,E\}$,  $\mathcal{B}_{C}\in \mathsf{Transf}_{\mathbb{C}}(C\to C'' \otimes C')$ is the unique single-system broadcasting map characterized in Theorem~\ref{thm:unique-broadcasting-map}. This proves Theorem~\ref{Thm:two_party_broadcasting}.

\section{State-linearity of multi-time states does not imply reconstructibility} \label{app:uniqeMTQS}

In Section~\ref{sec:MarkovCaopatible}, we have observed that the
reconstructibility axiom in Axiom~\ref{axm:PropMTQS} implies that the star
product is convex-linear in states. The converse, however, fails: the
multi-time state in Fig.~\ref{fig:ReconVsnonReconQSOT}(a), introduced in Ref.~\cite{xiao2025no}, provides a
counterexample in the quantum theory.

In Construction-I in Fig.~\ref{fig:ReconVsnonReconQSOT}(a), the multi-time quantum state is obtained by first applying the $1{\to}N$ virtual broadcasting map $\mathcal{B}_N$, followed by sequential applications of quantum channels:
\begin{align}
\tau_{A_N \cdots A_2A_1}= \big[ (\mathcal{E}_{N-1} \circ \cdots &\circ \mathcal{E}_{1}) \otimes \cdots \otimes (\mathcal{E}_{2} \circ \mathcal{E}_{1}) \nonumber \\
&\otimes \mathcal{E}_{1} \otimes \id  \big] \circ \mathcal{B}_N(\rho), \label{Eq:1_to_NBroadcastingMTQS}
\end{align}

Here the virtual broadcasting map $\mathcal{B}_N$ is defined as
\begin{align}
\mathcal{B}_N(\rho)
&:= \frac{1}{2} \{\rho_N, U_N\},
\label{eq:BNdef}
\end{align}
where
\begin{align}
\rho_N
&:= \frac{1}{N} \sum_{k=1}^N \rho^{(k)}_{N},
\quad
\rho^{(k)}_{N}
:=
\Id^{\otimes(k-1)}
\otimes
\rho
\otimes
\Id^{\otimes(N-k)},
\label{eq:rhoN}\\[2mm]
U_N
&:= \frac{1}{(N-1)!} \sum_{\sigma \in C_N} U_\sigma,
\quad
\abs{C_N}
:= (N-1)!.
\label{eq:UN}
\end{align}
Here $C_N \subset S_N$ denotes the set of all $N$-cycle permutations, and
$U_\sigma$ is the unitary representation of the permutation $\sigma$. For
$N=2$, $\mathcal{B}_N$ reduces to the virtual broadcasting map in
Eq.~\eqref{eq:VirtualBC} of the main text.

Note that this construction is linear in the initial states but nonlinear in
the channels. Interestingly, it nevertheless satisfies the compositionality axiom. To see
this, we first note the following structural property of the broadcasting
map $\mathcal{B}_N$ (see Appendix~\ref{app:BnB} for the proof). For
notational convenience, we write $\Tr_n$ for $\Tr_{A_n}$ throughout.

\begin{lemma}\label{lem:BnB}
Let $\mathcal{B}_N$ be the virtual broadcasting map defined in
Eq.~\eqref{eq:BNdef}. Then, for every output subsystem
$n \in \{1,\dots,N\}$,
\begin{align}\label{eq:BnB}
\Tr_n \circ \mathcal{B}_N = \mathcal{B}_{N-1}.
\end{align}
\end{lemma}

Lemma~\ref{lem:BnB} immediately implies that Construction-I satisfies
compositionality. Indeed, tracing out subsystem $A_n$ and using that each
$\mathcal{E}_k$ is trace preserving, we obtain
\begin{align}
\Tr_n &\Big[
\{\mathcal{E}_{N-1},\dots,\mathcal{E}_{n},
\mathcal{E}_{n-1},\dots,\mathcal{E}_{1}\} \star \rho
\Big]
\nonumber\\[2mm]
&=
\Big[
(\mathcal{E}_{N-1} \circ \cdots \circ \mathcal{E}_1)
\otimes \cdots
\otimes \Tr_n(\mathcal{E}_{n-1} \circ \cdots \circ \mathcal{E}_1)
\nonumber\\
&\qquad
\otimes \cdots
\otimes \mathcal{E}_1
\otimes \id
\Big]
\circ \mathcal{B}_N(\rho)
\nonumber\\[2mm]
&=
\Big[
(\mathcal{E}_{N-1} \circ \cdots \circ \mathcal{E}_1)
\otimes \cdots
\otimes (\mathcal{E}_{n} \circ \mathcal{E}_{n-1} \circ \cdots \circ \mathcal{E}_1)
\nonumber\\
&\qquad
\otimes (\mathcal{E}_{n-1} \circ \cdots \circ \mathcal{E}_1)
\otimes \cdots
\otimes \mathcal{E}_1
\otimes \id
\Big]
\circ \Tr_n \circ \mathcal{B}_N(\rho)
\nonumber\\[2mm]
&=
\Big[
(\mathcal{E}_{N-1} \circ \cdots \circ \mathcal{E}_1)
\otimes \cdots
\otimes (\mathcal{E}_{n} \circ \mathcal{E}_{n-1} \circ \cdots \circ \mathcal{E}_1)
\nonumber\\
&\qquad
\otimes (\mathcal{E}_{n-1} \circ \cdots \circ \mathcal{E}_1)
\otimes \cdots
\otimes \mathcal{E}_1
\otimes \id
\Big]
\circ \mathcal{B}_{N-1}(\rho)
\nonumber\\[2mm]
&=
\{\mathcal{E}_{N-1},\dots,(\mathcal{E}_{n} \circ \mathcal{E}_{n-1}),
\dots,\mathcal{E}_{1}\} \star \rho.
\end{align}

The observations that the star product in Construction-I is state-linear but channel
nonlinear, and that it satisfies compositionality, together rule out
reconstructibility. Indeed, if Construction-I were
reconstructible and satisfied compositionality, then by
Theorem~\ref{Thm:StructureThm_MarkovQSOT} it would admit a chain-like
structure as in Fig.~\ref{fig:ReconVsnonReconQSOT}(b), for some sequence of
broadcasting maps. Such a
chain-like structure, however, would force the star product to be linear in
the channels, contradicting the channel-nonlinearity observed above.

\smallskip

Reconstructibility can also be ruled out without
invoking compositionality. For completeness, we give such an argument below.
The following lemma also provides a systematic method for checking whether a
given star product is reconstructible.

\begin{lemma}\label{lem:Condition_Broadcasting}
Let $\mathcal{B}_N : A_1 \to A_N \otimes \cdots \otimes A_1$ be the
$1 \to N$ broadcasting map, and let $\mathcal{B}_n : A_1 \to A_n \otimes
\cdots \otimes A_1$ be the $1 \to n$ broadcasting map. Suppose the
multi-time state satisfies the reconstructibility axiom. Then, for every cut
$A_N,\dots,A_{n+1} \mid A_n,\dots,A_1$, every initial state $\rho_{A_1}$,
and every effect $e_{A_{n-1}\cdots A_1} \in \eff(A_{n-1} \otimes \cdots
\otimes A_1)$, such that
\begin{align}
(\id_{A_n} \otimes e_{A_{n-1}\cdots A_1})\mathcal{B}_n(\rho_{A_1})=0,
\label{eq:Condition_Broadcasting_left}
\end{align}
then we have
\begin{align}
(\id_{A_N\otimes\cdots\otimes A_n} \otimes e_{A_{n-1}\cdots A_1})
\mathcal{B}_N(\rho_{A_1})=0.
\label{eq:Condition_Broadcasting_right}
\end{align}
\end{lemma}

\begin{proof}
Suppose there exists a state $\rho \in \St(A_1)$ and an effect $e \in \eff(A_{n-1} \otimes \cdots \otimes
A_1)$ such that
\begin{align}
(\Id_{A_n} \otimes e_{A_{n-1}\cdots A_1})\mathcal{B}_n(\rho_{A_1})=0.
\end{align}
If the multi-time state satisfies the reconstructibility axiom, then there
exists a map $\mathcal{R}_n : A_n \to A_N \otimes \cdots \otimes A_n$ that
takes $\mathcal{B}_n$ to $\mathcal{B}_N$. Here we use the fact that a
multi-time state corresponds to a broadcasting map for identity channels. Hence,
\begin{align}
(\mathcal{R}_n \otimes \id_{A_{n-1}\otimes\cdots\otimes A_1}) \circ
(\id_{A_n} \otimes e_{A_{n-1}\cdots A_1})\mathcal{B}_n(\rho_{A_1})&=0,
\\
(\id_{A_N\otimes\cdots\otimes A_n} \otimes e_{A_{n-1}\cdots A_1})\circ
(\mathcal{R}_n \otimes \id_{A_{n-1}\otimes\cdots\otimes A_1})
\mathcal{B}_n(\rho_{A_1})&{=}0,
\\
(\id_{A_N\otimes\cdots\otimes A_n} \otimes e_{A_{n-1}\cdots A_1})
\mathcal{B}_N(\rho_{A_1})&{=}0.
\end{align}
\end{proof}

Note that, following a similar approach to
Ref.~\cite{chen2025virtual}, one can prove a converse: if for every
ancillary system $R$, every pair of state $\rho_{A_1 R}$, and effects
$e_{A_{n-1}\cdots A_1 R}$ satisfying
$(\id_{A_n} \otimes e_{A_{n-1}\cdots A_1 R})(\mathcal{B}_n \otimes \id_R)
(\rho_{A_1 R})=0$ implies 
$(\id_{A_N \otimes \cdots \otimes A_n} \otimes e_{A_{n-1}\cdots A_1 R})
(\mathcal{B}_N \otimes \id_R)(\rho_{A_1 R})=0$, then a unique reconstructing
map exists. We do not, however, pursue this direction here.

Lemma~\ref{lem:Condition_Broadcasting} implies that if the multi-time state
satisfies the reconstructibility axiom, then the corresponding broadcasting
map must satisfy the condition in Lemma~\ref{lem:Condition_Broadcasting}. Hence, by contrapositive, to prove that Construction-I does not satisfy reconstructibility,
it suffices to find one state and one effect for which the condition fails. It is enough to take
$n=2$, $N=3$, and $\mathcal{E}_1$ and $\mathcal{E}_2$ to be the identity
maps. With this choice, the $N=3$ multi-time state is simply the
broadcasting map $\mathcal{B}_3$, while the $n=2$ multi-time state is
simply the broadcasting map $\mathcal{B}_2$.

Now consider the effect $e_{A_1}=\proj{1}_{A_1}$ and the state
$\rho_{A_1}=\proj{0}_{A_1}$. Then by explicit calculation, we obtain
\begin{align}
&\Tr_{A_1}\!\left[(\ketbra{1}{1}_{A_1} \otimes \Id_{A_2})
\mathcal{B}_2(\proj{0})\right] = 0_{A_2}, \ \mathrm{but}
\\
&\Tr_{A_1}\!\left[(\ketbra{1}{1}_{A_1} \otimes \Id_{A_2A_3})
\mathcal{B}_3(\proj{0})\right] \nonumber \\
&\qquad =
\frac{1}{6}\left(\ketbra{0}{1}_{A_2} \otimes \ketbra{1}{0}_{A_3}
+ \ketbra{1}{0}_{A_2} \otimes \ketbra{0}{1}_{A_3}\right)
\neq 0_{A_2A_3}.
\end{align}
Hence, the necessary condition in Lemma~\ref{lem:Condition_Broadcasting} is
violated, and no reconstructing map  $\mathcal{R}_{A_2 \to A_3A_2}$ for for Construction-I exists.

In contrast, Construction-II in Fig.~\ref{fig:ReconVsnonReconQSOT}(b)
respects the structure of Theorem~\ref{Thm:StructureThm_MarkovQSOT} and
therefore satisfies both compositionality and reconstructibility, with
reconstructing map $\mathcal{R}_n
=
\mathcal{E}_{N-1} \circ \mathcal{B}_2 \circ \cdots \circ \mathcal{E}_{n+1}
\circ \mathcal{B}_2 \circ \mathcal{E}_n \circ \mathcal{B}_2.$

\begin{figure}
\centering
\includegraphics[width=\columnwidth]{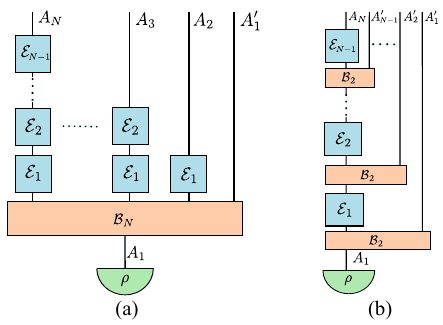}
\caption{ Two different multi-time states $\{\mathcal{E}_n\}_{n=1}^N\star \rho$ (Definition~\ref{def:MTQS}), that are convex-linear in states and satisfy `Compositionality'~\cite{xiao2025no}. However, Fig.(a) does not satisfy the `Reconstructibility' axiom, while Fig.(b) satisfies it. In Fig.(a), $\mathcal{B}_N$ represents the $1\to N$ virtual broadcasting map~\cite{xiao2025no}, see Eq.~\eqref{Eq:1_to_NBroadcastingMTQS}. In Fig.(b), $\mathcal{B}_2$ is the $1\to 2$ virtual broadcasting map ~\cite{parzygnat2024virtual,xiao2025no}, see Eq.~\eqref{eq:VirtualBC} in the main text.  }
    \label{fig:ReconVsnonReconQSOT}
\end{figure}

\section{Proof of Lemma~\ref{lem:BnB}}\label{app:BnB}

The proof hinges on two lemmas. Before introducing them, we fix the notations.

\medskip

\noindent \textbf{Notations:}
We write an $N$-cycle in the form
\begin{align}
\sigma = (a_1, a_2, \ldots, a_{N-1}, N) \in C_N,
\label{eq:Ncycle}
\end{align}
where the labels $a_1,\ldots,a_{N-1}$ are distinct elements of
$\{1,\ldots,N-1\}$, and $N$ is the final subsystem. Deleting $N$ from
$\sigma$ yields the $(N-1)$-cycle
\begin{align}
\sigma' = (a_1, a_2, \ldots, a_{N-1}) \in C_{N-1}.
\label{eq:Nminus1cycle}
\end{align}

For $M \in \mathbb{N}$ and $k \in \{1,\ldots,M\}$, let
\begin{align}
\rho^{(k)}_{M}
&:=
\Id^{\otimes(k-1)}
\otimes
\rho
\otimes
\Id^{\otimes(M-k)}
\label{eq:rho-k-M}
\end{align}
denote the operator on $M$ subsystems obtained by placing $\rho$ on
subsystem $k$ and the identity on all the others. With this, we denote the
average
\begin{align}
\rho_M
&:=
\frac{1}{M}
\sum_{k=1}^M \rho^{(k)}_{M}.
\label{eq:rhoM}
\end{align}

For a permutation $\sigma \in S_N$, let $U_\sigma$ denote its unitary
representation on the $N$-fold tensor product space. Furthermore, let
\begin{align}
U_N
&:=
\frac{1}{(N-1)!}
\sum_{\sigma \in C_N} U_\sigma,
&
|C_N|
&=
(N-1)!,
\label{eq:UN}
\end{align}
where $C_N \subset S_N$ is the set of all $N$-cycles, so that $U_N$ is the
uniform average over all $N$-cycles.

With this notation in place, we now state the two lemmas on which the proof
rests.

\begin{lemma}[Partial trace of a single $N$-cycle]
\label{lem:single-cycle}
Let
\begin{align}
\sigma = (a_1, a_2, \ldots, a_{N-1}, N) \in C_N
\end{align}
and let
\begin{align}
\sigma' = (a_1, a_2, \ldots, a_{N-1}) \in C_{N-1}
\end{align}
be the cycle obtained by deleting $N$. Then
\begin{align}
\Tr_N\!\left[\rho_N U_\sigma\right]
=
\frac{1}{N}
\left[
(N-1)\rho_{N-1}
+
\rho^{(a_1)}_{N-1}
\right]U_{\sigma'}.
\label{eq:single-cycle-partial-trace}
\end{align}
\end{lemma}

\begin{proof}
We use the decomposition
\begin{align}
\sigma = (N, a_1)(a_1, a_2, \ldots, a_{N-1}),
\end{align}
which gives
\begin{align}
U_\sigma
=
\operatorname{SWAP}_{N,a_1}
\left(U_{\sigma'} \otimes \Id_N\right).
\end{align}

Recall that
\begin{align}
\rho_N
=
\frac{1}{N}
\sum_{k=1}^N \rho^{(k)}_{N}.
\end{align}
There are three types of terms in this sum.

For the term with $\rho$ acting on $a_1$, the standard swap identity gives
\begin{align}
\Tr_N
\left[
\rho^{(a_1)}_{N} U_\sigma
\right]
=
\rho^{(a_1)}_{N-1} U_{\sigma'}.
\end{align}

Similarly, for the term with $\rho$ acting on the traced subsystem $N$,
\begin{align}
\Tr_N
\left[
\rho^{(N)}_{N} U_\sigma
\right]
=
\rho^{(a_1)}_{N-1} U_{\sigma'}.
\end{align}

Finally, if $k \notin \{a_1, N\}$, the operator $\rho^{(k)}_{N}$ acts on a
subsystem untouched by the swap involving $N$, and therefore
\begin{align}
\Tr_N
\left[
\rho^{(k)}_{N} U_\sigma
\right]
=
\rho^{(k)}_{N-1} U_{\sigma'}.
\end{align}

Consequently,
\begin{align}
\Tr_N\!\left[\rho_N U_\sigma\right]
&=
\frac{1}{N}
\left[
2\rho^{(a_1)}_{N-1}
+
\sum_{k \neq a_1, N} \rho^{(k)}_{N-1}
\right]U_{\sigma'} \\
&=
\frac{1}{N}
\left[
(N-1)\rho_{N-1}
+
\rho^{(a_1)}_{N-1}
\right]U_{\sigma'},
\end{align}
where in the last equation we used
\begin{align}
\sum_{k \neq N} \rho^{(k)}_{N-1}
=
(N-1)\rho_{N-1}.
\end{align}
This proves Eq.~\eqref{eq:single-cycle-partial-trace}.
\end{proof}

\begin{lemma}[Averaging over all $N$-cycles]
\label{lem:averaged-partial-trace}
For every $\rho$,
\begin{align}
\Tr_N\!\left[\rho_N U_N\right]
=
\rho_{N-1} U_{N-1}.
\label{eq:averaged-partial-trace}
\end{align}
\end{lemma}

\begin{proof}
Fix an $(N-1)$-cycle
\begin{align}
\sigma' = (a_1, \ldots, a_{N-1}) \in C_{N-1}.
\end{align}

Deleting $N$ sends each $N$-cycle $\sigma \in C_N$ to a unique
$(N-1)$-cycle in $C_{N-1}$. Conversely, every $\sigma' \in C_{N-1}$ has
exactly $N-1$ preimages in $C_N$, obtained by inserting $N$ into one of the
$N-1$ slots of $\sigma'$. We denote this mapping as $\sigma \mapsto \sigma'$,
and write $\sum_{\sigma \mapsto \sigma'}$ for the sum over all $N-1$
preimages of $\sigma'$.

Using Lemma~\ref{lem:single-cycle} and summing over the preimages of
$\sigma'$, we obtain
\begin{align}
\sum_{\sigma \mapsto \sigma'}
\Tr_N\!\left[\rho_N U_\sigma\right]
&=
\frac{1}{N}
\sum_{k=1}^{N-1}
\left[
(N-1)\rho_{N-1}
+
\rho^{(a_k)}_{N-1}
\right]U_{\sigma'} \\
&=
\frac{1}{N}
\left[
(N-1)^2 \rho_{N-1}
+
\sum_{k=1}^{N-1} \rho^{(a_k)}_{N-1}
\right]U_{\sigma'} \\
&=
\frac{1}{N}
\left[
(N-1)^2 + (N-1)
\right]
\rho_{N-1} U_{\sigma'} \\
&=
(N-1)\rho_{N-1} U_{\sigma'}.
\end{align}

We now average over all $N$-cycles. Since the $N$-cycles in $C_N$ can be
grouped according to which $(N-1)$-cycle they reduce to, summing over all
$\sigma \in C_N$ is the same as summing over all $\sigma' \in C_{N-1}$ and,
for each $\sigma'$, over its $N-1$ preimages:
$\sum_{\sigma \in C_N}
=
\sum_{\sigma' \in C_{N-1}}
\sum_{\sigma \mapsto \sigma'}$. Therefore
\begin{align}
\Tr_N[\rho_N U_N]
&=
\frac{1}{(N-1)!}
\sum_{\sigma \in C_N}
\Tr_N\!\left[\rho_N U_\sigma\right]
\nonumber\\[2mm]
&=
\frac{1}{(N-1)!}
\sum_{\sigma' \in C_{N-1}}
\sum_{\sigma \mapsto \sigma'}
\Tr_N\!\left[\rho_N U_\sigma\right]
\nonumber\\[2mm]
&=
\frac{1}{(N-1)!}
\sum_{\sigma' \in C_{N-1}}
(N-1)\rho_{N-1} U_{\sigma'}
\nonumber\\[2mm]
&=
\frac{1}{(N-2)!}
\rho_{N-1}
\sum_{\sigma' \in C_{N-1}} U_{\sigma'}.
\end{align}
Since
\begin{align}
U_{N-1}
=
\frac{1}{(N-2)!}
\sum_{\sigma' \in C_{N-1}} U_{\sigma'},
\end{align}
we conclude that
\begin{align}
\Tr_N[\rho_N U_N]
=
\rho_{N-1} U_{N-1}.
\end{align}
\end{proof}

With this, we can now show that tracing out the final subsystem of
$\mathcal{B}_N$ results in the virtual broadcasting map $\mathcal{B}_{N-1}$.

\begin{corollary}[Tracing out the final subsystem]
\label{cor:trace-final}
For every $\rho$,
\begin{align}
\Tr_N \circ \mathcal{B}_N(\rho)
&=
\frac{1}{2}
\Tr_N\!\left[\{\rho_N, U_N\}\right] \\
&=
\frac{1}{2}
\{\rho_{N-1}, U_{N-1}\} \\
&=
\mathcal{B}_{N-1}(\rho).
\label{eq:trace-N-result}
\end{align}
Equivalently,
\begin{align}
\Tr_N \circ \mathcal{B}_N
=
\mathcal{B}_{N-1}.
\end{align}
\end{corollary}

\begin{proof}
Taking the adjoint of Eq.~\eqref{eq:averaged-partial-trace} gives
\begin{align}
\Tr_N[U_N \rho_N]
=
U_{N-1} \rho_{N-1},
\end{align}
where we used the Hermiticity of $\rho_N$ and $U_N$.

Therefore,
\begin{align}
\Tr_N\!\left[\{\rho_N, U_N\}\right]
=
\{\rho_{N-1}, U_{N-1}\}.
\end{align}
Using the definition
\begin{align}
\mathcal{B}_N(\rho)
=
\frac{1}{2}\{\rho_N, U_N\},
\end{align}
we immediately obtain
\begin{align}
\Tr_N \circ \mathcal{B}_N(\rho)
&=
\frac{1}{2}
\Tr_N\!\left[\{\rho_N, U_N\}\right] \\
&=
\frac{1}{2}
\{\rho_{N-1}, U_{N-1}\} \\
&=
\mathcal{B}_{N-1}(\rho).
\label{eq:trace-N-result-2}
\end{align}
\end{proof}

It remains to show that the same identity holds when an arbitrary output
subsystem $n \in \{1,\ldots,N\}$ is traced out.

\begin{corollary}[Tracing out an arbitrary output subsystem]
\label{lem:arbitrary-subsystem}
For every $n \in \{1,\ldots,N\}$,
\begin{align}
\Tr_n \circ \mathcal{B}_N
=
\mathcal{B}_{N-1}.
\label{eq:arbitrary-subsystem}
\end{align}
\end{corollary}

\begin{proof}
By construction, the virtual broadcasting map is permutation invariant on
its output systems. In particular, for the swap unitary
\begin{align}
\mathcal{U}_{\operatorname{SWAP}_{n,N}}(X)
:=
\operatorname{SWAP}_{n,N} \, X \, \operatorname{SWAP}_{n,N},
\end{align}
we have
\begin{align}
\mathcal{U}_{\operatorname{SWAP}_{n,N}}
\circ \mathcal{B}_N
=
\mathcal{B}_N.
\end{align}
With this, we have
\begin{align}
\Tr_n \circ \mathcal{B}_N(\rho)
&=
\Tr_N
\circ
\mathcal{U}_{\operatorname{SWAP}_{n,N}}
\circ
\mathcal{B}_N(\rho) \\
&=
\Tr_N \circ \mathcal{B}_N(\rho) \\
&=
\mathcal{B}_{N-1}(\rho),
\end{align}
where the final equality follows from Corollary~\ref{cor:trace-final}. Since
$\rho$ is arbitrary, this proves Eq.~\eqref{eq:arbitrary-subsystem} for every
$n \in \{1,\ldots,N\}$.
\end{proof}

\bibliography{ref}
\end{document}